\documentclass[conference]{IEEEtran}

\ifCLASSOPTIONcompsoc
  \usepackage[nocompress]{cite}
\else
  \usepackage{cite}
\fi

\usepackage{ifthen}

\newboolean{talk}
\setboolean{talk}{false}
\newboolean{paper}
\setboolean{paper}{false}

\newboolean{IEEEstyle}
\setboolean{IEEEstyle}{false}
\newboolean{lipicsstyle}
\setboolean{lipicsstyle}{false}
\newboolean{eptcsstyle}
\setboolean{eptcsstyle}{false}
\newboolean{sigplanstyle}
\setboolean{sigplanstyle}{false}
\newboolean{PACMPL}
\setboolean{PACMPL}{false}
\newboolean{acmartstyle}
\setboolean{acmartstyle}{false}

\newboolean{needstheorems}
\setboolean{needstheorems}{false}
\newboolean{withimages}
\setboolean{withimages}{false}
\newboolean{withproofs}
\setboolean{withproofs}{true}

\newboolean{french}
\setboolean{french}{false}

\setboolean{IEEEstyle}{true}
\ifthenelse{\boolean{talk}}{}{\usepackage[utf8]{inputenc}}
\ifthenelse{\boolean{PACMPL}}{}{
	\ifthenelse{\boolean{french}}{\usepackage[french]{babel}}{
	  \ifthenelse{\boolean{lipicsstyle}}{}{\usepackage[english]{babel}}}
	  }
\ifthenelse{\boolean{PACMPL}}{}{
	\usepackage{amssymb} }
\usepackage{amsmath}
\usepackage{graphicx}
\usepackage{bussproofs}
\usepackage{fancybox}
\usepackage{cmll}
\usepackage{stmaryrd}
\usepackage{multirow}
\ifthenelse{\boolean{IEEEstyle}}{
\usepackage[bookmarks={false}]{hyperref}}{
\usepackage{hyperref}}
\usepackage{proof}
\usepackage{hhline}
\usepackage{xspace}
\usepackage{booktabs}
\usepackage{subcaption}
\usepackage{wrapfig}
\usepackage{array}
\usepackage{arydshln}
\usepackage{commath}
\usepackage{complexity}
\ifthenelse{\boolean{IEEEstyle}}{
\setboolean{needstheorems}{true}}{}

\ifthenelse{\boolean{eptcsstyle}}{
\setboolean{needstheorems}{true}}{}

\ifthenelse{\boolean{sigplanstyle}}{
\setboolean{needstheorems}{true}}{}

\ifthenelse{\boolean{needstheorems}}{

\usepackage{amsthm}

    \newtheorem{theorem}{Theorem}[section]
    \newtheorem{lemma}[theorem]{Lemma}
    \newtheorem{corollary}[theorem]{Corollary}
    \newtheorem{proposition}[theorem]{Proposition}
    \newtheorem{definition}[theorem]{Definition}

}{}

\newcommand{\myproof}[1]{
\ifthenelse{\boolean{withproofs}}{#1}{}}

\newcommand{\withproofs}[1]{
\ifthenelse{\boolean{withproofs}}{#1}{}}

\newcommand{\withoutproofs}[1]{
\ifthenelse{\boolean{withproofs}}{}{#1}}

\newcommand{\tm}{t}
\newcommand{\tmtwo}{u}
\newcommand{\tmthree}{r}
\newcommand{\tmfour}{w}
\newcommand{\tmfive}{s}

\newcommand{\var}{x}
\newcommand{\vartwo}{y}
\newcommand{\varthree}{z}

\newcommand{\Rew}[1]{\rightarrow_{#1}}

\renewcommand{\to}{\Rew{}}

\newcommand{\towh}{\Rew{wh}}

\newcommand{\symfont}[1]{\mathsf{#1}}

\newcommand{\varsym}{{\symfont{var}}}

\newcommand{\ctxholep}[1]{[#1]}
\newcommand{\ctxhole}{\ctxholep{\cdot}}

\newcommand{\ctx}{C}
\newcommand{\ctxtwo}{D}
\newcommand{\ctxthree}{E}

\newcommand{\ctxp}[1]{\ctx\ctxholep{#1}}

\newcommand{\hctx}{H}
\newcommand{\hctxtwo}{K}
\newcommand{\hctxthree}{G}
\newcommand{\hctxp}[1]{\hctx\ctxholep{#1}}
\newcommand{\hctxtwop}[1]{\hctxtwo\ctxholep{#1}}
\newcommand{\hctxthreep}[1]{\hctxthree\ctxholep{#1}}

\newcommand{\nbvctxtwo}[1]{\nbvctxtwo{#1}}

\newcommand{\defeq}{:=}
\newcommand{\eqdef}{=:}
\newcommand{\grameq}{::=}

\newcommand{\isub}[2]{\{#1/#2\}}

\newcommand{\llbrace}{\{ \kern -0.27em \vert}
\newcommand{\rrbrace}{\vert \kern -0.27em \}}

\newcommand{\grammarpipe}{\mathrel{\big |}}

\renewcommand{\l}{\lambda}
\newcommand{\ie}{\textit{i.e.}\xspace}
\newcommand{\eg}{\textit{e.g.}\xspace}
\newcommand{\ih}{\textit{i.h.}\xspace}
\newcommand{\fv}[1]{\symfont{fv}(#1)}

\newcommand{\red}[1]{{\color{red} {#1}}}
\newcommand{\blue}[1]{{\color{blue} {#1}}}

\newcommand{\ignore}[1]{}

\newcommand{\myinput}[1]{\ifthenelse{\boolean{withimages}}{\input{#1}}{}}

\newcommand{\nat}{\mathbb{N}}

\newcommand{\size}[1]{|#1|}
\newcommand{\sizeparam}[2]{|#1|_{#2}}

\ifthenelse{\boolean{PACMPL}}{\renewcommand{\state}{s}}{\newcommand{\state}{s}}
\newcommand{\statetwo}{{s'}}
\newcommand{\statethree}{s''}

\newcounter{numberone}

\newcounter{numbertwo}

\newcommand{\trpos}{logged position\xspace}
\newcommand{\trposs}{logged positions\xspace}

\newcommand{\TrPoss}{Logged Positions\xspace}

\newcommand{\Log}{Log\xspace}

\renewcommand{\ctxholep}[1]{\langle #1\rangle}
\newcommand{\ctxtwop}[1]{\ctxtwo\ctxholep{#1}}
\newcommand{\ctxthreep}[1]{\ctxthree\ctxholep{#1}}

\newcommand{\dom}[1]{dom(#1)}

\newcommand{\reflemma}[1]{Lemma~\ref{l:#1}}

\newcommand{\refprop}[1]{Prop.~\ref{prop:#1}}
\newcommand{\refpropp}[2]{Prop.~\ref{prop:#1}.\ref{p:#1-#2}}
\newcommand{\refsect}[1]{Sect.~\ref{sect:#1}}

\newcommand{\refthm}[1]{Theorem~\ref{thm:#1}}

\newcommand{\reffig}[1]{Fig.~\ref{fig:#1}}
\newcommand{\refapp}[1]{Appendix~\ref{sect:#1}}

\renewcommand{\isub}[2]{\{#1{\shortleftarrow}#2\}}

\newcommand{\resm}{\psym}
\renewcommand{\resm}{\bullet}

\newcommand{\lpos}{p}
\renewcommand{\lpos}{l}

\newcommand{\sizelpos}[1]{\size{#1}_{\lpos}}

\newcommand{\upp}{\blue{\uparrow}}
\newcommand{\downp}{\red{\downarrow}}
\newcommand{\uppt}{\red{\uparrow}}
\newcommand{\downpt}{\blue{\downarrow}}
\newcommand{\tlog}{L}
\newcommand{\tlogtwo}{\tlog'}

\newcommand{\tlogn}{L_n}
\newcommand{\tape}{T}
\newcommand{\tapetwo}{\tape'}

\newcommand{\pol}{d}
\newcommand{\poltwo}{\pol'}

\newcommand{\run}{\pi}
\newcommand{\runtwo}{\sigma}

\newcommand{\relf}{{\blacktriangleright}}

\newcommand{\nopolstate}[5]{(#1,#2,#4,#3,#5)}
\newcommand{\dstate}[4]{(\red{\underline{#1}},#2,#4,#3)}

\newcommand{\ustate}[4]{(#1,\blue{\underline{#2}},#4,#3)}

\newcommand{\dstatetab}[4]{\red{\underline{#1}} & #2 & #4 & #3 }

\newcommand{\ustatetab}[4]{#1 & \blue{\underline{#2}} & #4 & #3 }

\newcommand{\ndstatetab}[5]{\red{\underline{#1}} & #2 & #4 & #3 & #5}
\newcommand{\nustatetab}[5]{#1 & \blue{\underline{#2}} & #4 & #3 & #5}

\newcommand{\cons}{{\cdot}}

\newcommand{\mach}{\mathrm{M}}

\newcommand{\IAM}{IAM\xspace}

\newcommand{\LIAM}{$\lambda$IAM\xspace}

\newcommand{\SIAM}{SIAM\xspace}
\newcommand{\TIAM}{TIAM\xspace}
\newcommand{\KAM}{KAM\xspace}

\newcommand{\IAMold}{\mathrm{IAM}}

\newcommand{\tomachhole}[1]{\rightarrow_{#1}}
\newcommand{\tomach}{\tomachhole{}}

\newcommand{\btsym}{\mathsf{bt}}
\newcommand{\tomachdotone}{\tomachhole{\resm 1}}
\newcommand{\tomachdottwo}{\tomachhole{\resm 2}}
\newcommand{\tomachvar}{\tomachhole{\varsym}}
\newcommand{\tomachbttwo}{\tomachhole{\btsym 2}}

\newcommand{\iamdap}{\tomachdotone}
\newcommand{\iamdlamone}{\tomachdottwo}
\newcommand{\iamdvar}{\tomachvar}
\newcommand{\iamdlamtwo}{\tomachbttwo}

\newcommand{\argsym}{\mathsf{arg}}
\newcommand{\tomachdotthree}{\tomachhole{\resm 3}}
\newcommand{\tomachdotfour}{\tomachhole{\resm 4}}
\newcommand{\tomacharg}{\tomachhole{\argsym}}
\newcommand{\tomachbtone}{\tomachhole{\btsym 1}}
\newcommand{\iamuapltwo}{\tomachdotthree}
\newcommand{\iamulam}{\tomachdotfour}
\newcommand{\iamuaplone}{\tomacharg}
\newcommand{\iamuapr}{\tomachbtone}

\newcommand{\toliam}{\rightarrow_{\lambda\textsc{IAM}}}

\newcommand{\totiam}{\rightarrow_{\textsc{TIAM}}}

\newcommand{\stempty}{\epsilon}

\newcommand{\sizee}[1]{\sizelpos{#1}}

\newcommand{\la}[1]{\lambda #1.}

\newcommand{\sem}[2]{\llbracket#1\rrbracket_{#2}}

\newcommand{\exstates}{\mathcal{E}}

\newcommand{\midd}{\; \; \mbox{\Large{$\mid$}}\;\;}

\newcommand{\spacem}[1]{|#1|_\mathtt{s}}

\newcommand{\ccallbn}{Closed Call-by-Name\xspace}
\newcommand{\ccbn}{Closed CbN\xspace}

\newcommand{\mset}[1]{[#1]}
\newcommand{\emmset}{[\cdot]}

\newcommand{\initty}{\star}

\newcommand{\linty}{A}
\newcommand{\lintytwo}{\linty'}
\newcommand{\lintythree}{\linty''}

\newcommand{\gty}{G}
\newcommand{\gtytwo}{\gty'}

\newcommand{\gtyctx}{\mathbb{\gty}}
\newcommand{\gtyctxp}[1]{\gtyctx\ctxholep{#1}}

\newcommand{\tty}{T}
\newcommand{\ttytwo}{\tty'}

\newcommand{\ttyctx}{\mathbb{\tty}}
\newcommand{\ttyctxp}[1]{\ttyctx\ctxholep{#1}}

\newcommand{\ltyctx}{\mathbb{\linty}}
\newcommand{\ltyctxp}[1]{\ltyctx\ctxholep{#1}}
\newcommand{\ltyctxtwo}{\ltyctx'}
\newcommand{\ltyctxtwop}[1]{\ltyctxtwo\ctxholep{#1}}
\newcommand{\ltyctxthree}{\ltyctx''}
\newcommand{\ltyctxthreep}[1]{\ltyctxthree\ctxholep{#1}}

\newcommand{\leafctx}{\mathbb{L}}
\newcommand{\leafctxp}[1]{\leafctx\ctxholep{#1}}

\newcommand{\arr}[2]{#1\rightarrow #2}

\newcommand{\tye}{\Gamma}
\newcommand{\tyetwo}{\Delta}

\newcommand{\tjudg}[3]{#1\vdash #2:#3}

\newcommand{\treesyst}{\textsc{T}}

\newcommand{\ctjudg}[2]{\vdash #1:#2}
\newcommand{\wtjudgone}[1]{\vdash^{\textcolor{violet}{#1}}}
\newcommand{\wtjudg}[4]{#1\stackrel{\textcolor{violet}{#2}}{\vdash}#3:#4}

\newcommand{\tjudgw}[4]{#1\vdash^{\textcolor{violet}{#2}} #3:#4}

\newcommand{\tyvar}{\textsc{T-Var}}
\newcommand{\tylamstar}{\tylam_\star}
\newcommand{\tylam}{\textsc{T-}\lambda}
\newcommand{\tyapp}{\textsc{T-@}}
\newcommand{\tymany}{\textsc{T-many}}
\newcommand{\tymanys}{\textsc{T-m}}
\newcommand{\tynone}{\textsc{T-none}}

\newcommand{\ttyprom}{\treesyst_{\textsc{many}}}

\newcommand{\tyd}{\pi}
\newcommand{\tydtwo}{\tyd'}
\newcommand{\tydthree}{\tyd''}
\newcommand{\pof}{\;\triangleright}

\newcommand{\tsys}[1]{\mathsf{T}_{#1}}

\newcommand{\DiPref}[1]{\mathsf{DiPref}(#1)}

\newcommand{\ruleoc}{J}

\newcommand{\extr}[1]{\mathsf{ext}(#1)}

\newcommand{\extsym}{\symfont{ext}}
\newcommand{\bisimtypes}{\simeq_{\extsym}}
\newcommand{\etape}[1]{\tape_{\extsym}(#1)}
\newcommand{\etapeaux}[2]{\tape_{\extsym}^{#2}(#1)}
\newcommand{\etapeauxs}[1]{\tape_{\extsym}^{\state}(#1)}
\newcommand{\elog}[1]{\tlog_{\extsym}(#1)}
\newcommand{\elpos}[1]{\lpos_{\extsym}(#1)}
\newcommand{\estate}[1]{\state_{\extsym}(#1)}

\newcommand{\relfrdx}{\relf_\mathsf{rdx}}
	\newcommand{\relfbody}{\relf_\mathsf{body}}
	\newcommand{\relfarg}{\relf_\mathsf{arg}}
	\newcommand{\relfext}{\relf_\mathsf{ext}}

\newcommand{\focus}{f}
\newcommand\mydots{\hbox to .6em{.\hss.}}

\newcommand{\ssize}[1]{\norm{#1}}
\newcommand{\leaves}[1]{#1^{\ell}}
\newcommand{\flatt}[1]{\underline{#1}}

\newcommand{\bsize}[1]{\sizeparam{#1}{{\symfont b}}}
\newcommand{\ttm}[1]{\text{TTMeas}(#1)}
\renewcommand{\ttm}[1]{\bsize{#1}}
\newcommand{\tlm}[1]{\text{TLMeas}(#1)}
\renewcommand{\tlm}[1]{\bsize{#1}}
\newcommand{\lm}[1]{\lambda \text{Meas}(#1)}
\renewcommand{\lm}[1]{\sizeparam{#1}{\mathsf{sp}}}

\newcommand{\indet}{\mathsf{X}}

\newcommand{\tyy}{\mathbb{Y}}
\newcommand{\tyF}{\mathbb{F}}
\newcommand{\tyt}{\mathbb{T}}

\newcommand{\tyl}{{\overline{\linty}}}

\newcommand{\states}[1]{{\symfont{states}}(#1)}

\newcommand{\weight}{w}
\newcommand{\weighttwo}{v}

\usepackage{tikz}
\makeatletter
\protected\def\tikz@nonactivecolon{\ifmmode\mathrel{\mathop\ordinarycolon}\else:\fi} 
\makeatother

\usetikzlibrary{matrix,arrows}
\usetikzlibrary{decorations.shapes}
\usetikzlibrary{decorations.text}
\usetikzlibrary{decorations.pathmorphing}
\usetikzlibrary{decorations.markings}
\usetikzlibrary{arrows.meta}

\usepackage{microtype}
\usepackage{version}

\includeversion{LONG}
\excludeversion{SHORT}

\renewcommand{\run}{\rho}

\begin{document}
%
\title{The Space of Interaction}

\author{\IEEEauthorblockN{Beniamino Accattoli}
\IEEEauthorblockA{Inria \& \'Ecole Polytechnique}
\and
\IEEEauthorblockN{Ugo Dal Lago}
\IEEEauthorblockA{Universit\`a di Bologna \& Inria}
\and
\IEEEauthorblockN{Gabriele Vanoni}
\IEEEauthorblockA{Universit\`a di Bologna \& Inria}}


%

\maketitle

\begin{abstract}
The space complexity of functional programs is not well understood. In 
particular, traditional implementation techniques are tailored to time 
efficiency, and space efficiency induces time \emph{in}efficiencies, as it 
prefers 
re-computing to saving. 
Girard's geometry of interaction underlies an alternative approach based on the 
interaction abstract machine (\IAM), claimed as space efficient in the 
literature. It has also been conjectured to provide a reasonable notion of 
space for the $\l$-calculus, but such an important result seems to be elusive.

In this paper we introduce a new intersection type system precisely measuring 
the 
space consumption of the \IAM on the typed term. Intersection types have been 
repeatedly used to measure time, which they achieve by dropping idempotency, 
turning intersections into \emph{multisets}. Here we show that the space 
consumption of the 
\IAM 
is connected to a further structural modification, turning multisets into 
\emph{trees}.
Tree intersection types lead to a finer understanding of some space complexity 
results from the literature. They also shed new light on the conjecture 
about reasonable space: we show that the usual way of encoding Turing machines 
into the $\l$-calculus \emph{cannot} be used to prove that the space of the 
\IAM is a 
reasonable cost model.
\end{abstract}


%
\IEEEpeerreviewmaketitle

\section{Introduction}

Type systems are a form of compositional abstraction in which the
behavior of programs, particularly higher-order programs, is described
by \emph{types}, that is, specifications of the kinds of objects programs 
expect in input and are supposed 
to
produce as output. Typed programs usually \emph{cannot go wrong}, as types 
guarantee the absence of run-time 
errors~\cite{DBLP:books/crc/tucker97/Cardelli97}. Some type systems ensure also 
other
properties such as termination or (various forms of) security. 

\paragraph*{Intersection Types and Time Complexity} Among the many existing type systems one can define on top of the 
pure,
untyped, $\lambda$-calculus~\cite{barendregt_lambda_1984,krivine1993lambda}, 
\emph{intersection types} have the peculiar
feature of \emph{characterizing} a property rather than merely
\emph{guaranteeing} it. This, in particular, is known to hold for
various forms of termination properties since the pioneering work by
Coppo and Dezani in 1978~\cite{DBLP:journals/aml/CoppoD78}. 
More recently,
a variant of intersection types has been proved to reflect 
\emph{quantitative}
properties of terms, such as the number of $\beta$-steps to normal form, or the number of steps of the Krivine abstract 
machine~\cite{krivine_call-by-name_2007} (shortened to \KAM), as discovered by 
de
Carvalho~\cite{Carvalho07,deCarvalho18,DBLP:journals/tcs/CarvalhoPF11}. The variant requires dropping idempotency of the
intersection operator, therefore considering $\linty \wedge \linty$ as 
\emph{not} equivalent to $\linty$, and ultimately making the type system 
strongly related 
to the modeling of resources as in linear logic. Such types are sometimes 
called \emph{multi 
types}, as intersections become multisets. 
In the last few years, de Carvalho's results have been dissected and 
generalized in 
various 
ways 
\cite{DBLP:journals/jfp/AccattoliGK20,DBLP:conf/aplas/AccattoliG18,DBLP:conf/esop/AccattoliGL19,DBLP:conf/types/AlvesKV19,DBLP:conf/flops/BucciarelliKRV20,DBLP:conf/lics/KesnerV20,DBLP:journals/pacmpl/LagoFR21}.
 In particular, research on the topic received renewed attention after some 
recent progress in the study 
of reasonable cost models for the $\l$-calculus by Accattoli and Dal Lago \cite{accattoli_leftmost-outermost_2016} made 
evident that counting $\beta$-steps gives rise to a reasonable cost model for 
time.

\paragraph*{Sharing, Time, and Space} The mainstream way of implementing the 
$\l$-calculus, at work also in the \KAM, consists of mimicking $\beta$-steps 
while replacing meta-level substitutions with a finer approach based on 
environments. Environments are a way of realizing a form of \emph{subterm 
sharing}, 
which is known to be mandatory in order to implement $\beta$-reduction in a time-efficient way. Traditional environment-based machines do not seem to be the right tool for space-efficiency. The reason is that these machines create a sharing 
annotation---counting as a unit of space---for \emph{every} $\beta$-step---that 
is, for every abstract unit of time. They never garbage collect, because garbage collection is postponable, and ignoring it is safe for time, as its cost is negligible (it is polynomial, if not linear, in the number of $\beta$-steps). 
As a consequence, their use of space is 
linear in their time usage, which is the worst possible use of space. To study 
space efficiency, then, there are two possible approaches: either refining 
environment machines with an explicit treatment of garbage collection, or 
exploring alternative execution schemas. In this work, we follow the second 
approach. 



\paragraph*{Evaluating Without Sharing} Beyond the mainstream approach, there is an alternative execution 
schema for $\l$-terms, rooted in Girard's Geometry of 
Interaction~\cite{GIRARD1989221} and Abramsky, Jagadeesan, and Malacaria game semantics \cite{DBLP:journals/iandc/AbramskyJM00}, which does \emph{not} rely on sharing. The 
\emph{interaction abstract machine} (shortened to \IAM), first studied by Mackie~\cite{mackie_geometry_1995} and Danos \& Regnier~\cite{danos_regnier_1995,DBLP:conf/lics/DanosHR96}, evaluates a 
$\l$-term without tracing every $\beta$-step, thus disentangling time and space. As it is the case for space-efficient 
Turing machines, the \IAM sometimes repeats computations to 
retrieve unsaved intermediate results---thus treading time for space. 
Environments are replaced by a \emph{token}, a 
 data structure where the machine saves minimal information, and the repetition of computations is realized via a 
sophisticated backtracking mechanism (unrelated to backtracking as in control 
operators or classical 
logic). The minimal information in the token amounts to tracing the 
points 
along the execution history where the \IAM
may need to backtrack. The entries of the token are \emph{trees} of pointers called \emph{logged positions} in a recent presentation of the \IAM by Accattoli, Dal Lago, and Vanoni \cite{IamPPDPtoAppear} (called instead \emph{exponential signatures} by Danos \& Regnier~\cite{danos_regnier_1995}). Everything else is ignored, in particular the token does not record encountered 
$\beta$-redexes.

Further evidence of the relavance for space of the interaction paradigm, comes from Ghica's \emph{Geometry of
  Synthesis}~\cite{ghica_geometry_2007,DBLP:journals/entcs/GhicaS10},
in which the geometry of interaction is seen as a compilation scheme towards 
circuits, whose computation space is \emph{finite}, and of paramount importance.

\paragraph*{The Subtle Complexity of the \IAM} These 
considerations suggest that the \IAM is, roughly, bad for time and good for 
space. Interestingly, the situation is subtler. About time, there are indeed 
examples showing that the \IAM is sometimes exponentially slower than 
environment machines. The slow-down however is not uniform, as in many cases 
the \IAM exhibits good time performances.

 The situation about space is 
not clear, either. The \IAM has been used in the literature for obtaining 
sub-linear space bounds for functional 
programs~(Dal Lago and Sch\"opp \cite{bllspace,dal_lago_computation_2016}, 
Mazza \cite{DBLP:conf/csl/Mazza15} and Ghica~\cite{ghica_geometry_2007}), 
something hardly 
achievable with traditional environment machines. Having a close look at these results, however, one 
realizes that those bounds rely crucially on some tweaks (restricting to certain 
$\l$-terms or extending 
the language with ad-hoc constructs) and that they do not seem to be achievable on ordinary $\l$-terms.

A further evidence of the ambiguous space behavior of the \IAM is that the folklore conjecture that the \IAM space 
usage is a reasonable space cost model (that is, linearly related to the one of Turing machines) has been circulating among specialists for years, but has never found an answer.

\paragraph*{Multi Types and the Time of the \IAM} Very recently, there have 
been advances in the understanding of the subtle time behavior of the \IAM. 
Accattoli, Dal Lago, and Vanoni have shown in~\cite{ADLVPOPL21} how to 
extract exact time bounds for the \IAM from multi types derivations. Interestingly,
they use the same types as de Carvalho's, despite the \IAM and the \KAM computing in very different ways. While the 
time of the \KAM is given by the multiplicity of the multi sets in multi types, they show that the time of the \IAM 
requires to take into account also the size of the involved types. 
The result also provides a high-level understanding for the time (in)efficiency of the \IAM: the inefficiency is proportional to the size of types. Since higher-order types are
bigger than first-order ones, the more a program uses higher-order types the more its execution with the \IAM is slower than with, say, the \KAM. 

\subsection{Contributions of the Paper}
The aim of this paper is to provide advances in the understanding of the subtle 
\emph{space} behavior of the \IAM. Inspired by the mentioned recent results by 
Accattoli, Dal Lago, and Vanoni, we introduce a new variant 
of multi types 
from which we extract \emph{exact 
space bounds} for the \IAM. To our knowledge, it is the first use of 
intersection/multi types to 
measure space. 

A key point is that multi types as used in~\cite{ADLVPOPL21}---as well as in the 
many recent papers to extract quantitative bounds on $\l$-terms cited above---cannot measure the space consumption of the \IAM, as they do not have enough 
structure. Our 
work shows that at the type level, indeed, one needs to add a 
\emph{tree} 
structure to multi sets, similarly to 
how measuring time requires switching from idempotent intersections to 
multisets.

Once the main result is proved, we show how to use it to understand the subtle 
space behavior of the \IAM. Here, the main 
insight is a negative perspective about the elusive conjecture that the space of the \IAM is a reasonable cost model. We also show that the 
new type system enlightens the key ingredients of the sub-linear space bounds in the literature.

\paragraph*{Tree Types} Usually, multi types are structured in two mutually recursive 
and disjoint layers, linear types and multisets of linear types. Our types also 
have two layers, but they are no longer 
disjoint, as multisets can also contain multisets, not only linear types. 
This way multisets can naturally be 
seen as \emph{trees}, whose internal nodes are multisets and whose leaves are 
linear types.
Such tree types are very natural, and probably of independent interest. While flattening the tree 
structure recovers an ordinary multi type, the converse operation cannot be done in a canonical way. This fact
shows that tree types add something new, they do not simply express a structure 
already present in multi types.

Actually, the tree structure reflects information about the token. The 
intuition is that while the number of leaves of a tree type counts the 
different uses of a variable/argument---as for multi types---internal nodes on 
the path to a leaf instead count how many backtracking pointers are stored in 
the token by the time the \IAM gets to that copy.

\paragraph*{The Complexity Analysis} Our space bounds are obtained by building 
on the technique developed by Accattoli, Dal Lago, and Vanoni in 
\cite{ADLVPOPL21}, which is inherently different from the one by de 
Carvalho and used in other recent works. The technique in \cite{ADLVPOPL21} 
amounts to first defining an auxiliary machine evaluating multi type 
derivations and showing it bisimilar to the \IAM. In this way, we obtain a representation of the states of the  \IAM run on the type derivation. Then, bounds are extracted 
from a global analysis based on \emph{weighted} type derivations. We proceed 
similarly, replacing multi types with tree types, and introducing a new system 
of weights for space, based on the depth of tree types as trees. 

Once the subtle bisimulation is established, our space complexity analysis is 
extremely simple, and also naturally provides exact bounds. This provides 
evidence 
that our types system is not ad-hoc. On the contrary, we believe it unveils a 
fundamental enrichment of multi types, deserving further studies.

\paragraph*{Background on Reasonable Space} In order to discuss the mentioned 
conjecture about reasonable space, let us provide some background. First of 
all, studying interesting space complexity classes such as \L, requires being 
able to measure \emph{sub-linear} space. There is a recent result in the 
literature about reasonable space for the 
$\l$-calculus, by Forster, Kunze, and Roth 
\cite{DBLP:journals/pacmpl/ForsterKR20}, but their cost model---namely, the 
size of the term---can  only measure
linear and super-linear space, and thus it is not a solution for the general problem.

 Now, showing that a cost model is reasonable requires studying the 
relationship with another known-to-be reasonable model, typically 
Turing machines. While for time the delicate direction is the simulation of the $\l$-calculus into Turing machines, for 
space the subtle one is the simulation of Turing machines in the $\l$-calculus, 
as it is hard to use as little space as a
Turing machine. 

The iteration of the transition function of Turing machines is a form of tail 
recursion, which in the $\l$-calculus is encoded via fixed-point operators. 
Such operators are also used to represent minimization in the encoding of 
partial recursive functions. Our insight about the reasonable space conjecture 
stems from an analysis of fixed-points operators in our type system.

\paragraph*{The Elusive Reasonable Space Conjecture}

Our contribution here is the fact that the \IAM\ performs poorly when 
evaluating fixed-point operators, namely using an amount of space which is 
always at least \emph{linear} in the number of
performed recursive calls. This is done by deriving the type (schema) of a 
specific fixed-point operator in our system, 
then showing that its use of space is proportional to its use of time. 

The specific operator we type is the one used in the encoding of Turing machines in the $\l$-calculus used by Accattoli 
and Dal Lago in their study of reasonable 
time~\cite{DBLP:journals/corr/abs-1711-10078,DBLP:conf/rta/AccattoliL12}, as 
well as by Forster, Kunze, 
and Roth in 
\cite{DBLP:journals/pacmpl/ForsterKR20}. Seeing it as the natural way of encoding tail recursion, it follows that the 
\IAM space performance is poor with tail recursion, and, in turn, with the natural way of encoding Turing machines. 

Summing up, our insight is that a positive answer to the conjecture cannot be done using the standard encoding of Turing machines, which explains the elusiveness of the conjecture. 

\paragraph*{Trees Density} A way of abstracting away from the issue 
with fixed-points is to look at how information is organized in the tree 
structure of tree types, itself reflecting 
the structure of (logged positions in) the token. When the tree is dense (its 
height is roughly the log of its number of nodes), 
then the \IAM execution is space-efficient, while when the tree is sparse (height close to the number of 
nodes) it is inefficient---the type schema for fixed-points is indeed sparse. 

\paragraph*{Perspectives} The insight about the density of trees allows to re-understand some 
results in the literature, and opens new perspectives. As we detail in 
Section~\ref{sect:inefficiency}, it sheds light on Dal Lago 
and Sch\"opp's space bounds for a functional language~\cite{bllspace,dal_lago_computation_2016}, as well as on the encoding of sub-linear
space computations in the $\l$-calculus by Mazza \cite{DBLP:conf/csl/Mazza15}.

%


\paragraph*{Proofs} Proofs are in the Appendix. 


\section{\ccallbn and Abstract Machines}

	Let $\mathcal{V}$ be a countable set 
	of variables. 
	Terms of the \emph{$\lambda$-calculus} $\Lambda$ are defined as follows.
	\begin{center}$\begin{array}{rrcl}
	\textsc{$\l$-terms} & \tm,\tmtwo,\tmthree & \grameq & x\in\mathcal{V}\midd 
	\lambda x.\tm\midd 
	\tm\tmtwo.
	\end{array}$
\end{center}
	\emph{Free} and \emph{bound variables} are defined as 
	usual: $\la\var\tm$ binds $\var$ in $\tm$. A term is \emph{closed} when 
	there are no free occurrences of variables in it.
	Terms are considered modulo $\alpha$-equivalence, and capture-avoiding (meta-level) substitution of 
	all the free occurrences of $\var$ for $\tmtwo$ in $\tm$ is noted 
	$\tm\isub\var\tmtwo$. Contexts are just $\lambda$-terms containing exactly 
	one occurrence of a special symbol, the \emph{hole} $\ctxhole$, intuitively standing for a removed subterm. Here we 
adopt 
	\emph{leveled} contexts, whose index, \ie\ the level, stands for the number of 
	arguments (that is, the 
	number of !-boxes in linear logic terminology) the hole lies in.
	\begin{center}$
	\begin{array}{rclrrrcl}
	\multicolumn{8}{c}{\textsc{Leveled contexts}}
	\\
	\ctx_0		& \grameq &	\ctxhole \midd \la\var\ctx_0 \midd \ctx_0\tm;
	\\
	\ctx_{n+1}	& \grameq &	\ctx_{n+1}\tm\midd\la\var\ctx_{n+1}\midd\tm\ctx_{n}.
	\end{array}
$\end{center}
	We simply write $\ctx$ for a context whenever the level is not relevant. 
	The operation replacing the hole $\ctxhole$ with a term $\tm$ 
in a context $\ctx$ is noted $\ctxp\tm$ and called \emph{plugging}.
	
	The operational semantics that we adopt here is weak head evaluation 
	$\towh$, defined as follows:
	\begin{center}$
(\la\vartwo \tm) \tmtwo \tmthree_1 \ldots \tmthree_h \ \ \towh \ \ \tm \isub\vartwo \tmtwo \tmthree_1 \ldots \tmthree_h.
$\end{center}
where $\tm \isub\vartwo \tmtwo$ is our notation for meta-level substitution.
We further restrict the setting by considering only closed terms, and refer to 
our framework as \emph{\ccallbn} (shortened to \ccbn). Basic well known facts 
are that in \ccbn the normal forms are precisely the abstractions and that 
$\towh$ is 
deterministic.

\paragraph*{Abstract Machines Glossary}  In this paper, an \emph{abstract 
machine} $\mach = (\state, \tomach)$ is a transition system $\tomach$ over a 
set of states, noted $\state$. 
The machine considered in this paper moves over the code without ever changing it. A \emph{position} in a term 
$\tm$ is represented as a pair $(\tmtwo,\ctx)$ of a sub-term $\tmtwo$ and a context $\ctx$ such that $\ctxp\tmtwo=\tm$. 
States are composed by a position $(\tmtwo,\ctx)$ plus some data 
structures. 
A state is \emph{initial}, and noted $\state_\tm$, if it is positioned on $(\tm,\ctxhole)$, $\tm$ is closed, and all the 
data structures are empty---$\tm$ is always implicitly 
considered closed, without further mention. A state is \emph{final} if no 
transitions apply.

 A \emph{run} $\run: \state \tomach^*\statetwo$ is a possibly empty sequence of transitions, whose length is noted 
$\size\run$. An \emph{initial run} is a run from an initial state $\state_\tm$, 
and it is also called \emph{a run from $\tm$}. A state $\state$ is 
\emph{reachable} if it 
is the target state of an initial run. A \emph{complete run} is an initial run ending on a final state. 


\section{The Interaction Abstract Machine, Revisited}
\label{sec:IJK}

In this section we provide an overview of the Interaction Abstract Machine 
(IAM). We adopt the $\l$-calculus 
presentation of the \IAM, rather called \LIAM and recently developed by Accattoli, Dal Lago, and Vanoni in 
\cite{IamPPDPtoAppear}---we refer to 
their work for an in-depth study of the \LIAM. 
The literature usually studies the ($\l$)\IAM with respect to head evaluation 
of potentially open terms, here 
we only deal with \ccbn, which is closer to the practice of functional 
programming.
Keep in mind that the \LIAM is an unusual machine, and that finding 
it hard to grasp is 
normal. Also, in \cite{IamPPDPtoAppear} there is an alternative 
explanation of the \LIAM, that may be 
helpful, together with the relationship with proof nets, which is however not 
needed here.
 \begin{figure*}[t]
	\input{machines/LIAM}
	\vspace{-8pt}
	\caption{Data structures and transitions of the $\lambda$ Interaction 
	Abstract Machine (\LIAM).}
	\label{fig:iam}
\end{figure*}

\paragraph*{Bird's Eye view of the \LIAM} Intuitively, the behaviour of the 
\LIAM can be seen as 
that of a token that travels around the syntax tree of the program under 
evaluation. Similarly to environment machines such as Krivine's, it looks 
for the head variable of a term. The peculiarity of the \LIAM is that it does 
not store the 
encountered 
$\beta$-redexes in an environment. When it finds 
the head variable, the \LIAM looks for the argument which should replace 
it, because having no environment, it 
cannot simply look it up. These two search mechanisms are realized by two 
different phases and directions of 
exploration of the code, noted $\downp$ and $\upp$. The functioning is actually 
more involved because there is also a 
backtracking mechanism (which however has nothing to do with backtracking as 
modeled by classical logic 
and continuations), requiring to save and manipulate code positions in the 
token. Last, the machine never duplicates 
the code, but it distinguishes different uses of a same code (position) using 
\emph{logs}. There are no easy 
intuitions about how logs handle different uses---this is both the magic and 
the mystery of the geometry of 
interaction.

\paragraph*{\LIAM States}
The data structures of the \LIAM are in 
\reffig{iam}.
The \LIAM travels on a $\l$-term $\tm$ carrying data 
structures---representing the token---storing 
information about the computation and determining the next transition to apply. 
It travels according to a
  \emph{direction} of navigation that is either 
$\downp$ or $\upp$, pronounced \emph{down} and \emph{up}.

The \emph{token} is given by two stacks, 
called \emph{log} and \emph{tape}, whose main components are \emph{\trposs}. 
Roughly, a log is a trace of the relevant 
positions in the history of a computation, and a logged 
position is a position plus a log, meant to trace the history that led to that 
position. Logs and logged positions are 
defined by mutual induction. Note that in the definition of a 
logged position, the log is required to have length $n$, where $n$ is the level 
of the context of the position.
We use $\cdot$ also to concatenate logs, 
writing, \eg, $\tlog_n\cdot\tlog$, using $\tlog$ for a log of unspecified 
length. The \emph{tape} $\tape$ is a list of logged positions plus occurrences 
of the special symbol 
$\resm$, needed to record the crossing of abstractions 
and applications. 

A \emph{state} of the machine is given by a position and a 
token (that is, a log $\tlog$ and a tape $\tape$), together with a 
\emph{direction}.
Initial states have the form $\state_{\tm}\defeq 
\dstate{\tm}{\ctxhole}{\epsilon}{\epsilon}$.
Directions are often omitted and represented via colors and underlining: 
$\downp$ is represented by a
\red{red} and underlined code term, $\upp$ by a \blue{blue} and underlined code 
context.

\paragraph*{Transitions} The transitions of the \LIAM are in 
\reffig{iam}. Their union is noted $\toliam$. The idea is 
that $\downp$-states 
$\dstate\tm\ctx\tape\tlog$ are queries about the head variable of (the head 
normal form of) $\tm$ and 
$\upp$-states $\ustate\tm\ctx\tape\tlog$ are queries about the argument of an 
abstraction. A key point is that navigation is done locally, moving only 
between adjacent 
positions\footnote{Transitions $\tomachvar$ and $\tomachbttwo$ might not look local, as they jump from a 
bound variable occurrence to its binder, and viceversa. If $\l$-terms are represented by implementing occurrences as 
pointers to their binders, as in the proof net 
  representation of $\l$-terms---upon which some concrete implementation schemes 
  are based, see \cite{DBLP:conf/ppdp/AccattoliB17}---then they are local.}. Intuitively, the machine
evaluates the term $\tm$ until the head abstraction of the head normal form is 
found (more explanations below).
The transitions realize three entangled mechanisms.

	\paragraph*{Mechanism 1: Search Up to $\beta$-Redexes} 
Note that $\iamdap$ skips the 
 argument and adds a $\resm$ on the tape. The idea is that $\resm$ keeps track 
 that an argument has been encountered---its identity is however forgotten. 
 Then $\iamdlamone$ does the dual job: it skips an abstraction when the tape 
 carries a $\resm$, that is, the trace of a previously encountered 
 argument. Note that, when the \LIAM moves through a $\beta$-redex with the 
 two steps one after the other, the token is 
left unchanged. This mechanism thus realizes search \emph{up to 
$\beta$-redexes}, 
 that is, without ever recording them. Note that 
 $\iamuapltwo$ and $\iamulam$ do the same during the $\upp$ phase. 
 
 Let us 
 illustrate this mechanism with an example: the first steps of 
 evaluation on 
 the term $\mathsf{I}((\la\var\var\var)\mathsf{I})$, where 
 $\mathsf{I}$ 
 is the identity combinator. One can notice that the \LIAM traverses one 
 $\beta$-redexes without altering the token, that is empty both at the 
 beginning and at the end.
 \[{\footnotesize
 	\begin{array}{l@{\hspace{0.1cm}}l|c|c|c|l}
 		&\mathsf{Sub}\mbox{-}\mathsf{term} & \mathsf{Context} & 
 		\mathsf{\Log} & 
 		\mathsf{Tape} & \mathsf{Dir}
 		\\
 		\cline{1-6}
 		&\ndstatetab{(\la\varthree\varthree)((\la\var\var\var)\mathsf{I})} 
 		{\ctxhole} 
 		{\epsilon} {\epsilon} 
 		\downp\\
 		\iamdap&\ndstatetab{\la\varthree\varthree} 
 		{\ctxhole((\la\var\var\var)\mathsf{I})} {\resm} {\epsilon} \downp\\
 		\iamdap&\ndstatetab{\varthree} 
 		{(\la\varthree\ctxhole)((\la\var\var\var)\mathsf{I})} 
 		{\stempty} {\epsilon} \downp\\
 \end{array}}
 \]

\paragraph*{Mechanism 2: Finding Variables and Arguments} As a first 
approximation, navigating in direction $\downp$ 
corresponds to looking for the head variable of the term code, while navigating 
with direction $\upp$ corresponds to 
looking for the sub-term to replace the previously found head variable, what we 
call \emph{the argument}. More 
precisely, when the head variable $\var$ of the active subterm is found, 
transition $\iamdvar$  switches direction from 
$\downp$ to $\upp$, and the machine starts looking for potential substitutions 
for $\var$. The \LIAM then moves to 
the position of the binder $\lambda\var$ of $\var$, and starts exploring the 
context $\ctx$, looking for the first 
argument up to $\beta$-redexes. The relative position of $\var$ w.r.t. its 
binder is recorded in a new \trpos that is 
added to the tape. Since the machine moves out of a context of level $n$, 
namely $\ctxtwo_n$, the \trpos contains the 
first $n$ \trposs of the log. Roughly, this is an encoding of the run that led 
from the level of 
$\la\var\ctxtwo_n\ctxholep\var$ to the occurrence of $\var$ at hand, in case 
the machine would later need to backtrack.
 
 When the argument $\tm$ for the abstraction binding the variable $\var$ in 
 $\lpos$ is found, transition $\iamuaplone$ switches direction from $\upp$ to 
 $\downp$, making the machine looking for 
the head variable of $\tm$. 
 Note that moving to $\tm$, the level 
 increases, and that the \trpos $\lpos$ is moved from the tape to the 
log. 
The idea is that $\lpos$ is now a completed argument query, 
 and it becomes part of the history of
  how the machine got to the current 
 position, to be potentially used for backtracking.
  We continue the example of 
 the previous point: the machine finds the head variable $\varthree$ and looks 
 for 
 its argument in $\upp$ mode. When it has been found, the direction turns to
 $\downp$ again and the process continues as before: first the head variable is 
 found and then the machine looks for its argument. Let us set 
 $\lpos_\varthree\defeq(\varthree,(\la\varthree\ctxhole)((\la\var\var\var)\mathsf{I}),\stempty)$,
 $\lpos_{\ctxhole\var}\defeq(\var,\la x \ctxhole x,\epsilon)$ and 
 $\lpos_\vartwo\defeq(\vartwo,\la\vartwo\ctxhole,\stempty)$.
 	\[{\footnotesize
 		\begin{array}{l@{\hspace{0.1cm}}l|c|c|c|c}
 			&\mathsf{Sub}\mbox{-}\mathsf{term} & \mathsf{Context} & 
 			\mathsf{\Log} & 
 			\mathsf{Tape} & \mathsf{Dir}
 			\\
 			\cline{1-6}
 			&\ndstatetab{\varthree} 
 			{(\la\varthree\ctxhole)((\la\var\var\var)\mathsf{I})} 
 			{\stempty} {\epsilon} \downp\\
 			\iamdvar&\nustatetab{\la\varthree\varthree} 
 			{\ctxhole((\la\var\var\var)\mathsf{I})} 
 			{\lpos_\varthree} {\epsilon} \upp\\
 			\iamuaplone&\ndstatetab{(\la x xx)\mathsf{I}} {\mathsf{I}\ctxhole} 
 			{\epsilon} 
 			{\lpos_\varthree} 
 			\downp\\
 			\iamdap&\ndstatetab{\la x xx} {\mathsf{I}(\ctxhole\mathsf{I})} 
 			{\resm} 
 			{\lpos_\varthree} 
 			\downp\\
 			\iamdlamone&\ndstatetab{xx} {\mathsf{I}((\la x 
 			\ctxhole)\mathsf{I}} {\epsilon} 
 			{\lpos_\varthree} \downp\\
 			\iamdap&\ndstatetab{x} {\mathsf{I}((\la x \ctxhole x)\mathsf{I})} 
 			{\resm} 
 			{\lpos_\varthree}
 			\downp\\
 			\iamdvar&\nustatetab{\la x xx} {\mathsf{I}(\ctxhole(\la y y))} 
 			{\lpos_{\ctxhole\var}\cdot\resm} {\lpos_\varthree}{\upp} \\
 			\iamuaplone&\ndstatetab{\la y y} {\mathsf{I}((\la x xx)\ctxhole)} 
 			{\resm} 
 			{\lpos_{\ctxhole\var}\cdot\lpos_\varthree}{\downp} \\
 			\iamdlamone&\ndstatetab{y} {\mathsf{I}((\la x xx)(\la y \ctxhole))} 
 			{\epsilon} 
 			{\lpos_{\ctxhole\var}\cdot\lpos_\varthree}\downp \\
 			\iamdvar&\nustatetab{\la y y} {\mathsf{I}((\la x xx)\ctxhole)} 
 			{\lpos_\vartwo} 
 			{\lpos_{\ctxhole\var}\cdot\lpos_\varthree}\upp \\
 	\end{array}}
 	\] 
\paragraph*{Mechanism 3: Backtracking} It is started by transition $\iamuapr$. 
 The idea is that the search for an argument of the $\upp$-phase has to 
 temporarily stop, because there are no arguments left at the current level. 
 The search of the argument then has to be done among the arguments of the 
 variable occurrence that triggered the search, encoded in $\lpos$. Then the 
 machine enters into backtracking mode, which is denoted by a $\downp$-phase 
 with a \trpos on the tape, to reach the position in $\lpos$. 
 Backtracking is over when $\iamdlamtwo$ is fired.
 
 The $\downp$-phase and the \trpos
 on the tape mean that the \LIAM is backtracking. During backtracking,
 the machine is 
 not looking for the head variable of the current code $\la\var\tm$, it is 
 rather going back to the variable position in the tape, to find its 
 argument. This is realized by moving to the position in the tape and 
 changing direction. Moreover, the log $\tlog_n$
 encapsulated in the \trpos is put back on the global log. An invariant 
  guarantees that the \trpos on the tape always contains a position 
 relative to the active abstraction. 
 
 In our example, a backtracking 
 phase starts when the \LIAM looks for the 
 argument of $\vartwo$. Since $\la\vartwo\vartwo$ has been virtually 
 substituted for $\var$, its argument is actually the second occurrence of 
 $\var$. Backtracking retrieves the variable which a term was virtually 
substituted for.
  \[{\footnotesize
  	\begin{array}{l@{\hspace{0.2cm}}l|c|c|c|c}
  	&\mathsf{Sub}\mbox{-}\mathsf{term} & \mathsf{Context} & \mathsf{\Log} & 
  	\mathsf{Tape} & \mathsf{Dir}
  	\\
  	\cline{1-6}
  	&\nustatetab{\la y y} {\mathsf{I}((\la x xx)\ctxhole)} 
  	{\lpos_\vartwo} 
  	{\lpos_{\ctxhole\var}\cdot\lpos_\varthree}\upp \\
  	\iamuapr&\ndstatetab{\la x xx} {\mathsf{I}(\ctxhole\mathsf{I})} 
  	{\lpos_{\ctxhole\var}\cdot\lpos_\vartwo} 
  	{\lpos_\varthree}\downp \\
  	\iamdlamtwo&\nustatetab{x} {\mathsf{I}((\la x \ctxhole x)\mathsf{I})} 
  	{\lpos_\vartwo} {\lpos_\varthree}\upp \\
  	\iamuaplone&\ndstatetab{x} {\mathsf{I}((\la x x\ctxhole)\mathsf{I})} 
  	{\epsilon} 
  	{\lpos_\vartwo\cdot\lpos_\varthree}\downp \\
  	\end{array}}\]
For the sake of completeness, we conclude the example, which runs until
 the head abstraction of the weak head normal form of the term under 
evaluation, namely the second occurrence of \textsf{I}, is found. We set 
$\lpos_{\var\ctxhole}\defeq(x,\la x x\ctxhole,\lpos_\vartwo)$.
\begin{center}${\footnotesize
	\begin{array}{l@{\hspace{0.2cm}}l|c|c|c|c}
	&\mathsf{Sub}\mbox{-}\mathsf{term} & \mathsf{Context} & \mathsf{\Log} & 
	\mathsf{Tape} & \mathsf{Dir}
	\\
	\cline{1-6}
	&\ndstatetab{x} {\mathsf{I}((\la x x\ctxhole)\mathsf{I})} 
	{\epsilon} 
	{\lpos_\vartwo\cdot\lpos_\varthree}\downp \\
	\iamdvar&\nustatetab{\la x xx} {\mathsf{I}(\ctxhole(\la y y))} 
	{\lpos_{\var\ctxhole}} {\lpos_\varthree}\upp \\
	\iamuaplone&\ndstatetab{\la y y} {\mathsf{I}((\la x xx)\ctxhole)} 
	{\epsilon} 
	{\lpos_{\var\ctxhole}\cdot\lpos_\varthree}\downp\\
	\end{array}}$\end{center}

 \paragraph*{Basic invariants} Given a state 
$(\tm,\ctx,\tlog,\tape,\pol)$, the log and the tape, \ie the token, 
verify two easy invariants connecting them to the position $(\tm,\ctx)$ and the 
direction $\pol$. The log $\tlog$ and the current position 
$(\tm,\ctx)$  form a \trpos, \ie the length of $\tlog$ is exactly the level 
of 
the code context 
$\ctx$\footnote{ Then, the length of $\tlog$ is exactly the number of (linear 
logic) \emph{boxes} in which the code term is contained.}. This fact guarantees 
that the \LIAM never gets stuck because the log is too short for 
transitions $\iamdvar$ and $\iamuapr$ to apply.

About the tape, note that every time the machine switches from a 
$\downp$-state to an $\upp$-state (or vice versa), a \trpos is 
pushed (or popped) from the tape $\tape$. Thus, for reachable states, the 
number of \trposs in $\tape$ gives the 
direction of the state. These intuitions are formalized by the \emph{tape and 
direction} invariant below. Given a direction $\pol$ we use
$\pol^n$ for the direction obtained by switching $\pol$ exactly $n$
times (i.e., $\downp^0=\downp$, $\upp^0=\upp$, $\downp^{n+1} =
\upp^{n}$ and $\upp^{n+1}=\downp^{n}$).
%
\begin{lemma}[\LIAM basic invariants]\label{l:invarianttwo}
  Let $\state = \nopolstate{\tm}{\ctx_n}{\tape}{\tlog}{\pol}$ be a reachable 
  state and $\sizee\tape$ the number of 
\trposs in $\tape$. Then 
  \begin{enumerate}
  	\item \emph{Position and log}: $(\tm,\ctx_n, \tlog)$ is a \trpos, and 
	\item \emph{Tape and direction}: $\pol=\downp^{\sizee\tape}$.
  \end{enumerate}
\end{lemma}

\paragraph*{Final States}
If the \LIAM starts on the initial state 
$\state_{\tm}$, the execution may either 
never stop or end in a state $\state$ of the shape 
$\state=\dstate{\la\var\tmtwo}{\ctx}{\epsilon}{\tlog}$. 
The fact that no other shapes are possible for $s$ is proved 
in~\cite{IamPPDPtoAppear}. 
\begin{figure*}[t!]
	\[
	\begin{array}{c@{\hspace{1cm}}c@{\hspace{1cm}}c}
	\infer[\tyvar]{\tjudg{\var:\mset{\linty}}{\var}{\linty}}{}
	&	
	\infer[\tylam]{\tjudg{\tye}{\lambda\var.\tm}{\arr{\tty}{\linty}}} 
	{\tjudg{\tye,\var:\tty}{\tm}{\linty}}
	&
\infer[\tyapp]{\tjudg{\tye\uplus\tyetwo}{\tm\tmtwo}{\linty}}{\tjudg{\tye}{\tm}{\arr{\tty}{\linty}} & 
\tjudg{\tyetwo}{\tmtwo}{\tty}}
	\\[8pt]
	\infer[\tylamstar]{\tjudg{\tye}{\lambda\var.\tm}{\initty}}{}
	&
	
	\infer[\tymany]{\tjudg{\mset{\uplus_{i=1}^n\tye_i}}{\tm}{\mset{\gty_1,\mydots,\gty_n}}}
{\tjudg{\tye_i}{\tm}{\gty_i} & 1\leq i\leq n}
		& \infer[\tynone]{\tjudg{}{\tm}{\emmset}}{}

	\end{array}
	\]
	\vspace{-8pt}
	\caption{The tree type system.}
	\label{fig:asstypesystem}
	\end{figure*}
\paragraph*{Implementation} Usually, the \LIAM is shown to implement (a 
micro-step variant of) head reduction. The details are quite different from 
those in the usual notion of implementation for environment machines, such as 
the KAM. Essentially, it is shown that the \LIAM induces a 
semantics $\sem{\cdot}{\text{\LIAM}}$ of terms that is a sound and adequate 
with 
respect to head reduction, rather than showing a bisimulation between the 
machine and head reduction---this is explained at length in 
\cite{IamPPDPtoAppear}. 
For the sake of simplicity, here we restrict to \ccbn. The \LIAM semantics then 
reduces to observing 
termination: $\sem{\tm}{\text{\LIAM}}$ is defined if and only if weak head 
reduction terminates on $\tm$. Therefore, we avoid 
discussing semantics and only study termination. We say that the \emph{\LIAM
implements \ccbn} when its execution from the initial state $\state_\tm$ reaches a final state if and only if $\towh$ terminates on $\tm$, for every closed term $\tm$.

\begin{theorem}[\!\!\cite{IamPPDPtoAppear}]
	The \LIAM implements \ccbn.
\end{theorem}

\paragraph*{\LIAM Space Consumption} The space needed to represent a  
\LIAM state is given by the following definition (the meta-variable $\Gamma$ to denote either a tape 
$\tape$ or a log $\tlog$):
\[
\begin{array}{rclrcl}
\multicolumn{6}{c}{\lm{\nopolstate{\tm}{\ctx}{\tape}{\tlog}{\pol}}\defeq\lm{\tlog}+\lm{\tape}}\\[3pt]
\lm{(\var,\ctxtwo,\tlogtwo)}&\defeq&\indet+\lm{\tlogtwo}&\lm{\stempty}&\defeq&0\\[3pt]
\lm{\lpos\cdot\Gamma}&\defeq&\lm{\lpos}+\lm{\Gamma}&\lm{\resm\cdot\tape}&\defeq&1+\lm{\tape}
\end{array}
\]
The value of the unknown $\indet$ is simply the size of a pointer to a subterm 
of the term under evaluation, \ie $\indet=\log\size{\ctxp\tm}$. Then, we are 
able to 
define the space of a \LIAM run by taking the maximum size of the states 
reached during the run.
\begin{definition}
	Let $\run:\state_0\toliam^*\state$ be a \LIAM run. Then,
	\[\lm{\run}\defeq\max_{\statetwo\in\run}\lm{\statetwo}\]
\end{definition}
It is worth noticing what happens in the case of diverging computations. In 
principle, two cases could occur: either the space consumption is finite, or it 
is infinite. Actually, it is easy to prove that the first case is not 
possible.
\begin{proposition}\label{prop:inf}
	Let $\run$ be an infinite \LIAM run. \\Then $\lm{\run}=\infty$.
\end{proposition}

\renewcommand{\tty}{T}


\section{Tree (Intersection) Types}
\label{sect:types}
Here we introduce a type system that we shall use to measure the space used by \LIAM runs.

\paragraph*{From Intersections Types to Tree Types} The framework that we adopt is the one of intersection types, with 
three tweaks:
\begin{enumerate}
	\item \emph{No idempotency}: we use the non-idempotent variant, where the intersection 
type $A\wedge A$ is not equivalent 
to $A$, and with 
 strong ties to linear logic and time analyses, because it takes into account how many times a resource/type $A$ is 
used, and not just whether $A$ is used or not. Non-idempotent intersections are 
multisets, which is why these types 
are 
sometimes called \emph{multi types} and an intersection $A\wedge B\wedge A$ is 
rather noted $\mset{A,B,A}$. 
 	\item \emph{Nesting/tree shape}: multi types are usually defined by two mutually 
dependent layers, a linear one containing ground types and (linear) arrow 
types, and the multiset level containing 
linear types. Here we also have two layers, except that we allow multisets to 
also contain multisets, thus we can 
have $\mset{A,\mset{\mset{B,B},A,\mset A},A,B}$. A nested multiset is a 
\emph{tree} 
whose leaves are linear types and whose internal nodes are nested multisets.

 	\item \emph{No commutativity}: we also consider \emph{non-commutative} tree types. 
Removing commutativity turns multisets into lists, or sequences, and trees 
into ordered trees. Removing commutativity is an inessential tweak. Our study does not depend on the ordered
structure, we shall only 
need bijections between multisets, to describe the reformulation of the \LIAM on type derivations, and these 
bijections are just more easily managed 
if commutativity is removed. This \emph{rigid} approach is also used by Tsukada, Asada, and Ong \cite{ongrigid} and Mazza, Pellissier, and Vial
\cite{MazzaPellissierVial}. The inessential aspect is stressed by referring to our types as to \emph{tree types}, 
rather than as to \emph{ordered tree types}, despite adopting the ordered variant.
	\end{enumerate}
\paragraph*{Basic Definitions} As for multi types, there are two mutually defined layers of types, \emph{linear types} and \emph{tree 
types}.
	\begin{center}$
	\begin{array}{rrcl@{\hspace{.6cm}}rrcl}
	\textsc{Linear types}&\linty,\lintytwo&\grameq&\initty\grammarpipe\arr{\tty}{\linty} \\
	\textsc{Tree types}&\tty,\ttytwo&\grameq&\mset{\gty_1,\ldots,\gty_n} & n\geq 0\\
	\textsc{(Generic) Types}&\gty,\gtytwo&\grameq&\linty\grammarpipe\tty
	\end{array}$
	\end{center}
Note that there is a ground 
type $\initty$, which can be thought as 
the type of normal forms, that in \ccbn are precisely abstractions. Note also that arrow (linear) types 
$\arr{\tty}{\linty}$ 
can have a tree type only on the left. About trees, since commutativity is 
ruled out, we have, for instance, that
$\mset{\linty,\lintytwo} \neq 
\mset{\lintytwo,\linty}$. 
Note that the empty tree type/sequence is a valid type, which is
noted $\emmset$. The concatenation of two sequences $\tty$ and $\ttytwo$ is noted $\tty\uplus\ttytwo$. We use $\size\tty$ for the length of $\tty$ as a sequence, that is, 
$\size{\mset{\gty_1,\mydots,\gty_n}} \defeq n$. 

Type 
judgments have the form $\tjudg{\tye}{\tm}{\gty}$, where $\tye$ is a type 
environment, defined below. Type 
derivations 
are noted $\tyd$ and we write 
$\tyd\pof\tjudg{\tye}{\tm}{\gty}$ for a type derivation $\tyd$ of ending 
judgment $\tjudg{\tye}{\tm}{\gty}$.  
Type environments, ranged over by $\tye,\tyetwo$ are total maps
from variables to tree types such that only finitely
many variables are mapped to non-empty tree types, and we write $\tye = 
\var_1:\tty_1,\ldots,\var_n:\tty_n$ if 
$\dom\tye = \set{\var_1,\ldots,\var_n}$---note that type environments 
\emph{are} commutative, what is non-commutative is 
the sequence constructor $\mset\cdot$, only. Given two type environments $\tye,\tyetwo$, we use
$\tye\uplus\tyetwo$ for the type environment
assigning to every variable $\var$ the list
$\tye(\var)\uplus\tyetwo(\var)$.

The 
typing rules are in \reffig{asstypesystem}. With respect to the literature, the 
difference is in rule $\tymany$. There 
are two differences. The first one is the already mentioned fact that premises may assign both linear types and tree 
types, while the literature usually only 
allows linear types. The second one is that the rule surrounds 
$\tye\defeq\uplus_{i=1}^n\tye_i$ with an additional nesting level---the 
notation $\mset{\tye}$ standing for the type environment 
$\var_1:\mset{\tty_1},\ldots,\var_n:\mset{\tty_n}$ if $\tye = \var_1:\tty_1,\ldots,\var_n:\tty_n$. 

\paragraph*{A Small Example} We show an instance of the rule $\tymany$ in the 
delicate case in which the premises contain the same free variable $\var$.
\[
\infer[\tymany]{\tjudg{\var:\mset{\mset{\linty_1,\linty_2}}} 
{\tm}{\mset{\gty_1,\gty_2}}}
{\tjudg{\var:\mset{\linty_1}}{\tm}{\gty_1}& 
\tjudg{\var:\mset{\linty_2}}{\tm}{\gty_2}}
\]
In particular, please note that first $\mset{\linty_1}$ and $\mset{\linty_2}$ 
are joined, and then they are surrounded by an additional nesting level. The other option would have been  
$\var:\mset{\mset{\linty_1},\mset{\linty_2}}$, but it is not what $\tymany$ does.

\paragraph*{Leaves Extraction and Leaf Contexts} Every tree type $\tty$ induces 
the sequence $\flatt\tty$--- 
equivalently, the flat tree type---of its leaves, defined by the following \emph{leaves extraction} operation.
\begin{center}$	\small\begin{array}{r@{\hspace{.33cm}}c@{\hspace{.33cm}}l}
\leaves{\emmset} \defeq \emmset
&
\leaves{(\mset{\linty}\uplus\tty)} \defeq \mset{\linty}\uplus\leaves{\tty}
&
\leaves{(\mset{\ttytwo}\uplus\tty)} \defeq \leaves\ttytwo\uplus\leaves{\tty}
	\end{array}
$\end{center}
We shall describe the leaves of a tree type also via a notion of leaf context.
\begin{center}$
	\begin{array}{r@{\hspace{.5cm}}lcl}
		\textsc{Leaf 
ctxs}&\leafctx\grameq\mset{\gty_1,\mydots,\ctxhole,\mydots,\gty_n}\grammarpipe
		\mset{\gty_{1},\mydots,\leafctx,\mydots,\gty_n}
	\end{array}$
\end{center}
If $\leaves\tty = \mset{\linty_1,\mydots,\linty_n}$ then for every $\linty_i$ there is a leaf context $\leafctx^i$ such 
that $\tty = \leafctx^i\ctxholep{\linty_i}$. Therefore, we shall use the notation $\tty = \leafctx^i\ctxholep{\linty}$, 
or even simply $\tty^i = \linty$, to say that the linear type $\linty$ is the $i$-th leaf of $\tty$.

In the following we use two basic properties of the type system, collected in the following straightforward lemma. One 
is the absence of 
weakening, and the other one is a correspondence between sequence types and 
axioms. 

\begin{lemma}[Relevance and axiom sequences]
\label{l:relevance}
If $\tyd \pof \tjudg{\tye}{\tm}{\linty}$ then $\dom\tye \subseteq \fv\tm$, thus if $\tm$ is closed then $\tye$ is 
empty. 
Moreover, there are exactly $\size{\leaves{\tye(\var)}}$ axioms typing $\var$ in $\tyd$, which appear from left to 
right 
as 
leaves of $\tyd$ (seen as an ordered tree) in the order given by $\leaves{\tye(\var)}= \mset{\linty_1,\ldots, 
\linty_k}$ 
and that the $i$-th axiom types $\var$ with $\linty_i$.
\end{lemma}
	
\paragraph*{Characterization of Termination} It is well-known that intersection and multi types characterize \ccbn 
termination, that is, they type \emph{all} and only those $\l$-terms that 
terminate with respect to \ccbn. Moreover, every term that is \ccbn normalizable can be typed with $\initty$. The same 
characterization holds with tree types, following the standard recipe\footnote{Namely, substitution lemma plus 
subject reduction for correctness,
and anti-substitution lemma, subject expansion, and typability of all normal forms for completeness (here trivial, because all normal forms are typed by $\tylamstar$).} for multi types, 
without surprises. See the Appendix for details.
\begin{theorem}[Correctness and completeness of tree types for \ccbn]
	A closed term $\tm$ is \ccbn normalizable if and only if there exists a 
	tree type derivation $\tyd \pof \ctjudg{\tm}{\initty}$.
\end{theorem}

\paragraph*{Relationship with Multi Types} The leaves extraction operation can easily be extended to a flattening 
function turning a tree type into a multi type. Flattening can also be extended to derivations, by collapsing trees of 
$\tymany$ rules into the more traditional rule for multi sets that does not \emph{nodify} the type context. 
In this way, one obtains a forgetful transformation, easily defined by induction on derivations. A converse 
\emph{lifting} transformation, however, cannot be defined by induction on derivations---it is unclear how to define it 
on applications. This fact is evidence that tree types are strictly richer than multi types, because the tree structure 
cannot be inferred from the multiset one.


\section{The Tree IAM}
\label{sect:TIAM}
This section introduces a machine evaluating type derivations, the \emph{Tree \IAM}, or \emph{\TIAM}, that mimics the 
\LIAM directly on top of a type derivation $\tyd$. It is the key tool that we shall use to measure the space cost of 
\LIAM runs. The \TIAM is a very minor variation over the similar \SIAM machine 
evaluating type derivations for sequence 
types by Accattoli, Dal Lago, and Vanoni in \cite{ADLVPOPL21}. This and the next section mostly just recall and adapt 
notions and results from that paper.

\paragraph*{Preamble about Duplications and (No) Logs} $\beta$-reducing a $\l$-term (potentially) duplicates arguments, whose different 
copies may be used differently, typically being applied to different further arguments. The \LIAM 
never 
duplicate arguments, but has nonetheless to distinguish different uses of the 
same piece of code. This is why 
it uses \emph{logged} positions instead of simple 
positions: the log is a trace of (part of) the previous run that allows to distinguishing different uses of the 
position.

The key point of multi/tree type derivations is that duplication is explicitly 
accounted for, \emph{in 
advance}, by multisets/trees: all arguments come with as many type derivations 
as the times they are 
duplicated during evaluation. With tree types, the number of copies is the number of leaves of the tree. 
Note that a multi/tree type derivation may be way bigger than the term itself, 
while this is not possible with, say, simple types. 

The intuition behind the \TIAM is that the walk over the code done by the \LIAM can be rephrased and simplified on multi/tree type derivations, because the mechanism of logs---needed to distinguish different copies of arguments---is no longer needed, since all copies are already there: simple positions in the type 
derivation (not in the term!) are informative enough.
\begin{figure*}[t]

\begin{center}\footnotesize
\begin{tabular}{c}
$\begin{array}{ccc||ccc}
	\infer{\tjudg{}{\red{\tm\tmtwo}}{\ltyctxp{\initty_\uppt} (=\linty)}} 
	{\tjudg{}{\tm}{\arr{\tty}{\linty}} & \vdash} 
	&\tomachdotone&
	\infer{\tjudg{}{\tm\tmtwo}{\linty 
		}}{\tjudg{}{\red\tm}{\arr{\tty}{\ltyctxp{\initty_\uppt}}} & 
		\vdash}
		&
	\infer{\tjudg{}{\red{\lambda\var.\tm}}{\arr{\tty} 
			{\ltyctxp{\initty_\uppt}}}} 
	{\tjudg{}{\tm}{\linty (= \ltyctxp{\initty})}}
	& \tomachdottwo &
	\infer{\tjudg{}{\lambda\var.\tm}{\arr{\tty}{\linty}}} 
	{\tjudg{}{\red\tm}{\ltyctxp{\initty_\uppt}}}
	 \\[8pt]\hhline{======}&&&\\

	\infer{\tjudg{}{\tm\tmtwo}{\linty(= \ltyctxp{\initty}) 
		}}{\tjudg{}{\blue\tm}{\arr{\tty}{\ltyctxp{\initty_\downpt}}} & 
		\vdash}
		
	&\tomachdotthree&
		\infer{\tjudg{}{\blue{\tm\tmtwo}}{\ltyctxp{\initty_\downpt}}} 
		{\tjudg{}{\tm}{\arr{\tty}{\linty}} & \vdash}
		&
		\infer{\tjudg{}{\lambda\var.\tm}{\arr{\tty}{\linty}}} 
		{\tjudg{}{\blue\tm}{\ltyctxp{\initty_\downpt} (=\linty)}}
		 & \tomachdotfour &
		\infer{\tjudg{}{\blue{\lambda\var.\tm}}{\arr{\tty} 
				{\ltyctxp{\initty_\downpt}}}} 
		{\tjudg{}{\tm}{\linty}}
		\\[8pt]\hhline{======}&&&\\
		\infer*{\infer{\tjudg{}{\la\var\ctxp{\var}} 
			{\arr{\leafctxp{\linty_i}}\lintytwo}}{}}
	{\infer[i]{\tjudg{}{\red\var}{\ltyctxp{\initty_\uppt}_i (= \linty_i)}}{}}   
	&\tomachvar&
	 \infer*{\infer{\tjudg{}{\blue{\la\var\ctxp{\var}}} 
			{\arr{\leafctxp{\ltyctxp{\initty_\downpt}_i}}\lintytwo}}{}}
	{\infer[i]{\tjudg{}{\var}{\linty_i}}{}}
	&
	
	\infer*{\infer{\tjudg{}{\red{\la\var\ctxp{\var}}} 
			{\arr{\leafctxp{\ltyctxp{\initty_\uppt}_i}}\lintytwo}}{}}
	{\infer[i]{\tjudg{}{\var}{\linty_i (=\ltyctxp{\initty}_i)}}{}}
	 & \tomachbttwo &
	 \infer*{\infer{\tjudg{}{\la\var\ctxp{\var}} 
			{\arr{\leafctxp{\linty_i}}\lintytwo}}{}}
	{\infer[i]{\tjudg{}{\blue\var}{\ltyctxp{\initty_\downpt}_i}}{}} 
	\\[8pt]\hhline{======}\\
		\end{array}$
		\\
	$\begin{array}{ccccccc }
	%
%
	%
	
		%
		\infer{\tjudg{}{\tm\tmtwo}{\linty}} 
		{\tjudg{}{\blue\tm}{\arr{\leafctxp{\ltyctxp{\initty_\downpt}_i}}{\linty}}
			& \infer=[\tymany]{\tjudg{}{\tmtwo}{\tty (=\leafctxp{\ltyctxp{\initty}_i)}} }{
			\ldots\ \tjudg{}{\tmtwo}{\ltyctxp{\initty}_i}\ \ldots
			 } }
		& \tomacharg &
		\infer{\tjudg{}{\tm\tmtwo}{\linty}} 
		{\tjudg{}{\tm}{\arr{\tty}{\linty}}
			& \infer=[\tymany]{\tjudg{}{\tmtwo}{\leafctxp{\ltyctxp{\initty}_i}} }{
			\ldots\ \tjudg{}{\red\tmtwo}{\ltyctxp{\initty_\uppt}_i}\ \ldots
			 } }
	
		\\[8pt]\hhline{===}\\

		\infer{\tjudg{}{\tm\tmtwo}{\linty}} 
		{\tjudg{}{\tm}{\arr{\tty}{\linty}}
			& \infer=[\tymany]{\tjudg{}{\tmtwo}{\leafctxp{\ltyctxp{\initty}_i} (=\tty)} }{
			\ldots\ \tjudg{}{\blue\tmtwo}{\ltyctxp{\initty_\downpt}_i}\ \ldots
			 } }
		 & \tomachbtone &
		 		\infer{\tjudg{}{\tm\tmtwo}{\linty}} 
		{\tjudg{}{\red\tm}{\arr{\leafctxp{\ltyctxp{\initty_\uppt}_i}}{\linty}}
			& \infer=[\tymany]{\tjudg{}{\tmtwo}{\tty} }{
			\ldots\ \tjudg{}{\tmtwo}{\ltyctxp{\initty}_i}\ \ldots
			 } }
		\end{array}$
		\end{tabular}
		
		\end{center}

	\vspace{-8pt}
	\caption{The transitions of the Tree \IAM (\TIAM).}
	\label{fig:tiam}  
\end{figure*}
\paragraph*{The \TIAM} On the one hand, the \TIAM is simpler than the \LIAM because it has no logs, on the other hand, it is more technical to define because tree type derivations are less easily manipulated than $\l$-terms. The underlying idea however is simple. The \TIAM moves over a fixed type derivation $\tyd\pof\tjudg{}{\tm}{\initty}$, to be thought as the 
code, following the occurrence of $\initty$ in the final judgment through 
$\tyd$, according to the transitions in \reffig{tiam}. We shall now explain every 
involved concept. 

The position of the machine is given by an occurrence of a type judgment\footnote{A 
judgment may occur repeatedly in a derivation, which is why we talk about \emph{occurrences} of  
judgments. To avoid too many technicalities, however, we usually just write the judgment, leaving implicit that we 
refer to an 
occurrence of that judgment.}  $\ruleoc$ 
of $\tyd$. As the \LIAM, the \TIAM has two possible directions, noted $\downpt$ and $\uppt$\footnote{Type 
derivations are upside-down wrt to the term structure, then direction $\downp$ of the \LIAM becomes here $\uppt$, and 
$\upp$ is $\downpt$.}. In direction $\uppt$ the machine looks at the rule above the focused judgment, in direction 
$\downpt$ at the rule below. The only ``data structure''---encoding the tape of 
the \LIAM, as we shall explain---is a type context $\ltyctx$ isolating an 
occurrence of 
$\initty$ in the type $\linty$ of the focused judgment (occurrence) $\tjudg{\tye}{\tmtwo}{\linty}$, defined as follows 
(careful to not 
 confuse type contexts $\ltyctx$ and $\ttyctx$ with type environments $\tye$):
\begin{center}$
	\begin{array}{r@{\hspace{.4cm}}lll}
		\textsc{Linear ctxs}&\ltyctx \grameq  \ctxhole \grammarpipe \arr\tty\ltyctx 
		\grammarpipe \arr{\ttyctx}\linty
		\\[3pt]
		\textsc{Tree 
		ctxs}&\ttyctx\grameq\mset{\gty_1,\mydots,\gtyctx,\mydots,\gty_n}\\[3pt]
		\textsc{Type ctxs}&\gtyctx\grameq\ltyctx\grammarpipe\ttyctx
	\end{array}$
\end{center}
Summing up, a state $\state$ is a 
quadruple $(\tyd, \ruleoc, \ltyctx, \pol)$. If $\ruleoc$ is in the form 
$\tjudg{\tye}{\tmtwo}{\linty}$, we often write $\state$ as 
$\tjudg{}{\tmtwo}{\ltyctxp{\initty_\pol}}$, where $\ltyctxp{\initty}=\linty$. We shall see that type environments play 
no role.

\paragraph*{Intuitions about Positions (and Logs)} The intuition is that the 
active position $(\tm, \ctx_n)$ of a \LIAM state corresponds to the judgment 
occurrence $\ruleoc$ in the \TIAM, or, more precisely, to its position in the 
type derivation $\tyd$. The sub-term $\tm$ is exactly the term typed by 
$\ruleoc$. The context $\ctx_n$ is exactly the context giving the term 
$\ctx_n\ctxholep\tm$ typed by the final judgment of $\tyd$. The level $n$ of 
the context $\ctx_n$ of the active position counts the number of arguments in 
which the hole of $\ctx_n$ is contained. On the type derivation, each such 
argument is associated to a $\tyapp$ rule having the current judgment $\ruleoc$ 
in its right sub-derivation. Last, note that in the \LIAM the current log 
$\tlogn$ has length equal to $n$. We shall see in the next section that one can 
recover the log $\tlogn$ applying an extraction process to the $\tyapp$ rules 
found descending from $\ruleoc$ towards the final judgment.

\begin{figure*}[t]
	\[{\footnotesize
		\infer[\tyapp]{\tjudgw{}{3\indet}{(\la 
				\varthree\varthree)((\la\var\var\var)(\la\vartwo\vartwo))} 
			{\initty_{\uppt\red{1}}}}{\infer[\tylam]{\tjudgw{}{\indet} 
				{\la\varthree\varthree}{\arr{\mset{\initty_{\downpt\blue{4}}}}{\initty_{\uppt\red{2}}}}}
			{\infer[\tyvar]{\tjudgw{\varthree:\mset\initty}{0}{\varthree} 
					{\initty_{\uppt\red{3}}}}{}}& 
			\infer[\tymanys]{\tjudgw{}{3\indet}{(\la\var\var\var)(\la\vartwo\vartwo)}{\mset{\initty}}}{
				\infer[\tyapp]{\tjudgw{}{2\indet}{(\la\var\var\var)(\la\vartwo\vartwo)}{\initty_{\uppt{\red{5}}}}}{
					\infer[\tylam]{\tjudgw{}{2\indet}{\la\var{\var\var}}{\arr{\mset{\arr{\mset{\initty_{\uppt{\red
													{13}}}}}{\initty_{\downpt{\blue
												{9}}}},\mset{\initty_{\downpt{\blue
												{16}}}}}}{\initty_{\uppt{\red
										{6}}}}}}{
						\infer[\tyapp]{\tjudgw{\var:\mset{\arr{\mset{\initty}}{\initty},\mset{\initty}}}{\indet}{\var\var}{\initty_{\uppt{\red
										{7}}}}}{
							\infer[\tyvar]{\tjudgw{\var:\mset{\arr{\mset{\initty}}{\initty}}}{\indet}{\var}{\arr{\mset{\initty_{\downpt{\blue
													{14}}}}}{\initty_{\uppt{\red
												{8}}}}}}{}
							& 
							\infer[\tymanys]{\tjudgw{\var:\mset{\mset\initty}}{\indet}{\var}{\mset{\initty}}}
							{\infer[\tyvar]{\tjudgw{\var:\mset\initty}{0}{\var}{\initty_{\uppt{\red
												{15}}}}}{}}}} &
					\infer[\tymanys]{\tjudgw{}{2\indet}{\la\vartwo\vartwo}{\mset{\arr{\mset{\initty}}
								{\initty},\mset{\initty}}}}{\infer[\tylam]{\tjudgw{}{\indet}
							{\la\vartwo\vartwo}{\arr{\mset{\initty_{\downpt{\blue
												{12}}}}}{\initty_{\uppt{\red
											{10}}}}}}{
							\infer[\tyvar]{\tjudgw{\vartwo:\mset{\initty}}{0}{\vartwo}{\initty_{\uppt{\red
											{11}}}}}{}} &
						\infer[\tymanys]{\tjudgw{}{\indet}{\la\vartwo\vartwo}{\mset{\initty}}}
						{\infer[\tylamstar]{\tjudgw{}{0}{\la\vartwo\vartwo}{\initty_{\uppt{\red
											{17}}}}}{}}}}}}
	}\]
	\caption{An example of \TIAM execution.}
	\vspace{-8pt}
	\label{fig:tiam-example}
\end{figure*}

\paragraph*{Transitions} The \TIAM starts on the final judgment of $\tyd$, with empty type context $\ltyctx = \ctxhole$ 
and direction $\uppt$. It moves from judgment to judgment, following occurrences of $\initty$ around $\tyd$. To specify 
the transitions, we use the leaf contexts defined in the previous section.

The 
transitions are in \reffig{tiam}, their union is noted $\totiam$. We now 
explain them one by one---the transitions have the 
 labels of \LIAM transitions, because they correspond to each other, as we shall show. 

Let's start with the simplest, $\tomachdottwo$. The state focuses on the 
conclusion judgment $\ruleoc$ of a $\tylam$ 
rule with direction $\uppt$. The eventual type environment $\tye$ is omitted because the transition does not depend 
on it---none of the transitions does, so type environments are omitted from all transitions. The judgment assigns  
type $\arr\tty\linty$ to $\la\var\tm$, and the type context is $\arr\tty\ltyctx$, that is, it selects an occurrence of 
$\initty$ in the target type $\linty = \ltyctxp\initty$. The transition then simply moves to the judgment above, 
stripping 
down the type context to $\ltyctx$, and keeping the same direction. Transition $\resm 4$ does the opposite move, in 
direction $\downpt$, and transitions $\resm 1$ and $\resm 3$ behave similarly on $\tyapp$ rules: 
the extra premise $\vdash$ simply denotes the right premise of the $\tyapp$ rule that is left unspecified since not 
relevant to the transition.

Transitions $\tomacharg$: the focus is on the left premise of a $\tyapp$ rule, of type $\arr\tty{\lintytwo}$ isolating 
$\initty$ 
inside the $i$-th leaf $\linty_i$ of the tree type $\tty$. The transition then moves the state of the \TIAM to the 
$i$-th leaf of the tree of $\tymany$ rules on the right premise, changing direction. Transition $\tomachbtone$ does the 
opposite move.

Transitions $\tomachvar$ and $\tomachbttwo$ are based on the axiom sequences property of  \reflemma{relevance}. 
Consider 
a $\tylam$ rule occurrence whose right-hand type of the conclusion is $\arr\tty{\lintytwo}$. The premise has shape 
$\tjudg{\tye, \var:\tty}{\tm}{\lintytwo}$, and by the lemma there is a bijection between the leaves in 
$\tty$ and the axioms on $\var$, respecting the order in $\leaves\tty$. The left side of $\tomachbttwo$ focuses on the 
$i$-th leaf $\linty_i$ of $\tty$ and the \TIAM moves to the judgment of the axiom corresponding to that type, which is 
exactly 
the $i$-th from left to right seeing the derivation as a tree where the children of nodes are ordered as in the typing 
rules. Transition $\tomachvar$ does the opposite move, which can always happen because the code is the type derivation 
of a closed term. 

The only typing rule not inducing a transition is $\tylamstar$. Accordingly, 
when the \TIAM reaches a $\tylamstar$ rule, 
it is in a final state. Exactly as the \LIAM, the \TIAM is bi-deterministic.

\begin{proposition}
	The \TIAM is bi-deterministic for each type derivation $\tyd\pof\tjudg{}{\tm}{\initty}$.
\end{proposition}

\paragraph*{An example} In \reffig{tiam-example} we present the very same example analyzed in 
Section~\ref{sec:IJK}. We have reported its type derivation, with the 
occurrences of $\initty$ on the right of $\vdash$ annotated with increasing 
integers and a direction. The occurrence of $\initty$ marked with 1 represents 
the first state, and so on.
One can immediately notice that every occurrence of $\initty$ is visited 
exactly once. Moreover, the sequence of the visited subterms is the same as the 
one obtained in the example of Section~\ref{sec:IJK}. Please note that 
judgments are decorated with weights (such as $\indet$), which shall be 
introduced only in \refsect{space}, in order to later provide an example of 
decoration---they can be safely ignored for now.


\section{The \LIAM and \TIAM Bisimulation}
\label{sect:bisimulation}
The aim of this section is to explain the strong bisimulation between the \TIAM and the 
\LIAM, that is essentially the same between the \SIAM and the \LIAM studied in 
\cite{ADLVPOPL21}. 
A striking point of the \TIAM is that it does not have the log nor the tape. 
They are encoded in the position of the
judgment occurrence $\ruleoc$ and in the type context $\ltyctx$ of its states, 
as we shall show.


\paragraph*{Relating Logs and Tapes with Typed Positions} In the \LIAM, the log 
$\tlog=\lpos_1\cons\ldots\cons\lpos_n$ has a logged position for every argument $\tmtwo_1, \ldots, \tmtwo_n$ in which 
the position of the current state is contained. The argument $\tmtwo_i$ is the 
answer to the query of an argument for the variable in the logged position 
$\lpos_i$. The \TIAM does not keep a trace of the variables for which it 
completed a query, but the answers to those (forgotten) queries are simply 
given by the sub-derivations for 
$\tmtwo_1, \ldots, \tmtwo_n$ in which the current judgment occurrence $\ruleoc$ 
is contained---the way in which 
$\lpos_k$ identifies a copy of $\tmtwo_k$ in the \LIAM corresponds on the type derivation $\pi$ to the index $i$ of the 
leaf (in the tree of 
sub-derivations) typing $\tmtwo_k$ in which $\ruleoc$ is located. Note that the \LIAM manipulates the 
log only via 
transitions $\tomacharg$ and $\tomachbtone$, that on the \TIAM correspond exactly to entering/exiting derivations for 
arguments.
The tape, instead, contains logged positions for which the \LIAM either has not yet found the associated argument, or 
it is backtracking to. Note that the \LIAM puts logged positions on the tape via transitions $\tomachvar$ and 
$\tomachbtone$, and removes them using $\tomacharg$ and $\tomachbttwo$. By 
looking at \reffig{tiam}, it is 
clear that there is a logged position on the \LIAM tape for every type 
sequence of the flattening of $\tty$ in which it lies the hole 
$\ctxhole$ of the current type context $\ltyctx$ of the \TIAM.

\paragraph*{Extracting \LIAM States} These ideas are used to extract from every \TIAM state $\state$ a \LIAM state $\extr\state$ in a quite technical 
way. In particular, the extraction process retrieves a log $\elog\state$ from 
the judgement $\ruleoc$ of $\state$ and a tape $\etape\state$ from the type 
context $\ltyctx$ of $\state$, using a sophisticated \emph{T-exhaustible 
invariant} of the \TIAM to retrieve the exact shape of the logged positions in 
$\elog\state$ and $\etape\state$. 

Let us give a high-level description of how extraction works. The invariant is 
based on the pairing of every \TIAM state $\state$ with a set of \emph{test 
states}, some coming from the judgment $\ruleoc$ of $\state$, called 
\emph{judgment tests}, and some coming from the type context $\ltyctx$, called 
\emph{type (context) tests}. The invariant guarantees a certain recursive 
property of each test state. The extraction process uses this property to 
extract a logged position $\lpos_{\statetwo}$ from each test state $\statetwo$ 
of $\state$. 

Given a \TIAM state $\state=(\tyd, \ruleoc, \ltyctx, \pol)$, its judgment tests 
are associated to the $\tyapp$ rules having $\ruleoc$ in their right 
sub-derivation. Their extractions give logged positions 
$\lpos_1\cdot\mydots\cdot\lpos_n$ forming the extracted log $\elog\state$, 
following the correspondence described above.

Type tests are associated to the leaf contexts surrounding the hole of $\ltyctx$. The extraction of the tape $\etape\state$ from $\ltyctx$ is done according to the following schema:
\[\begin{array}{lcl}
			\etape{\ctxhole} &\defeq &\stempty 
			\\
			\etape{\arr\tty{\ltyctx}} & \defeq & \resm\cdot 
			\etape{\ltyctx}
			\\
			\etape{\arr{\leafctxp\ltyctx}\lintytwo}
			& \defeq & \elpos{\state_{\leafctx}}\cdot\etape{\ltyctx}
		\end{array}
		\]
where $\state_{\leafctx}$ is the state test associated to the leaf context $\leafctx$.

For lack of space, and because this is essentially identical to what is done in \cite{ADLVPOPL21} for the \SIAM, 
the technical 
development is in 
\begin{SHORT}~\cite{TRPOPL2021}\end{SHORT}\begin{LONG}\refapp{bisimulation}\end{LONG}.
 The 
extraction process induces a relation $\state \bisimtypes \extr\state$ that is easily proved to be a strong 
bisimulation between the \TIAM and the \LIAM.

\begin{proposition}[\TIAM and \LIAM bisimulation]
\label{prop:str-bisim-typed}
Let $\tm$ a closed term and $\tyd\pof\tjudg{}{\tm}{\initty}$ a tree type derivation. Then 
$\bisimtypes$ is a strong bisimulation between \TIAM states on $\tyd$ and \LIAM 
states on $\tm$. Moreover, if $\state_\tyd \bisimtypes \state_\l$ then 
$\state_\tyd$ is 
\TIAM reachable if and only if $\state_\l$ is 
\LIAM reachable.
\end{proposition}

The \emph{moreover} part of the above statement hints at a bijection between 
\emph{all} the states in $\tyd$ and reachable \LIAM states. However, there 
still could be the possibility that some of the states in $\tyd$ are not 
reachable. This is actually \emph{not} the case, and the technical development 
is in the Appendix.
\begin{proposition}
	Let $\tm$ a closed term and $\tyd\pof\tjudg{}{\tm}{\initty}$ a tree type 
	derivation. Then every state of $\tyd$ is reached exactly once.
\end{proposition}

\begin{figure*}[t!]
	\[\small
	\begin{array}{c@{\hspace{1cm}}c@{\hspace{1cm}}c}
	\infer[\tyvar]{\wtjudg{\var:\mset{\linty}}{\ssize{\linty}}{\var}{\linty}}{}
	&

\infer[\tylam]{\wtjudg{\tye}{\max\{w,\,\ssize{\arr{\tty}{\linty}}\}}{\lambda\var.\tm}{\arr{\tty}{\linty}}}{\wtjudg{\tye,
\var:\tty}{w}{\tm}{\linty}}
	&
\infer[\tyapp]{\wtjudg{\tye\uplus\tyetwo}{\max\{w,v\}}{\tm\tmtwo}{\linty
}}{\wtjudg{\tye}{w}{\tm}{\arr{\tty}{\linty}} & 
\wtjudg{\tyetwo}{v}{\tmtwo}{\tty}}
	\\[8pt]
	\infer[\tylamstar]{\wtjudg{\tye}{0}{\lambda\var.\tm}{\initty}}{}
	&
	\infer[\tymany]{\wtjudg{\mset{\uplus_{i=1}^n\tye_i}}{\indet+\max_{i}\set{w_i}}
	 {\tm}{\mset{\gty_1,\mydots,\gty_n}
}}
{\wtjudg{\tye_i}{w_i}{\tm}{\gty_i} & 1\leq i\leq n}
	& \infer[\tynone]{\wtjudg{}{0}{\tm}{\emmset}}{}
	\end{array}
	\]
	\vspace{-8pt}
	\caption{The tree type system with weights.}
	\label{fig:space-weigths}
	\end{figure*}
\section{Measuring the Space of Interactions}\label{sect:space}
This section contains the main contribution of the article: it gives a way of 
measuring the space consumed by the complete \LIAM run on the term $\tm$ via a 
quantitative analysis of the tree type derivation for $\tm$. We proceed in two 
steps. \begin{enumerate}
	\item \emph{The space of single extracted states}: given a \TIAM state $\state=(\tyd, \ruleoc, \ltyctx, \pol)$, we 
show how to measure the space of the extracted \LIAM state $\extr\state$ from $\tyd$, $\ruleoc$, and $\ltyctx$.
	\item \emph{The space of the whole execution}: we refine tree type 
	derivations adding \emph{weights} on judgments, 
and show that the weight of the final judgment coincides with the maximum space 
consumption over all extracted states, 
that is, along the whole \LIAM execution.
\end{enumerate}
\paragraph*{The Undetermined Pointers Size $\indet$} A technical point common to both parts is that the quantitative 
study of tree types derivations is relative to an undetermined value $\indet$. The reason for using $\indet$ is that 
our space analyses have both local and global components. Locally, we count how many occurrences of $\resm$ and how 
many 
logged positions are involved in a state (for step 1) or in all states in and 
above that judgment (for step 2). The 
global component comes from the fact that all logged positions of the \LIAM, independently of where they arise, are 
implemented via pointers to the \emph{global} code. Essentially, $\indet$ is meant to be replaced, at the very end of 
both analyses, by the size of pointers to the \LIAM global code, that is, by $\log\size\tm$, where $\tm$ is the term 
typed in the final judgment of the type derivation $\tyd$. Therefore, locally 
our measures shall include $\indet$, 
which shall substituted at the end with $\log\size\tm$.

\subsection{The Space of Single Extracted States}
\paragraph*{Trees and the Size of Extracted Logged Positions} Basically, given a \TIAM state $\state=(\tyd, \ruleoc, 
\ltyctx, \pol)$, the size of logged positions in $\extr\state$ is obtained by counting $\indet$ for 
\begin{itemize}
	\item \emph{Extracted tape}: every sequence constructor $\mset\cdot$ surrounding the hole $\ctxhole$ in $\ltyctx$;
	\item \emph{Extracted log}: every $\tymany$ rule on the path from $\ruleoc$ 
	to the final judgment of $\tyd$.
\end{itemize}
Clearly, it is the newly introduced tree structure that allows to measure the size of extracted logged positions, as 
expected.

First, we define a size of type contexts meant to measure the size of the extracted tapes.
\begin{definition}[Branch size of type contexts/extracted tapes]
	Let $\state=(\tyd, \ruleoc, \ltyctx, \pol)$ be a \TIAM state. The branch size $\ttm{\cdot}$ for type contexts is 
defined as follows:
	\[\begin{array}{rcl@{\hspace{.5cm}} rcl}
		\ttm{\ctxhole}&\defeq&0 & 
		\ttm{\arr\tty\ltyctx}&\defeq&1+ \ttm{\ltyctx}\\[2pt]
		\ttm{\arr{\ttyctx}\linty}&\defeq&\ttm{\ttyctx} &
		\ttm{\mset{\gty_1,\mydots,\gtyctx,\mydots,\gty_n}}&\defeq&\indet+
		\ttm{\gtyctx}
	\end{array}\]
\end{definition}
Let us interpret the branch size with respect to the tape extraction schema of 
the previous section. The $+1$ 
in the 
clause for $\arr\tty\ltyctx$ is there to count $\resm$. The clause for sequences instead gives $\size{\leafctx} = 
n\cdot\indet$ if the hole has height $n$ in the leaf context $\leafctx$ seen as a tree---whence the name \emph{branch 
size}. 

Then, we define a branch size for judgments, meant to measure the size of  
extracted logs, and a branch size for states.
\begin{definition}
	Let $\state=(\tyd, \ruleoc, \ltyctx, \pol)$ be a \TIAM state.
	\begin{itemize}
	\item \emph{Branch size of judgments/extracted logs}: let $n$ be 
	the number of $\tymany$ rules encountered descending from $\ruleoc$ to the 
	root of $\tyd$. Then $\tlm{\ruleoc} \defeq n\cdot\indet$.
	\item \emph{Branch size of \TIAM states}: $\bsize\state \defeq \ttm\ltyctx + \tlm{\ruleoc}$.
	\end{itemize}
\end{definition}

We prove that the defined branch sizes do correspond to their intended meanings, that is, the branch sizes of 
extracted logs and tapes, showing that the size of \TIAM states captures the space size of the extracted \LIAM state.
\begin{proposition}[Space of Single Extracted States]
\label{prop:space-single-states}
	Let $\state=(\tyd, \ruleoc, \ltyctx, \pol)$ be a reachable \TIAM state. 
	Then $\ttm{\ltyctx}=\lm{\etape\state}$ and 
	$\tlm{\ruleoc}=\lm{\elog\state}$, and thus $\bsize\state = \lm{\extr\state}$. Moreover,
	\begin{enumerate}
		\item if $\elog\state=\lpos_1\mydots\lpos_n$, and let $h_i$ be the 
		number 
		of $\tymany$ rules of the $i^{th}$ 
	$\tymany$ rule tree found descending from $\ruleoc$ to the root of $\tyd$, 
	then $\lm{\lpos_i}=h_i$;
	\item for each extracted tape position $\lpos$, \ie for each $\gtyctx$ such 
	that $\ltyctx=\gtyctxp{\arr{\leafctxp{\ltyctxtwop{\initty}}}\linty}$, then 
	$\lm\lpos=\size\leafctx\cons\indet$.
	\end{enumerate}
\end{proposition}

\paragraph*{The Need for Tree Types} A subtle point is that the tree structure of types is not needed in order to define 
the extraction process---indeed, \LIAM states are extracted from \emph{sequence} type derivations in \cite{ADLVPOPL21}. 
Extraction is an indirect process---a sort of logical relation---whose functioning is guaranteed by an invariant (the 
T-exhaustible invariant in the Appendix). The process
does not describe explicitly the shape of the extracted logged positions, it only guarantees that adequate logged 
positions exist. Without tree types, the structure of multi types derivations somehow encodes enough information to 
retrieve $\extr\state$, but how many logged positions are involved can be discovered only by unfolding the 
whole extraction process, the information is not encoded into the types themselves. 

There is a further subtlety. Tree types trace the \emph{number} of pointers as precisely described by the 
\emph{moreover} part of \refprop{space-single-states}, but do not describe the internal structure of 
logged positions. Given a \TIAM state $\state$, we can easily know the length of  
$\elog\state$ and $\etape\state$, and know the number of pointers to implement each logged position $\lpos$ in them, 
which is enough to measure space. The internal structure of $\lpos$, however, cannot be read from tree types. Again, it 
is determined only by unfolding the whole extraction process.

\subsection{The Space of the Whole Execution}
\paragraph*{Type Weights} To obtain the space cost of the whole execution we endow tree types derivations with 
\emph{weights}\footnote{We introduce a different word for measuring space of 
the whole execution, because judgments 
are measured in two different ways: the \emph{branch size} measures what is 
below the judgment, and corresponds to the 
size of the extracted log for a state, the \emph{weight} measures what is above 
a judgment, and gives the maximum space 
over all states in the rooted sub-derivation.}. 
In turn, we first 
have to define a notion of weight for types. The intuition is that we are taking the max of the branch size for 
type contexts 
$\gtyctx$ used above, over all the ways of writing a type $\gty$ as 
$\gtyctxp\initty$, as confirmed by the associated lemma.
\[\begin{array}{c@{\hspace{.9cm}} c}
	\ssize{\initty}\defeq0
	&
	\ssize{\arr{\tty}{\linty}}\defeq\max\{\ssize{\tty},\ssize{\linty}+1\}
	\\[7pt]
	\multicolumn{2}{c}{
	\ssize{\mset{\gty_1,\mydots,\gty_n}}\defeq 
	\indet+\max_{i}\{\ssize{\gty_i}\}
	}
	\end{array}\]
\begin{lemma}
\label{l:weightthree}
	Let $\gty$ be a type. Then 
	$\ssize\gty=\max_{\gtyctx|\gty=\gtyctxp\initty}\ttm\gtyctx$.
\end{lemma}
Note that, via the space of single extracted states (\refprop{space-single-states}), the previous lemma states that the 
size of $\gty$ is the maximum space of all the tapes extracted from \TIAM 
states over a same judgment 
$\tjudg{\tye}{\tm}{\gty}$.

\paragraph*{Judgements and Derivations Weights} Weights are extended to 
judgments in \reffig{space-weigths}, and the 
weight of a derivation is the weight of its final judgment. The idea is that 
the weight $w$ of a 
weighted judgment $\ruleoc = \wtjudg{\tye}{w}{\tm}{\gty}$ gives the maximum 
space of all the states over $\ruleoc$ 
and---crucially---above $\ruleoc$.

Now, we prove that the weight of a judgment $\ruleoc$ is greater than the maximum size of 
the tape of the states in its derivation.

\begin{lemma}[Judgment weights bound extracted tapes]
\label{l:weightfour}
	Let $\tyd:\wtjudg{\tye}{w}{\tm}{\gty}$ be a weighted derivation and $\mathcal{J}$ be the set of all the judgments  
	occurring in $\tyd$. 
	Then $w\geq\max_{\tjudg{\tyetwo}{\tmtwo}{\gtytwo}\in\mathcal{J}}\ssize{\gtytwo}$.
\end{lemma}

%

Judgement weights actually take also logs into account.
\begin{lemma}[Weights bound also extracted logs]
\label{l:weight}
	Let $\tyd\pof\wtjudg{\tye}{w}{\tm}{\gty}$ be a weighted derivation. Then 
	$w\geq v+\tlm{\ruleoc}$ for every weighted 
judgment $\ruleoc$ $\wtjudgone{v}$ in $\tyd$.
\end{lemma}

We then obtain that the weight of a derivation $\tyd$ for $\tm$ bounds the space used by the \TIAM execution of $\tyd$, 
and so by the \LIAM execution of $\tm$. 
\begin{theorem}[\LIAM space bounds]
	Let $\tyd\pof\wtjudg{}{\weight}{\tm}{\initty}$ be a weighted tree types derivation. Then $\lm{\extr\state} \leq w$ 
for every $\state\in \states\tyd$.
\end{theorem}

Last, we show that weights provide \emph{exact} bounds, as there always is a witness state 
using as much space as in the weight.
\begin{proposition}[Weight witness]
\label{prop:witness}
Let $\tyd\pof\wtjudg{\tye}{w}{\tm}{\gty}$ be a weighted derivation and 
$\gty\neq\emmset$. Then there exists a 
	\TIAM state $\state$ over $\tyd$ such that 	$w=\bsize\state$.
\end{proposition}

We can then conclude our complexity analysis.

\begin{corollary}[\LIAM exact bound via tree types derivations]
	Let $\tyd\pof\wtjudg{}{\weight}{\tm}{\initty}$ be a tree types derivation and $\run$
	the complete \LIAM run on $\tm$. Then $\spacem{\run}=\weight$.
\end{corollary}

Now, the reader can fully understand and appreciate the weights in the derivation of 
\reffig{tiam-example}. Please note that we have considered 
$\max\{1,\indet\}=\indet$ when assigning the weights.

\section{On the \LIAM Space (In)Efficiency}\label{sect:inefficiency}
\paragraph*{Tree Density} As we have proved in the last section, the space 
consumed by a 
\LIAM run ultimately depends on the level of nesting of tree types. Let us 
be slightly more precise. When we are dealing with multi types, the 
argument $\tmtwo$ of an application $\tm\tmtwo$ has to be typed with a 
multiset, say $M \defeq\mset{\initty,\mydots,\initty}$ where $\size{M}=n$ and $n$ is the number of times that $\tmtwo$ shall be copied. 
In our type 
system, $M$ can be represented as a tree $\tty$, capturing the space needed for evaluating the copies of $\tmtwo$, in several ways. In the tree type derivation for $\tm\tmtwo$ only one of 
these representations is used, and depends on the type of
$\tm$. The important point is that different representations lead to very different space consumptions. The ideal case is the flat representation
$\tty_f\defeq\mset{\initty,\mydots,\initty}$, for which
$\ssize{\tty_f}=\indet$. Another good case is given by a full binary tree 
$\tty_b$, or, more generally, by a tree $\tty_d$ whose height is logarithmic in 
$n$, so that $\ssize{\tty_d}=\log n\cdot\indet$. Let us deem this case 
\emph{dense}. A bad case is given by linearly shaped trees such as 
$\tty_l\defeq\tty_n$, where $\tty_0\defeq\mset{\initty}$ and 
$\tty_{n+1}\defeq\mset{\initty,\tty_n}$, for which 
$\ssize{\tty_s}=(n+1)\cdot\indet\geq\ssize{\tty_d}$.
 We call \emph{sparse} a tree with $n$ leaves and height $n$ like $\tty_l$. Not that there can even be worse, as in general a tree can have an arbitrary number of internal nodes, for instance a multi set $\mset\initty$ can be represented as a tree also as $\mset{\mset{\mset{\mset{\initty}}}}$---again, it depends on the type of $\tm$. Intuitively, if a 
 term can be typed with flat or dense trees, the \LIAM evaluates it space 
 \emph{efficiently}, 
 while in the case of sparse trees (or worse) the evaluation is space \emph{inefficient}.
 \paragraph*{Fixed-Points Are Sparse} As is well known, the $\lambda$-calculus 
 is a universal (or Turing
 complete) model of computation. The fundamental ingredient that allows one
 to achieve universality is the presence of fixed-point
 combinators. These combinators allow for an encoding of 
 general recursion, as needed when simulating, e.g.,  partial recursive
 functions or Turing machines.  In
 particular, the fixed point combinator can capture unbounded
 iteration (i.e. \texttt{while} loops) or tail-recursion, as needed
 \eg in the encoding
 of minimization from Kleene algebra.
 
As an example, consider how the encoding of a Turing machine $M$ may
look like, in the $\l$-calculus. Let $\tm_M$ be the encoding of
$M$'s transition function (which is typically very simple if states
and tapes are encoded using, e.g., Scott's numerals~\cite{Wadsworth}). Then, the
(recursively defined) function that iterates $\tm_M$, thus
capturing the overall behavior of $M$, can be written as follows:
\[
\textsf{iter}=\underbrace{(\la f{\la s \texttt{if}\text{ }s\text{ is final } 
		\texttt{then}\text{ 
			halt }\texttt{else} \text{ 
			}f(\tm_Ms)})}_{\textsf{iteraux}}\textsf{iter}
\]
How can we build a solution to this equation in the form of a
$\l$-term? Apart from an encoding of the conditional operator, itself
very easy to write, we need a fixed-point combinator $\Theta$, such as 
Turing's: 
\begin{center}$\Theta\defeq\theta\theta\qquad\qquad\text{where 
	}\quad\theta\defeq\la\var{\la\vartwo{\vartwo(\var\var\vartwo)}}$\end{center}
We highlight that 
$\Theta\tm\towh(\la\vartwo{\vartwo(\Theta\vartwo)})\tm\towh\tm(\Theta\tm)$. 
Then, we can set, as expected, $\textsf{iter}\defeq\Theta\textsf{iteraux}$.

Let us analyze how $\Theta$ implements recursion, independently
on what the argument of $\Theta$ is.  Please note that
during the reduction of $\Theta\tm$, the variable $\vartwo$ is
substituted for the term $\tm$, which after two $\beta$-steps appears
twice, once in head position (call this occurrence $\tm_0$), and once
applied to $\Theta$. The latter copy of $\tm$, together with $\Theta$,
can be copied potentially many times, depending on how
$\tm_0$ uses its argument. Some of these copies, say $m$, will
eventually appear in head position, and the same process starts
again. In other words, the copies of $\tm$ that the combinator $\Theta$
will eventually create can be organized in a tree, see
Figure~\ref{fig:firstfigure}. This is a faithful description of
how recursion unfolds, independently on how $\tm$ uses its
argument.

If $\tm=\textsf{iteraux}$, however, the situation is much
simpler: $\tm$ uses its argument at most once, and the complicated
tree in Figure~\ref{fig:firstfigure} becomes the one in
Figure~\ref{fig:secondfigure}. Every copy $\tm_{01^n}$ of $\tm$ either
brings $\Theta\tm_{1^{n+1}}$ in head position (without copying it), or
discards it, depending on whether the current state is final or not.
Saying it another way, the height of the tree
in Figure~\ref{fig:secondfigure} is nothing more than the number of reduction
steps the Turing machine $M$ performs.
\begin{figure}
	\fbox{
		\begin{minipage}{.48\textwidth}
			\begin{center}
				\begin{subfigure}[b]{0.24\textwidth}\includegraphics{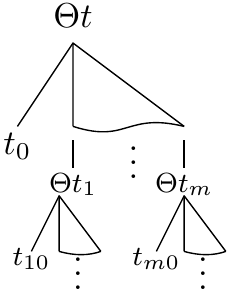}\caption{}\label{fig:firstfigure}\end{subfigure}
				\qquad\qquad
				\begin{subfigure}[b]{0.24\textwidth}\includegraphics{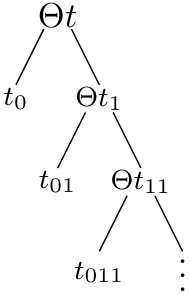}\caption{}\label{fig:secondfigure}\end{subfigure}
			\end{center}
		\end{minipage}}
		\caption{Different ways $\Theta$ can copy its argument.}
	\end{figure}
 
In order to understand if the \LIAM could be reasonable in space, \ie if it can 
simulate Turing machines with a constant overhead in space, we apply our 
technique by giving $\Theta$ a suitable type (schema) $\tyF_n^{\tyl}$, 
depending on a list of types 
$\tyl\defeq\linty_n,\mydots,\linty_0$ and on $n$, which is the number of times 
the fixed point is unfolded. In particular, 
$\tyF_n^{\tyl}$ turns out to be
 \emph{sparse}. The details of the technical development can be 
 found in the Appendix.
 \begin{proposition}
 	For each $n\geq 0$, and for each list of types $\tyl$ such that 
 	$\size\tyl\geq n+1$, $\wtjudg{}{\geq (2n+1)\indet}{\Theta}{\tyF_n^{\tyl}}$.
 \end{proposition}
 In particular, $\Theta:\tyF_n^{\tyl}$ when $\Theta$ is unfolded $n$ times and 
 thus, inside the encoding of a Turing machine $M$ which takes $n$ steps to 
 halt. This way the \LIAM, independently of the space
 consumption of $M$, requires space at least linear in the number of
 reduction steps it performs, and there is thus no hope to stay within
 sub-linear space constraints.
 
 \paragraph*{(In)Efficiency in Related Works} How could we reconcile all this 
 with the claims 
 from the literature about
 the existence of sub-linear space bounds 
 (Mazza \cite{DBLP:conf/csl/Mazza15} and Dal Lago and Sch\"opp \cite{DBLP:conf/esop/LagoS10,dal_lago_computation_2016}) 
 within the realm of geometry of 
 interaction
 machines? The answer is relatively simple: the kind of machines considered
 in the cited works are \emph{fundamentally different} than ours,
 being based on the idea that the information stored in logged positions
 can be taken to be a natural number \emph{smaller or equal to} the
 cardinality of the underlying multi type. In Mazza 
 \cite{DBLP:conf/csl/Mazza15}, 
 this is possible due to the
 peculiarities of the underlying type system. There, the use of non-linearity
 is much more restricted than in the $\l$-calculus, being based on parsimonious types. This 
 allows for tail-recursive schemas \emph{only}: in our terminology, all trees 
 are linearly shaped, and then logged positions can be represented differently and with less space, namely by simply taking  
the height of the tree. The works 
 by Dal Lago and 
 Sch\"opp~\cite{DBLP:conf/esop/LagoS10,dal_lago_computation_2016} rely on the 
 fact that the underlying type
 system is resource-aware, this way allowing for a different representation of 
 logged positions as natural numbers rather than trees. Moreover,
 space-efficient simulations are done via an ad-hoc combinator for skewed 
 iterations, thus
 circumventing the problem with fixed-point operators. This, unfortunately, is 
 not available in the
 realm of the $\l$-calculus. 

 A question, however, remains. Is it possible \emph{at all} that the
\LIAM\ consumes an amount of computation \emph{space} significantly smaller
than computation \emph{time}? The answer is positive: there is a family
of terms $\{\tm_i\}_{i\in\mathbb{N}}$ such that $\tm_i$ reduces in \ccbn to normal form in time
exponential in $i$ (namely, taking an exponential number of $\beta$-steps) but requires space linear in $i$, when reduced
by the \LIAM. Details are in the Appendix.

\section{Conclusions}
Space efficiency of higher-order languages accommodating sub-linear complexity 
is a topic about which almost nothing is known. 
The literature has suggested that the right tool to achieve it is the alternative paradigm of the \IAM, a machine rooted in the geometry of 
interaction, without however clarifying whether its use of space can be used as a reasonable space cost model.

In this paper we develop a sharp tool---a type system---for the understanding of the space consumption of the ($\lambda$)\IAM. Our new tree intersection types provide---for the first time---exact space bounds, via a simple system of weights measuring the depth of the tree structure in the types. 

The tree structure seems to naturally complement the multiset one needed for 
measuring time, and it is of independent interest, given the relevance of 
intersection types in semantics. For instance, can such a structure be seen in game semantics? What does it measure with respect to cut-elimination or environment machines?

Beyond the theoretical result, the type system has a direct application, as we 
show by studying the space usage of the \IAM on the traditional encoding of 
Turing machines. Such a usage turns out to be very inefficient, providing 
negative insights on the elusive conjecture about the reasonable space usage of 
the \IAM.\bigskip

\noindent\textbf{Acknowledgements.}
The second author is funded by the ERC CoG ``DIAPASoN'' (GA
  818616). This work has been partially funded by the ANR JCJC grant
  ``COCA HOLA'' (ANR-16-CE40-004-01).

\bibliographystyle{IEEEtranS}
\bibliography{main}
\onecolumn
\pagebreak
\appendices

\section{Proof of Proposition~\ref{prop:inf}}
A key property of the \LIAM is bi-determinism, or
reversibility: the machine is deterministic, and moreover for each state $s$ 
there is at most one state $s'$ such that $s'\toliam s$. 
The property follows by simply inspecting the rules. Moreover, a run can be 
reverted by just switching the direction. 
\begin{proposition}[Reversibility]
	If $\nopolstate{\tm}{\ctx}{\tape}{\tlog}{\pol}\toliam
	\nopolstate{\tmtwo}{\ctxtwo}{\tapetwo}{\tlogtwo}{\poltwo}$, then
	$\nopolstate{\tmtwo}{\ctxtwo}{\tapetwo}{\tlogtwo}{\poltwo^1}\toliam
	\nopolstate{\tm}{\ctx}{\tape}{\tlog}{\pol^1}$.
\end{proposition}

From bi-determinism it is immediate to prove acyclicity for reachable states.
\begin{proposition}
	Let $\state$ be a reachable \LIAM state. Then $\state$ is reached exactly once.
\end{proposition}
\begin{proof}
	We proceed by induction on the length of the run 
	$\run:\state_0\toliam^*\state$. If $\size\run=0$, the result is trivial. 
	Otherwise, if $\size\run>0$, we have $\run:\state_0\toliam^*\statetwo 
	\toliam\state$. Let us call $\runtwo:\state_0\toliam^*\statetwo$. By \ih, 
	every state in $\runtwo$ is reached exactly once. Then $\state$ cannot be 
	part of $\runtwo$, because otherwise it would have two different 
	predecessors, contradicting bi-determinism. Thus $\state$ is reached 
	exactly once in $\run$.
\end{proof}
Then, since there are no cycles, an infinite run goes through infinite different states and thus consumes unbounded space.
\begin{proposition}
	Let $\run$ be an infinite \LIAM run. Then $\lm{\run}=\infty$.
\end{proposition}
\begin{proof}
	The result comes from the fact that in finite amount of memory, only a 
	finite amount of configurations can be encoded. Since in an infinite run, 
	an infinite number of \emph{different states} are reached, then the \LIAM 
	needs an unbounded space to perform the computation.
\end{proof}
\section{Correctness and Completeness of the Tree Type System}
The size $\size{\tyd}$ of a tree 
types derivation $\tyd$ is the number of its rules 
that 
are not $\ttyprom$. It is the quantity that is used to prove the termination argument for typed terms.

\paragraph*{Correctness} In order to prove that typability implies termination via a simple combinatorial argument, we 
need to refine the standard statements of the substitution lemma and of subject reduction with quantitative information.

The next lemma is used in the substitution lemma (namely, the implication from left to right), and shall also be used 
in the anti-substitution lemma (the converse implication). 
\begin{lemma}[Tree splitting and merging]
\label{l:tree-split-merge}
 Let $\tty = \tty_1 \uplus\mydots\uplus\tty_k$. Then there exists $\tyd \pof \tjudg{\mset\tye}{\tm}{\tty}$ if and only if 
there exist $\tyd_i \pof \tjudg{\mset{\tye_i}}{\tm}{\tty_i}$ for $i\in \set{1,\mydots, k}$. Moreover, $\mset\tye = 
\mset{\tye_1 \uplus\mydots\uplus\tye_k}$ and $\size\tyd = \sum_{i=1}^n \size{\tyd_i}$.
\end{lemma}

\begin{proof}
 We prove the statement by first examining the rule $\tymany$, which is the last rule used in $\tyd$, as $\tm$ is typed with a tree type.
 \[
 \infer[\tymany]{\tjudg{\mset{\uplus_{i=1}^n\tye_i}}{\tm}{\mset{\gty_1,\mydots,\gty_n}=\tty_1 \uplus\mydots\uplus\tty_k}}
 {\tjudg{\tye_i}{\tm}{\gty_i} & 1\leq i\leq n}
 \]
 We can prove the statement considering $k$ derivations, each of them deriving the judgment $\tm:\tty_j$.
 \[
 \infer[\tymany]{\tjudg{\mset{\uplus_{i=1}^{\size{\tty_j}}\tye_i}}{\tm}{\mset{\gty_1,\mydots,\gty_{\size{\tty_j}}}=\tty_j }}
 {\tjudg{\tye_i}{\tm}{\gty_i} & 1\leq i\leq \size{\tty_j}}
 \]
\end{proof}

\begin{lemma}[Quantitative substitution]
\label{l:q-sub-lemma}
	Let 
	$\tyd_\tm\pof\tjudg{\tye,\var:\tty}{\tm}{\gty}$
	and 
	$\tyd_\tmtwo\pof\tjudg{}{\tmtwo}{\tty}$. Then 
	there exists 
	$\tyd_{\tm\isub\var\tmtwo}\pof\tjudg{\tye}{\tm\isub\var\tmtwo}{\linty}$ such that $\size{\tyd_{\tm\isub\var\tmtwo}} = 
\size{\tyd_{\tm}} + \size{\tyd_{\tmtwo}} -\size{\leaves\tty}$.
\end{lemma}
\begin{proof}
	By induction on the derivation $\tyd_\tm$.
	\begin{itemize}
		\item \emph{Rule $\tyvar$}. Two sub-cases:
		\begin{enumerate}
		 \item $\tm=\var$: then $\tm\isub\var\tmtwo = \tmtwo$, $\gty = \linty$, and $\tty = 
\mset \linty$ is a singleton. Then the hypothesis $\tyd_\tmtwo\pof\tjudg{}{\tmtwo}{\tty}$ is necessarily obtained by 
applying a unary $\tymany$ rule to a derivation of the form $\tydtwo_\tmtwo \pof\tjudg{}{\tmtwo}{\linty}$. The typing 
derivation $\tyd_{\tm\isub\var\tmtwo} \defeq \tydtwo_\tmtwo$ satisfies the statement because
$\size{\tyd_{\tm\isub\var\tmtwo}} = \size{\tydtwo_\tmtwo} = 1 + \size{\tydtwo_\tmtwo} -1 =
\size{\tyd_{\tm}} + \size{\tyd_{\tmtwo}} -\size{\leaves\tty}$.

		 \item $\tm=\vartwo$: then $\tm\isub\var\tmtwo = \vartwo$ and $\tty = \emmset$. Then the hypothesis 
$\tyd_\tmtwo\pof\tjudg{}{\tmtwo}{\tty}$ is necessarily obtained by 
applying a $\tynone$ rule, and $\size{\tyd_\tmtwo} = 0$. The typing 
derivation $\tyd_{\tm\isub\var\tmtwo} \defeq \tyd_\tm$ satisfies the statement because
$\size{\tyd_{\tm\isub\var\tmtwo}} = \size{\tyd_\tm} = \size{\tyd_\tm} + 0 -0 =
\size{\tyd_{\tm}} + \size{\tyd_{\tmtwo}} -\size{\leaves\tty}$.
		\end{enumerate}

		\item \emph{Rule $\tylamstar$}. Then $\tm = \la\vartwo\tmthree$, $\gty = \initty$, and $\tty = \emmset$. It goes as 
for the second 
variable case ($\tm = \vartwo$). Namely, the hypothesis 
$\tyd_\tmtwo\pof\tjudg{}{\tmtwo}{\tty}$ is necessarily obtained by 
applying a $\tynone$ rule, and $\size{\tyd_\tmtwo} = 0$. The typing 
derivation $\tyd_{\tm\isub\var\tmtwo}$ is also a single $\tylamstar$ rule, because $\tm\isub\var\tmtwo = 
(\la\vartwo\tmthree)\isub\var\tmtwo  = \la\vartwo\tmthree\isub\var\tmtwo$ is also an abstraction. Note that 
$\tyd_{\tm\isub\var\tmtwo}$  satisfies the statement because
$\size{\tyd_{\tm\isub\var\tmtwo}} = \size{\tyd_\tm} = \size{\tyd_\tm} + 0 -0 =
\size{\tyd_{\tm}} + \size{\tyd_{\tmtwo}} -\size{\leaves\tty}$.

 		\item \emph{Rule $\tylam$}. Then $\tyd_\tm$ has the following shape:
 		\[\infer*{\infer[\tylam]{  \tjudg{\tye,\var:\tty}
				{\lambda\vartwo.\tmthree}{\arr{\ttytwo}{\lintytwo}}}
			{\tyd_\tmthree 
\pof \tjudg{\tye,\var:\tty,\vartwo:\ttytwo}
				{\tmthree}{\lintytwo}}}{}
		\]
		with 
 		$\tm=\la\vartwo\tmthree$ and $\gty=\arr{\ttytwo}{\lintytwo}$. By \ih, there exists a derivation 
$\tyd_{\tmthree\isub\var\tmtwo}\pof\tjudg{\tye,\vartwo:\ttytwo}{\tmthree\isub\var\tmtwo}{\lintytwo}$ such that 
$\size{\tyd_{\tmthree\isub\var\tmtwo}} = 
\size{\tyd_{\tmthree}} + \size{\tyd_{\tmtwo}} -\size{\leaves\tty}$. Applying back rule $\tylam$ we obtain 
$\tyd_{(\la\vartwo\tmthree)\isub\var\tmtwo}$
		\[\infer*{\infer[\tylam]{  \tjudg{\tye}
				{\la\vartwo\tmthree\isub\var\tmtwo}{\arr{\ttytwo}{\lintytwo}}}
			{\tyd_{\tmthree\isub\var\tmtwo}\pof\tjudg{\tye,\vartwo:\ttytwo}{\tmthree\isub\var\tmtwo}{\lintytwo}}}{}
		\]
		which satisfies $\size{\tyd_{(\la\vartwo\tmthree)\isub\var\tmtwo}} = \size{\tyd_{\tmthree\isub\var\tmtwo}} +1 
=_{\ih} 
\size{\tyd_{\tmthree}} + \size{\tyd_{\tmtwo}} -\size{\leaves\tty} +1 = \size{\tyd_{\la\vartwo\tmthree}} + \size{\tyd_{\tmtwo}} 
-\size{\leaves\tty}$.

		\item \emph{Rule $\tyapp$}. Then $\tyd_\tm$ has the following shape:
		\[
		\infer[\tyapp]{
      \tjudg{\var:\tty_1\uplus \tty_2,\tye\uplus\tyetwo}{\tmthree\tmfour}{\linty}
		}{
      \tyd_\tmthree \pof \infer*{\tjudg{\var:\tty_1,\tye}{\tmthree}{\arr{\ttytwo}{\linty}}}{}
			& 
			\tyd_\tmfour \pof \infer*{\tjudg{\var:\tty_2,\tyetwo}{\tmfour}{\ttytwo}}{}
    }
		\]
		with $\tm=\tmthree\tmfour$, $\gty = \linty$, and $\tty = \tty_1\uplus\tty_2$. By \reflemma{tree-split-merge}, the 
hypothesis $\tyd_\tmtwo\pof\tjudg{}{\tmtwo}{\tty}$ splits into two derivations 
$\tyd_\tmtwo^1\pof\tjudg{}{\tmtwo}{\tty_1}$ and $\tyd_\tmtwo^2 \pof\tjudg{}{\tmtwo}{\tty_2}$ such that 
$\size{\tyd_\tmtwo} = \size{\tyd_\tmtwo^1} + \size{\tyd_\tmtwo^2}$. By \ih, there exist:
		\begin{enumerate}
		 \item $\tyd_{\tmthree\isub\var\tmtwo} \pof \tjudg{\tye}{\tmthree\isub\var\tmtwo}{\arr{\ttytwo}{\linty}}$ 
such that $\size{\tyd_{\tmthree\isub\var\tmtwo}} = 
\size{\tyd_{\tmthree}} + \size{\tyd_{\tmtwo}^1} -\size{\leaves\tty_1}$, and 
		 \item $\tyd_{\tmfour\isub\var\tmtwo} \pof \tjudg{\tyetwo}{\tmfour\isub\var\tmtwo}{\ttytwo}$ such that 
$\size{\tyd_{\tmfour\isub\var\tmtwo}} = 
\size{\tyd_{\tmfour}} + \size{\tyd_{\tmtwo}^2} -\size{\leaves\tty_2}$.
		\end{enumerate}
    Note that $(\tmthree\tmfour)\isub\var\tmtwo=\tmthree\isub\var\tmtwo\tmfour\isub\var\tmtwo$. Then the derivation 
$\tyd_{(\tmthree\tmfour)\isub\var\tmtwo}$ is defined as follows:
    \[
		\infer[\tyapp]{
      \tjudg{\tye\uplus\tyetwo}{\tmthree\isub\var\tmtwo\tmfour\isub\var\tmtwo}{\linty}
		}{
      \tyd_{\tmthree\isub\var\tmtwo} \pof \tjudg{\tye}{\tmthree\isub\var\tmtwo}{\arr{\ttytwo}{\linty}}
			& &
			\tyd_{\tmfour\isub\var\tmtwo} \pof \tjudg{\tyetwo}{\tmfour\isub\var\tmtwo}{\ttytwo}
    }
		\]
		for which 
		\[\begin{array}{rll}
      \size{\tyd_{(\tmthree\tmfour)\isub\var\tmtwo}} & = & \size{\tyd_{\tmthree\isub\var\tmtwo}} + 
\size{\tyd_{\tmfour\isub\var\tmtwo}} +1 
		   \\
		   & =_{\ih}& \size{\tyd_{\tmthree}} + \size{\tyd_{\tmtwo}^1} 
		   -\size{\leaves\tty_1} + \size{\tyd_{\tmfour}} + 
\size{\tyd_{\tmtwo}^2} -\size{\leaves\tty_2} +1 
      \\
      & =& (\size{\tyd_{\tmthree}} + \size{\tyd_{\tmfour}} +1) + (\size{\tyd_{\tmtwo}^1} + 
\size{\tyd_{\tmtwo}^2}) -( \size{\leaves\tty_1}  + \size{\leaves\tty_2})
      \\
      & = & \size{\tyd_{\tmthree\tmfour}} + \size{\tyd_{\tmtwo}} - \size{\leaves\tty}
		  \end{array}\]

    \item \emph{Rule $\tynone$}. Then $\gty = \emmset$, and $\tty = \emmset$. It goes as 
for the second 
variable case ($\tm = \vartwo$). Namely, the hypothesis 
$\tyd_\tmtwo\pof\tjudg{}{\tmtwo}{\tty}$ is necessarily obtained by 
applying another $\tynone$ rule, and $\size{\tyd_\tmtwo} = 0$. The typing 
derivation $\tyd_{\tm\isub\var\tmtwo}$ is also a single $\tynone$ rule for the term $\tm\isub\var\tmtwo$. Note that 
$\tyd_{\tm\isub\var\tmtwo}$  satisfies the statement because
$\size{\tyd_{\tm\isub\var\tmtwo}} = \size{\tyd_\tm} = \size{\tyd_\tm} + 0 -0 =
\size{\tyd_{\tm}} + \size{\tyd_{\tmtwo}} -\size{\leaves\tty}$.
    
    \item \emph{Rule $\tymany$}. Then $\tyd_\tm$ has the following shape:
    \[
      \infer[\tymany]{
        \tjudg{\mset{\uplus_{i=1}^n\tye_i}, 
        \var:\mset{\uplus_{i=1}^n\tty_i}}{\tm}{\mset{\gty_1,\mydots,\gty_n}}
        }{
        \infer*{\tjudg{\tye_i, \var:\tty_i}{\tm}{\gty_i}}{} & 1\leq i\leq n
        }
    \]
By \reflemma{tree-split-merge}, the 
hypothesis $\tyd_\tmtwo\pof\tjudg{}{\tmtwo}{\tty}$ splits into $n$ derivations 
$\tyd_\tmtwo^i\pof\tjudg{}{\tmtwo}{\tty_i}$ with $i\in\set{1,\mydots, n}$ such that $\size{\tyd_\tmtwo} 
= \sum_{i=1}^n \size{\tyd_\tmtwo^i}$. By \ih, there exist $n$ derivations $\tyd_{\tm\isub\var\tmtwo}^i \pof 
\tjudg{\tye_i}{\tm\isub\var\tmtwo}{\gty_i}$ 
such that $\size{\tyd_{\tm\isub\var\tmtwo}}^i = 
\size{\tyd_{\tm}}^i + \size{\tyd_{\tmtwo}^i} -\size{\leaves\tty_i}$.
Then the derivation 
$\tyd_{\tm\isub\var\tmtwo}$ is defined as follows:
    \[
      \infer[\tymany]{
        \tjudg{\mset{\uplus_{i=1}^n\tye_i}}{\tm\isub\var\tmtwo}{\mset{\gty_1,\mydots,\gty_n}}
        }{
        \tyd_{\tm\isub\var\tmtwo}^i \pof \infer*{\tjudg{\tye_i}{\tm\isub\var\tmtwo}{\gty_i}}{} & 1\leq i\leq n
        }
    \]
		for which 
		\[\begin{array}{rll}
      \size{\tyd_{\tm\isub\var\tmtwo}} & = & \sum_{i=1}^n\size{\tyd_{\tm\isub\var\tmtwo}}^i 
		   \\
		   & =_{\ih}& \sum_{i=1}^n(\size{\tyd_{\tm}}^i + 
		   \size{\tyd_{\tmtwo}^i}-\size{\leaves\tty_i})
      \\
      & =& \sum_{i=1}^n\size{\tyd_{\tm}}^i + 
      \sum_{i=1}^n\size{\tyd_{\tmtwo}^i}-\sum_{i=1}^n\size{\leaves\tty_i}
      \\
      & = & \size{\tyd_{\tm}} + \size{\tyd_{\tmtwo}} - \size{\leaves\tty}
		  \end{array}\]
	\end{itemize}
\end{proof}

\begin{proposition}[Quantitative Subject Reduction]
\label{prop:q-subj-red}
	If $\tm$ is closed, $\tyd\pof\tjudg{}{\tm}{\linty}$, and $\tm\towh\tmtwo$ then there exists 
	$\tydtwo\pof\tjudg{}{\tmtwo}{\linty}$ such that $\size\tyd> \size\tydtwo$.
\end{proposition}
\begin{proof}
	By induction on $\tm\towh\tmtwo$.
	\begin{itemize}
		\item \emph{Base case, step at top level}: $\tm=(\la\var\tmthree)\tmfour \towh\tmthree\isub\var\tmfour=\tmtwo$. 
Note that $\tmfour$ is closed because $\tm$ is closed by hypothesis. Then $\tyd$ has the following shape.
		\[
		\infer[\tyapp]{
      \tjudg{}{(\la\var\tmthree)\tmfour}{\linty}
    }{
      \infer[\tylam]{
        \tjudg{}{\lambda\var.\tmthree}{\arr{\tty}{\linty}}
      }{
        \tyd_\tmthree \pof \infer*{\tjudg{\var:\tty}{\tmthree}{\linty}}{}
      }
			& 
			\tyd_\tmfour \pof \infer*{\tjudg{}{\tmfour}{\tty}}{}
    }
		\]
		We can apply the quantitative substitution lemma (\reflemma{q-sub-lemma}) to the sub-derivations $\tyd_\tmthree$ 
and $\tyd_\tmfour$ obtaining a derivation $\tyd_{\tmthree\isub\var\tmfour} \pof 
\tjudg{}{\tmthree\isub\var\tmtwo}{\linty}$ such that $\size{\tyd_{\tmthree\isub\var\tmfour}} = \tyd_{\tmthree} + 
\tyd_{\tmfour} -\size{\leaves\tty}< \tyd_{\tmthree} + 
\tyd_{\tmfour} +2 = \size\tyd$.

		\item \emph{Inductive case, step on the left of the root application}: 
		$\tm=\tmthree\tmfour\towh \tmthree'\tmfour=\tmtwo$ with $\tmthree \towh \tmthree'$. Then $\tyd$ has the following 
shape.
		\[
		\infer[\tyapp]{\tjudg{}{\tmthree\tmfour}{\linty 
		}}{
		 \tyd_\tmthree \pof \infer*{\tjudg{}{\tmthree}{\arr{\tty}{\linty}}}{}
			&
			\tyd_\tmfour \pof \infer*{\tjudg{}{\tmfour}{\tty}}{}}
		\]
		Applying the \ih to the left sub-derivation $\tyd_\tmthree$, we obtain $\tydtwo_{\tmthree'} \pof 
\tjudg{}{\tmthree'}{\arr{\tty}{\linty}}$ such that $\size{\tyd_\tmthree} > 
\size{\tydtwo_{\tmthree'}}$. Then $\tydtwo$ is 
defined as follows.
		\[\tydtwo \defeq
		\infer[\tyapp]{\tjudg{}{\tmthree'\tmfour}{\linty 
		}}{
      \tydtwo_{\tmthree'} \pof 
      \infer*{\tjudg{}{\tmthree'}{\arr{\tty}{\linty}}}{}
			&
			\tyd_\tmfour \pof \infer*{\tjudg{}{\tmfour}{\tty}}{}}
		\]
  for which $\size{\tyd} = \size{\tyd_\tmthree} +1 +\size{\tyd_\tmfour}>_{\ih} \size{\tydtwo_{\tmthree'}} + 1 +\size{\tyd_\tmfour}= \size{\tydtwo}$, as 
required.
	\end{itemize}
\end{proof}

\begin{theorem}[Correctness of tree types for \ccbn]
 If $\tyd\pof\tjudg{}{\tm}{\linty}$ then $\tm$ is \ccbn terminating.
\end{theorem}

\begin{proof}
By induction on $\size\tyd$ and case analysis of whether $\tm$ $\towh$-reduces. Cases:
\begin{enumerate}
  \item \emph{$\tm$ does not reduce}. Then it is $\towh$-normal. 
 
  \item \emph{$\tm \towh \tmtwo$ for some $\tmtwo$}. By quantitative subject reduction 
(\refprop{q-subj-red}), there exists $\tydtwo\pof\tjudg{}{\tmtwo}{\linty}$ such that $\size\tyd>\size\tydtwo$. Then we 
can apply the \ih to $\tydtwo$, obtaining that $\tmtwo$ is \ccbn normalizing. Therefore, so is $\tm$.
\end{enumerate} 
\end{proof}

\paragraph{Completeness.} The completeness of the type system is easier to prove, because there is no need to develop 
the quantitative analysis, not having to show termination of a relation.

\begin{lemma}[Anti-substitution]
\label{l:q-anti-sub-lemma}
	Let $\tyd \pof\tjudg{\tye}{\tm\isub\var\tmtwo}{\gty}$ with $\tmtwo$ closed. Then there exist
	\begin{itemize}
	 \item a tree type $\tty$,
	 \item a derivation $\tyd_\tm\pof\tjudg{\tye,\var:\tty}{\tm}{\gty}$, and 
	 \item a derivation $\tyd_\tmtwo\pof\tjudg{}{\tmtwo}{\tty}$.
	\end{itemize}
\end{lemma}

\begin{proof}
	By lexycographic induction  on $(\tm, \gty)$. We first deal with  the case 
	in which $\gty$ is a tree type $\tty$. We look at the last rule of $\tyd$. 
	Cases:
	\begin{itemize}
    \item \emph{Rule $\tynone$}. Then $\gty = \emmset$. The statement holds with respect to $\tty \defeq \emmset$, $\tyd_\tm \defeq \tyd$ and $\tyd_\tmtwo$ being another $\tynone$ rule of term $\tmtwo$.
    
    \item \emph{Rule $\tymany$}. Then $\tyd$ has the following shape:
       \[
      \infer[\tymany]{
        \tjudg{\mset{\uplus_{i=1}^n\tye_i}}{\tm\isub\var\tmtwo}{\mset{\gty_1,\mydots,\gty_n}}
        }{
        \tyd^i \pof \infer*{\tjudg{\tye_i}{\tm\isub\var\tmtwo}{\gty_i}}{} & 1\leq i\leq n
        }
    \]
By \ih (2nd component), for $1\leq i\leq n$ there exist $\tty_i$ and derivations $\tyd_{\tm}^i \pof 
\tjudg{\tye_i, \var:\tty_i}{\tm}{\gty_i}$ and $\tyd_{\tmtwo}^i \pof \tjudg{}{\tmtwo}{\tty_i}$. 
By \reflemma{tree-split-merge}, the derivations $\tyd_{\tmtwo}^i$ merge into a derivation $\tyd_{\tmtwo} \pof \tjudg{}{\tmtwo}{\tty}$ where $\tty \defeq \tty_1\uplus\mydots\uplus\tty_n$. The derivation $\tyd_\tm$ is instead obtained as follows.
     \[
      \infer[\tymany]{
        \tjudg{\mset{\uplus_{i=1}^n\tye_i}, 
        \var:\mset{\uplus_{i=1}^n\tty_i}}{\tm}{\mset{\gty_1,\mydots,\gty_n}}
        }{
        \infer*{\tjudg{\tye_i, \var:\tty_i}{\tm}{\gty_i}}{} & 1\leq i\leq n
        }
    \]
	\end{itemize}
	
	Now, we assume $\gty$ to be a linear type $\linty$, and look at the cases for the last rule of $\tyd$.
	\begin{itemize}
		\item \emph{Variable}. Two sub-cases:
		\begin{enumerate}
		 \item $\tm=\var$: then $\tm\isub\var\tmtwo = \tmtwo$. The statements holds by taking 
		 \begin{itemize}
			 \item $\tty \defeq \mset \linty$
			 \item $\tyd_\tm$ as an axiom assigning type $\linty$ to $\var$, and 
			 \item $\tyd_\tmtwo\defeq \tyd$.
		 \end{itemize}
		 \item $\tm=\vartwo$: then $\tm\isub\var\tmtwo = \vartwo$. The statements holds by taking 
		 \begin{itemize}
			 \item $\tty \defeq \emmset$
			 \item $\tyd_\tm$ as an axiom assigning type $\linty$ to $\vartwo$, and 
			 \item $\tyd_\tmtwo$ as a $\tynone$ rule of term $\tmtwo$.
		 \end{itemize}
		\end{enumerate}

		\item \emph{Rule $\tylamstar$}. Then $\tm = \la\vartwo\tmthree$ and $\gty = \initty$. The statements holds by taking 
		 \begin{itemize}
			 \item $\tty \defeq \emmset$
			 \item $\tyd_\tm$ as a $\tylamstar$ rule of term $\la\vartwo\tmthree$, and 
			 \item $\tyd_\tmtwo$ as a $\tynone$ rule of term $\tmtwo$.
		 \end{itemize}

 		\item \emph{Rule $\tylam$}. Then $\tyd$ has the following shape:
 		\[\infer*{\infer[\tylam]{  \tjudg{\tye}
				{\la\vartwo\tmthree\isub\var\tmtwo}{\arr{\ttytwo}{\lintytwo}}}
			{\tyd_{\tmthree\isub\var\tmtwo}\pof\tjudg{\tye,\vartwo:\ttytwo}{\tmthree\isub\var\tmtwo}{\lintytwo}}}{}
		\]
		with 
 		$\tm=\la\vartwo\tmthree$ and $\linty=\arr{\ttytwo}{\lintytwo}$. By \ih (1st component), there exists $\tty$ and derivations $\tyd_\tmthree \pof \tjudg{\tye,\var:\tty,\vartwo:\ttytwo}{\tmthree}{\lintytwo}$ and $\tyd_\tmtwo\pof\tjudg{}{\tmtwo}{\tty}$. Applying back rule $\tylam$ to $\tyd_\tmthree$, we obtain 
$\tyd_{\la\vartwo\tmthree}$ as follows.
		\[\infer*{\infer[\tylam]{  \tjudg{\tye,\var:\tty}
				{\lambda\vartwo.\tmthree}{\arr{\ttytwo}{\lintytwo}}}
			{\tyd_\tmthree 
\pof \tjudg{\tye,\var:\tty,\vartwo:\ttytwo}
				{\tmthree}{\lintytwo}}}{}
		\]

		\item \emph{Rule $\tyapp$}. Then $\tyd$ has the following shape:
		 \[
		\infer[\tyapp]{
      \tjudg{\tye\uplus\tyetwo}{\tmthree\isub\var\tmtwo\tmfour\isub\var\tmtwo}{\linty}
		}{
      \tyd_{\tmthree\isub\var\tmtwo} \pof \tjudg{\tye}{\tmthree\isub\var\tmtwo}{\arr{\ttytwo}{\linty}}
			& &
			\tyd_{\tmfour\isub\var\tmtwo} \pof \tjudg{\tyetwo}{\tmfour\isub\var\tmtwo}{\ttytwo}
    }
		\]
		with $\tm=\tmthree\tmfour$ and $\tty = \tty_1\uplus\tty_2$. By \ih (1st component), there exist:
		\begin{enumerate}
		 \item a tree type $\tty_1$ and derivations $\tyd_{\tmthree} \pof \tjudg{\tye, \var: \tty_1}{\tmthree}{\arr{\ttytwo}{\linty}}$ and $\tyd_\tmtwo^1\pof\tjudg{}{\tmtwo}{\tty_1}$;
		 \item a tree type $\tty_2$ and derivations $\tyd_{\tmfour} \pof \tjudg{\tyetwo, \var: \tty_2}{\tmfour}{\ttytwo}$ and $\tyd_\tmtwo^2\pof\tjudg{}{\tmtwo}{\tty_2}$.
		\end{enumerate}
By \reflemma{tree-split-merge}, the derivations $\tyd_\tmtwo^1$ and $\tyd_\tmtwo^2$ merge into a derivation $\tyd_\tmtwo\pof\tjudg{}{\tmtwo}{\tty}$ with $\tty \defeq \tty_1 \uplus \tty_2$. 
The derivation $\tyd_\tm$ is instead obtained as follows.
   		\[
		\infer[\tyapp]{
      \tjudg{\var:\underbrace{\tty_1\uplus \tty_2}_{\tty},\tye\uplus\tyetwo}{\tmthree\tmfour}{\linty}
		}{
      \tyd_\tmthree \pof \infer*{\tjudg{\var:\tty_1,\tye}{\tmthree}{\arr{\ttytwo}{\linty}}}{}
			& 
			\tyd_\tmfour \pof \infer*{\tjudg{\var:\tty_2,\tyetwo}{\tmfour}{\ttytwo}}{}
    }
		\]
\end{itemize}
\end{proof}

\begin{proposition}[Subject expansion]
\label{prop:subj-exp}
	If $\tm$ is closed, $\tyd\pof\tjudg{}{\tmtwo}{\linty}$, and $\tm\towh\tmtwo$ then there exists 
$\tydtwo\pof\tjudg{}{\tm}{\linty}$.
\end{proposition}
\begin{proof}
	By induction on $\tm\towh\tmtwo$.
	\begin{itemize}
		\item \emph{Base case, step at top level}: $\tm=(\la\var\tmthree)\tmfour \towh\tmthree\isub\var\tmfour=\tmtwo$. 
Note that $\tmfour$ is closed because $\tm$ is closed. The derivation in the hypothesis is $\tyd \pof 
\tjudg{}{\tmthree\isub\var\tmfour}{\linty}$.
Then by the anti-substitution lemma (\reflemma{q-anti-sub-lemma}) we obtain a tree type $\tty$ and two derivations 
$\tyd_\tmthree \pof \tjudg{\var:\tty}{\tmthree}{\linty}$ and $\tyd_\tmfour \pof \tjudg{}{\tmfour}{\tty}$. The derivation 
$\tydtwo$ of the statement is then defined as follows:
		\[
		\infer[\tyapp]{
      \tjudg{}{(\la\var\tmthree)\tmfour}{\linty}
    }{
      \infer[\tylam]{
        \tjudg{}{\lambda\var.\tmthree}{\arr{\tty}{\linty}}
      }{
        \tyd_\tmthree \pof \infer*{\tjudg{\var:\tty}{\tmthree}{\linty}}{}
      }
			& 
			\tyd_\tmfour \pof \infer*{\tjudg{}{\tmfour}{\tty}}{}
    }
		\]

		\item \emph{Inductive case, step on the left of the root application}: 
		$\tm=\tmthree\tmfour\towh \tmthree'\tmfour=\tmtwo$ with $\tmthree \towh \tmthree'$. Then $\tyd$ has the following 
shape.
		\[
		\infer[\tyapp]{\tjudg{}{\tmthree'\tmfour}{\linty 
		}}{
		 \tyd_{\tmthree'} \pof \infer*{\tjudg{}{\tmthree'}{\arr{\tty}{\linty}}}{}
			&
			\tyd_\tmfour \pof \infer*{\tjudg{}{\tmfour}{\tty}}{}}
		\]
		Applying the \ih to the left sub-derivation $\tyd_{\tmthree'}$, we obtain $\tydtwo_\tmthree \pof 
\tjudg{}{\tmthree}{\arr\tty\linty}$. Then $\tydtwo$ is 
defined as follows.
		\[\tydtwo \defeq
		\infer[\tyapp]{\tjudg{}{\tmthree\tmfour}{\linty 
		}}{
      \tydtwo_\tmthree \pof \infer*{\tjudg{}{\tmthree}{\arr\tty\linty}}{}
			&
			\tyd_\tmfour \pof \infer*{\tjudg{}{\tmfour}{\tty}}{}}
		\]
	\end{itemize}
\end{proof}

\begin{theorem}[Completeness of tree types for \ccbn]
 If $\tm$ is \ccbn terminating then there exists a tree type derivation $\tyd\pof\tjudg{}{\tm}{\linty}$.
\end{theorem}

\begin{proof}
Let $\tm \towh^n \tmtwo$ the reduction of $\tm$ to weak head normal form. By induction on $n$. Cases:
\begin{enumerate}
 \item If $n= 0$ then $\tm = \tmtwo$ is a weak head normal form, that is, an abstraction. Then it is typable with rule 
$\tylamstar$.

  \item If $n>0$ then $\tm \towh \tmthree \towh^{n-1} \tmtwo$. By \ih, there 
  exists 
$\tydthree\pof\tjudg{}{\tmthree}{\linty}$. By subject expansion \refprop{subj-exp}, there exists 
$\tydtwo\pof\tjudg{}{\tm}{\linty}$.
\end{enumerate}
\end{proof}

\section{Proofs from Section~\ref{sect:bisimulation}}
In the first part of this section we prove the T-exhaustible state invariant 
for the \TIAM, then use it to extract 
\LIAM states from \TIAM ones, and finally prove the strong bisimulation between 
the two machines.

In the second part we deal with showing that the \TIAM never loops on type 
derivations. The key tool shall be a 
loop-preserving bisimulation between \TIAM states of the type derivation of 
$\tm$ and $\tmtwo$ if $\tm \towh \tmtwo$.

\subsubsection{T-Exhaustible Invariant}
We present an example of type derivation for the term 
$\tm=(\la\vartwo{\la\var{\var\vartwo}})\mathsf{I}(\la\varthree\varthree)$, the 
same 
example used in Section~\ref{sec:IJK}. We use it to explain the next 
technical definitions. We have annotated the occurrences of $\initty$ with 
natural numbers, so that they represent the run on the type derivation.


We start by defining the notions of typed tests used to define T-exhaustible 
states.

\paragraph*{Type Positions and Generalized States} To define tests, we have to 
consider 
a slightly more general notion of \TIAM state. In \refsect{TIAM}, a state 
is a quadruple $(\tyd, \ruleoc, \ltyctx, \pol)$ where $\ruleoc$ is an 
occurrence 
of a judgement 
$\tjudg{\tye}{\tmtwo}{\linty}$ in $\tyd$, $\pol$ is a direction, and $\ltyctx$ 
is 
a linear
type context isolating an occurrence of 
$\initty$ in $\linty$. The generalization simply is to consider linear type 
contexts 
$\ltyctx$ such that $\ltyctxp{\lintytwo} = \linty$ 
for some $\lintytwo$, that is, not necessarily isolating $\initty$. A pair 
$(\lintytwo, \ltyctx)$ such that $\ltyctxp{\lintytwo} = 
\linty$ is called a position in $\linty$. 

Note that the \TIAM can be naturally adapted to this more general notion of 
state, that follows an arbitrary formula 
$\lintytwo$, not necessarily $\initty$
---it can be found in \reffig{genTIAM}, 
and it amounts to simply replace $\initty$ 
with $\lintytwo$.

To easily manage \TIAM states we also use a concise notations, writing 
$\tjudg{}{\tm}{\linty,\ltyctx}$ for a state 
$\state= (\tyd, \ruleoc, (\linty,\ltyctx), \pol)$ where $\ruleoc$ is 
$\tjudg{\tye}{\tm}{\ltyctxp\linty}$ for some $\tye$, potentially 
specifying the direction via colors and under/over-lining.

\begin{figure}[t]

\begin{center}\footnotesize
\begin{tabular}{c}
$\begin{array}{ccc||ccc}
	\infer{\tjudg{}{\red{\tm\tmtwo}}{\ltyctxp{\lintytwo_\uppt} (=\linty)}} 
	{\tjudg{}{\tm}{\arr{\tty}{\linty}} & \vdash} 
	&\tomachdotone&
	\infer{\tjudg{}{\tm\tmtwo}{\linty 
		}}{\tjudg{}{\red\tm}{\arr{\tty}{\ltyctxp{\lintytwo_\uppt}}} & 
		\vdash}
		&
	\infer{\tjudg{}{\red{\lambda\var.\tm}}{\arr{\tty} 
			{\ltyctxp{\lintytwo_\uppt}}}} 
	{\tjudg{}{\tm}{\linty (= \ltyctxp{\lintytwo})}}
	& \tomachdottwo &
	\infer{\tjudg{}{\lambda\var.\tm}{\arr{\tty}{\linty}}} 
	{\tjudg{}{\red\tm}{\ltyctxp{\lintytwo_\uppt}}}
	 \\[8pt]\hhline{======}&&&\\

	\infer{\tjudg{}{\tm\tmtwo}{\linty(= \ltyctxp{\lintytwo}) 
		}}{\tjudg{}{\blue\tm}{\arr{\tty}{\ltyctxp{\lintytwo_\downpt}}} & 
		\vdash}
		
	&\tomachdotthree&
		\infer{\tjudg{}{\blue{\tm\tmtwo}}{\ltyctxp{\lintytwo_\downpt}}} 
		{\tjudg{}{\tm}{\arr{\tty}{\linty}} & \vdash}
		&
		\infer{\tjudg{}{\lambda\var.\tm}{\arr{\tty}{\linty}}} 
		{\tjudg{}{\blue\tm}{\ltyctxp{\lintytwo_\downpt} (=\linty)}}
		 & \tomachdotfour &
		\infer{\tjudg{}{\blue{\lambda\var.\tm}}{\arr{\tty} 
				{\ltyctxp{\lintytwo_\downpt}}}} 
		{\tjudg{}{\tm}{\linty}}
		\\[8pt]\hhline{======}&&&\\
		\infer*{\infer{\tjudg{}{\la\var\ctxp{\var}} 
			{\arr{\leafctxp{\linty_i}}\lintythree}}{}}
	{\infer[i]{\tjudg{}{\red\var}{\ltyctxp{\lintytwo_\uppt}_i (= \linty_i)}}{}}   
	&\tomachvar&
	 \infer*{\infer{\tjudg{}{\blue{\la\var\ctxp{\var}}} 
			{\arr{\leafctxp{\ltyctxp{\lintytwo_\downpt}_i}}\lintythree}}{}}
	{\infer[i]{\tjudg{}{\var}{\linty_i}}{}}
	&
	
	\infer*{\infer{\tjudg{}{\red{\la\var\ctxp{\var}}} 
			{\arr{\leafctxp{\ltyctxp{\lintytwo_\uppt}_i}}\lintythree}}{}}
	{\infer[i]{\tjudg{}{\var}{\linty_i (=\ltyctxp{\lintytwo}_i)}}{}}
	 & \tomachbttwo &
	 \infer*{\infer{\tjudg{}{\la\var\ctxp{\var}} 
			{\arr{\leafctxp{\linty_i}}\lintythree}}{}}
	{\infer[i]{\tjudg{}{\blue\var}{\ltyctxp{\lintytwo_\downpt}_i}}{}} 
	\\[8pt]\hhline{======}\\
		\end{array}$
		\\
	$\begin{array}{ccccccc }
	%
%
	%
	
		%
		\infer{\tjudg{}{\tm\tmtwo}{\linty}} 
		{\tjudg{}{\blue\tm}{\arr{\leafctxp{\ltyctxp{\lintytwo_\downpt}_i}}{\linty}}
			& \infer=[\tymany]{\tjudg{}{\tmtwo}{\tty 
			(=\leafctxp{\ltyctxp{\lintytwo}_i)}} }{
			\ldots\ \tjudg{}{\tmtwo}{\ltyctxp{\lintytwo}_i}\ \ldots
			 } }
		& \tomacharg &
		\infer{\tjudg{}{\tm\tmtwo}{\linty}} 
		{\tjudg{}{\tm}{\arr{\tty}{\linty}}
			& 
			\infer=[\tymany]{\tjudg{}{\tmtwo}{\leafctxp{\ltyctxp{\lintytwo}_i}}
			 }{
			\ldots\ \tjudg{}{\red\tmtwo}{\ltyctxp{\lintytwo_\uppt}_i}\ \ldots
			 } }
	
		\\[8pt]\hhline{===}\\

		\infer{\tjudg{}{\tm\tmtwo}{\linty}} 
		{\tjudg{}{\tm}{\arr{\tty}{\linty}}
			& 
			\infer=[\tymany]{\tjudg{}{\tmtwo}{\leafctxp{\ltyctxp{\lintytwo}_i}
			 (=\tty)} }{
			\ldots\ \tjudg{}{\blue\tmtwo}{\ltyctxp{\lintytwo_\downpt}_i}\ \ldots
			 } }
		 & \tomachbtone &
		 		\infer{\tjudg{}{\tm\tmtwo}{\linty}} 
		{\tjudg{}{\red\tm}{\arr{\leafctxp{\ltyctxp{\lintytwo_\uppt}_i}}{\linty}}
			& \infer=[\tymany]{\tjudg{}{\tmtwo}{\tty} }{
			\ldots\ \tjudg{}{\tmtwo}{\ltyctxp{\lintytwo}_i}\ \ldots
			 } }
		\end{array}$
		\end{tabular}
		
		\end{center}

	\vspace{-8pt}
	\caption{The transitions of the (Generalized) Tree $\IAMold$ (\TIAM).}
	\label{fig:genTIAM}  
\end{figure}

\paragraph*{\TIAM Tests} Given a \TIAM state $\state= (\tyd, \ruleoc, 
(\linty,\ltyctx), \pol)$, the underlying idea is that 
the judgement occurrence $\ruleoc$ encodes the log of the \LIAM, while the type 
context $\ltyctx$ encodes the 
tape. It is then natural to define two kinds of test, one for judgements and 
one for type contexts.

The intuition is that a test focuses on (the occurrence of) a leaf 
$\lintytwo$ 
of a tree $\tty$ related to 
$\state$, and that these leaf elements play the role of logged positions in 
the \LIAM. These leaf elements 
are of two kinds:
\begin{enumerate}
	\item \emph{Elements containing $\ruleoc$}: those in which the focused 
	judgment $\ruleoc$ itself is contained, 
	corresponding to the logged positions in the log of the \LIAM. Note that 
	the positions on the log are those for which 
	the \LIAM has previously found the corresponding arguments. In the \TIAM 
	these arguments are exactly those in which the 
	focused judgment is contained.
	
	\item \emph{Elements appearing in $\ltyctx$}: those in 
	the right-hand type of $\state$ in which the focused type $\linty$ is 
	contained, 
	corresponding to the logged positions on 
	the tape of the \LIAM. They correspond to \LIAM queries for which the 
	argument has not yet been found, or positions to 
	which the \LIAM is backtracking to.
\end{enumerate}
Each one of these elements is then identified by a judgement occurrence 
$\ruleoc'$ and a position $(\lintytwo,\ltyctxtwo)$ 
in the right-hand type of $\ruleoc'$. 

\begin{definition}[Focus]
	A \emph{focus} $\focus$ in a derivation $\tyd$ is a pair $\focus = 
	(\ruleoc, 
	(\linty,\ltyctx))$ of a judgement occurrence 
	$\ruleoc$ and of a type position $(\linty,\ltyctx)$ in the right-hand type 
	$\ltyctxp\linty$ of $\ruleoc$.
\end{definition}

The intuition is that exhausting a test $\state_{\ruleoc, (\linty,\ltyctx)}$ in 
$\tyd$ shall amount to retrieve the 
axiom of 
$\tyd$ of type $\linty$ that would be substituted by that sequence element of 
type 
$\linty$ by reducing $\tyd$ via 
cut-elimination---the 
definition of exhaustible tests is given below, after the definition of tests.

\begin{definition}[Judgement tests]
	Let $\state=(\tyd, \ruleoc, (\linty,\ltyctx), \pol)$ be a \TIAM state. Let 
	$r_i$ be $i$-th $\tymany$ rule tree found traversing $\tyd$ 
	by descending from the focused judgment $\ruleoc$ towards the final 
	judgment of $\tyd$. Let $\ruleoc_i$ be the 
	topmost traversed judgment of $r_i$ in such a descent. 
	Let $\ruleoc_i$ be $\tjudg{\tye}{\tm}{\lintytwo}$. Then 
	$\state_{\focus}^i = (\tyd,\ruleoc_i,(\lintytwo,\ctxhole), 
	\downpt)$ is the $i$-th judgement test of $\state$, having as focus $\focus 
	\defeq (\ruleoc_i,(\lintytwo,\ctxhole))$. 
\end{definition}
We 
often
omit the judgement from the focus, writing simply 
$\state_{(\lintytwo,\ctxhole)}$, 
and even concisely 
note $\state_{\focus}$ as 
$\tjudg{}{\blue\tm}{\lintytwo,\ctxhole_\downpt}$.

Note that judgement tests always have type context $\ctxhole$. According to the 
intended correspondance judgement/ 
log and type context/tape between the \TIAM and the \LIAM, having type context 
$\ctxhole$ corresponds to the fact that 
the log tests of the \LIAM have an empty tape.


%

\paragraph*{Type (Context) Tests} While judgement tests depend only on the 
judgement occurrence $\ruleoc$ of a state 
$\state = (\tyd, \ruleoc, (\linty,\ltyctx), \pol)$, type context 
tests---dually---fix $\ruleoc$ and depend only on the 
type 
context $\ltyctx$ of $\state$, that is, they all focus on sequence elements of 
the form $(\ruleoc, (\lintytwo,\ltyctxtwo))$ 
where $\ltyctxtwop\lintytwo = \ltyctxp\linty$ and $\ltyctx = 
\ltyctxtwop\ltyctxthree$ for 
some type context $\ltyctxthree$. Namely, 
there is one type context test (shortened to \emph{type test}) for every 
flattened tree (\ie\ sequence) 
in which the hole of $\ltyctx$ is contained. We need some notions about type 
contexts, in particular a notion of level 
analogous to the one for term contexts.

\paragraph*{Terminology About Type Contexts} Define type contexts $\ltyctx_n$ of 
level $n\in\nat$ as follows:

\[\begin{array}{lclr}
	\ltyctx_0 &\defeq &\ctxhole \mid \arr\tty\ltyctx_0
	\\
	\ltyctx_{n+1} &\defeq &\arr{\leafctxp{\ltyctx_n}}\linty \mid 
	\arr\tty\ltyctx_{n+1}	
\end{array}\]
Clearly, every type context $\ltyctx$ can be seen as a type context $\ltyctx_n$ 
for a unique $n$, and viceversa a type 
context of level $n$ is also simply a type context---the level is then 
sometimes omitted.
A \emph{prefix} of a context $\ltyctx$ is a context $\ltyctxtwo$ such that 
$\ltyctxtwop\ltyctxthree = \ltyctx$ for some 
$\ltyctxthree$. Given $\ltyctx$ of level $n>0$, there is a smallest prefix 
context $\ltyctx|_i$ of level $0<i\leq 
n$, and it has the form 
$\ltyctxtwo\ctxholep{\arr\leafctx\linty}$ for a type 
context $\ltyctxtwo$ 
of level $i-1$.


\begin{definition}[Type tests]
	Let $\state=(\tyd, \ruleoc, (\linty,\ltyctx), \pol)$ be a \TIAM state and 
	$n$ 
	be the level of $\ltyctx$. The sequence 
	of directed prefixes $\DiPref\ltyctx$ of $\ltyctx$ is the sequence of pairs  
	$(\ltyctxtwo,\poltwo)$, where $\ltyctxtwo$ is a prefix of $\ltyctx$, defined 
	as follows:
	\[\begin{array}{lclll}
		\DiPref\ltyctx & \defeq & \mset\cdot & \mbox{if }n=0
		\\
		\DiPref\ltyctx & \defeq & \mset{(\ltyctx|_1,\uppt^0), 
		\ldots,(\ltyctx|_n,\uppt^{n-1})} &\mbox{if }n > 0	
	\end{array}\]
	The $i$-th directed prefix (from left to right) $(\ltyctxtwo,\poltwo)$ in 
	$\DiPref\ltyctx$ induces the type test 
	$\state_{\focus}^i \defeq (\tyd, \ruleoc, (\ltyctxthreep\linty,\ltyctxtwo), 
	\poltwo)$ of $\state$ and focus $\focus \defeq 
	(\ruleoc,(\ltyctxthreep\lintytwo,\ltyctxtwo))$, where $\ltyctxthree$ is the 
	unique type context such that $\ltyctx = \ltyctxtwop\ltyctxthree$.
\end{definition}

According to the idea that type tests correspond to 
the tape tests of the \LIAM,  note that the first element (on the left) of the 
sequence $\DiPref\ltyctx$ has $\uppt$ 
direction, and that the direction alternates along the sequence. This is the 
analogous to the fact that the tape test 
associated to the first logged position on the tape 
(from left to right) has always direction $\downp$, and passing to the test of 
the next logged position on the tape switches the direction.


\begin{definition}[State respecting a focus]
	Let $\focus=(\ruleoc, (\linty,\ltyctx))$ be a focus. A \TIAM state $\state$ 
	respects $\focus$ if it is an axiom 
	$\tjudg{}{\blue\var}{\ctxholep\linty_\downpt}$ for some variable $\var$ 
	(the 
	typing context of $\state$, which is omitted 
	by convention, is $\var:\mset\linty$).
\end{definition}

\begin{definition}[T-Exhaustible states]
	The set $\exstates_T$ of T-exhaustible states is the 
	smallest set such that if $\state\in\exstates_T$, then for each type or 
	judgement test of $\state_\focus$ of focus $\focus$ there exists a run 	
	$\run: \state_f \totiam^*\tomachbttwo\statetwo$ where $\statetwo$ respects 
	$\focus$ and for 
	the shortest such run $\statetwo\in\exstates_T$.
\end{definition}

\begin{lemma}[T-exhaustible invariant]
	\label{l:S-invariant-siam}
	Let $\tm$ be a closed term, $\tyd\pof\tjudg{\tye}{\tm}{\linty}$ a tree 
	type derivation for it, and $\runtwo:\ \tjudg{}{\tm}{\ctxholep\linty_\uppt} 
	\totiam^k
	\state$ an initial \TIAM run. Then $\state$ is T-exhaustible. 
\end{lemma}

\begin{proof}
		By
	induction on $k$. For $k=0$ there is nothing to prove because the initial 
	state $\state_0 = \tjudg{}{\tm}{\ctxholep\linty_\uppt}$ has no 
	judgement nor type tests. Then suppose
	$\runtwo':\state_0\totiam^{k-1}\statetwo$ and that the run continues with 
	$\statetwo\totiam\state$. By \ih, $\statetwo$ is 
T-exhaustible.

\emph{Terminology}: when a test state satisfies the clause in the definition of T-exhaustible states we say that it is \emph{positive}. 

	 Cases of 
	$\statetwo\totiam\state$:
	
	\begin{itemize}
		\item Case $\tomachdotone$.
		\[\begin{array}{clc}
		\statetwo=\infer{\tjudg{}{\red{\tm\tmtwo}}{\ltyctxp{\initty_{\uppt}}(=\linty)}}
		{\tjudg{}{\tm}{\arr{\tty}{\linty}} & \vdash} &
		\tomachdotone &
		\infer{\tjudg{}{\tm\tmtwo}{\linty 
			}}{\tjudg{}{\red\tm}{\arr{\tty}{\ltyctxp{\initty_{\uppt}}}} & 
			\vdash}=\state
		\end{array}\]
		\begin{itemize}
			\item \emph{Judgement tests.} Note that $\state$ has the same judgement tests of $\statetwo$, which are 
positive by the \ih
			\item \emph{Type tests.} We first consider the type tests of direction $\uppt$. 
			Let us $\red{\state_\focus}$ be one of them. We observe that there 
			is a 
			corresponding type test $\red{\statetwo_\focus}$ of $\statetwo$, that by \ih it is positive, and that 
			$\red{\statetwo_\focus}\totiam\red{\state_\focus}$. Since the machine is deterministic also $\red{\state_\focus}$ 
is positive. Let us now consider a type test $\blue{\state_\focus}$ of direction $\downpt$. We observe 
			that there is a 
			corresponding type test $\blue{\statetwo_\focus}$ of $\statetwo$, that it is positive by \ih, and that 
			$\blue{\state_\focus}\to\blue{\statetwo_\focus}$. Then $\blue{\state_\focus}$ is positive.
		\end{itemize}
		\item Case $\tomachdottwo$. Identical to the previous one.
		\item Case $\tomachvar$.
		\[
		\begin{array}{clc}
		\statetwo=\infer*{\infer{\tjudg{}{\la\var\ctxp{\var}} 
				{\arr{\leafctxp{\linty_i}}\lintytwo}}{}}
		{\infer[i]{\tjudg{}{\red\var}{\ltyctxp{\initty_\uppt}_i (= 
		\linty_i)}}{}}   
		&\tomachvar&
		\infer*{\infer{\tjudg{}{\blue{\la\var\ctxp{\var}}} 
				{\arr{\leafctxp{\ltyctxp{\initty_\downpt}_i}}\lintytwo}}{}}
		{\infer[i]{\tjudg{}{\var}{\linty_i}}{}}=\state		
		\end{array}
		\]
		\begin{itemize}
			\item \emph{Judgement tests.} Judgement tests of $\state$ are a subset 
			of judgement tests of $\statetwo$ and thus positive by \ih
			
			\item \emph{Type tests.} Let $n$ be the level of $\ltyctx$. Let $\state^j$ be the type test of $\state$ 
associated to the $j$-th triple in 
$\DiPref{\arr{\leafctxp{\ltyctx_i}}\lintytwo}$. Three 
cases, depending on the index $j$ of $\state^j$:
\begin{enumerate}
 \item $j = 1$: then $\state^1$ is $\tjudg{}{\red{\la\var\ctxp{\var}} }
			{\ltyctxp{\initty}_{i\uppt},\arr\leafctx\lintytwo}$.
			 Note that 
$\state^1\tomachbttwo\,\tjudg{} 
			{\blue\var}{\ltyctxp{\initty}_{i\downpt},\ctxhole}$, which has no type tests 
			and has the same judgement tests of $\statetwo$, which by 
			\ih are positive. Hence, $\state^1$ is T-exhaustible.
 
    \item \emph{$j$ is even}: for 
			$\blue\state^j$ (of direction $\downpt$) there is a corresponding type test $\red\statetwo^{j-1}$ of odd index 
of
			$\statetwo$, having 
direction $\uppt$ and such that 
			$\red\statetwo^{j-1}\tomachvar\blue\state^j$.  Thus one can 
			conclude by \ih and determinism of the \TIAM.
    
    \item \emph{$j \neq 1$ is odd}: for 
			$\red\state^j$ (of direction $\uppt$) there is a 
			corresponding type test $\blue\statetwo^{j-1}$ of even index of
			$\statetwo$, having direction $\downpt$ and such that 
			$\red\state^j\tomachbttwo \blue\statetwo^{j-1}$. Thus one can 
			conclude by \ih
\end{enumerate}

		\end{itemize}
		\item Case $\tomachbttwo$.
		\[\begin{array}{clc}
		\statetwo=\infer*{\infer{\tjudg{}{\red{\la\var\ctxp{\var}}} 
				{\arr{\leafctxp{\ltyctxp{\initty_\uppt}_i}}\lintytwo}}{}}
		{\infer[i]{\tjudg{}{\var}{\linty_i (=\ltyctxp{\initty}_i)}}{}}
		& \tomachbttwo &
		\infer*{\infer{\tjudg{}{\la\var\ctxp{\var}} 
				{\arr{\leafctxp{\linty_i}}\lintytwo}}{}}
		{\infer[i]{\tjudg{}{\blue\var}{\ltyctxp{\initty_\downpt}_i}}{}}=\state
		\end{array}\]
		\begin{itemize}
			\item \emph{Judgement tests.} The first type test of $\statetwo$ is 
			$\statetwo^1 \defeq \tjudg{}{\red{\la\var\ctxp{\var}}} 
			{\ltyctxp{\initty}_{i\uppt},\arr\leafctx\lintytwo}$.
			 Note that 
$\statetwo^1\tomachbttwo\,\tjudg{} 
			{\blue\var}{\ltyctxp{\initty}_{i\downpt},\ctxhole} \eqdef \statethree$ and that $\statethree$ exhausts 
$\statetwo^1$, and it is the first such state. Since $\statetwo^1$ is positive, $\statethree$ is T-exhaustible. Note 
that $\statethree$ has the same judgment tests of $\state$, which are then 
positive.
			
			\item \emph{Type tests.} For each odd type test 
			$\red\state^i$ of $\state$ (whose direction is $\uppt$), the corresponding even 
			type test $\blue\statetwo^{i+1}$ of $\statetwo$ has direction $\downpt$, is positive by \ih, and such that 
			$\red\state^i\tomachvar\blue\statetwo^{i+1}$. Then $\red\state^i$ is positive. For each even type 
			test 
			$\blue\state^i$ of $\state$ (whose direction is $\downpt$), the corresponding odd 
			type test $\red\statetwo^{i+1}$ of $\statetwo$ has direction $\uppt$, is positive by \ih, and such that 
			$\red\statetwo^{i+1}\tomachbttwo\blue\state^i$. Then $\state^i$ is positive by determinism of the 
\TIAM.
		\end{itemize}
		
		\item Cases $\tomachdotthree$ and $\tomachdotfour$. They are identical to case 
		$\tomachdotone$.
		
		\item Case $\tomacharg$.
		\[\small\begin{array}{clc}
		\statetwo=\infer{\tjudg{}{\tm\tmtwo}{\linty}} 
		{\tjudg{}{\blue\tm}{\arr{\leafctxp{\ltyctxp{\initty_\downpt}_i}}{\linty}}
			& \infer=[\tymany]{\tjudg{}{\tmtwo}{\tty 
			(=\leafctxp{\ltyctxp{\initty}_i)}} }{
				\ldots\ \tjudg{}{\tmtwo}{\ltyctxp{\initty}_i}\ \ldots
		} }
		& \tomacharg &
		\infer{\tjudg{}{\tm\tmtwo}{\linty}} 
		{\tjudg{}{\tm}{\arr{\tty}{\linty}}
			& \infer=[\tymany]{\tjudg{}{\tmtwo}{\leafctxp{\ltyctxp{\initty}_i}} 
			}{
				\ldots\ \tjudg{}{\red\tmtwo}{\ltyctxp{\initty_\uppt}_i}\ \ldots
		} }=\state
		\end{array}\]
		\begin{itemize}
			\item \emph{Judgement tests.} Judgement tests of $\state$ are those of 
			$\statetwo$, which are positive by \ih, plus 
			$\state^\tmtwo \defeq \tjudg{}{\blue\tmtwo}{\ltyctxp{\initty}_{i\downpt},\ctxhole}$. 
			Please note that 
			$\state^\tmtwo\tomachbtone\,\tjudg{}{\red\tm}{\ltyctxp{\initty}_{i\uppt},
				\arr\leafctx{\linty}}\eqdef\statetwo^t$. 
				Now,
			$\statetwo^t$ is a type test of $\statetwo$ and by \ih is 
			positive. Then $\state^\tmtwo$ is positive. 
			
			\item \emph{Type tests.} For each odd type test 
			$\red\state^i$ of $\state$ (whose direction is $\uppt$), the corresponding even 
			type test $\blue\statetwo^{i+1}$ of $\statetwo$ has direction $\downpt$, is positive by \ih, and such that 
			$\blue\statetwo^{i+1}\tomacharg\red\state^i$. Then $\state^i$ is positive by determinism of the 
\TIAM. For each even type 
			test 
			$\blue\state^i$ of $\statetwo$ (whose direction is $\downpt$), the corresponding odd 
			type test $\red\statetwo^{i+1}$ of $\statetwo$ has direction $\uppt$, is positive by \ih, and such that 
			$\blue\state^i\tomachbtone\red\statetwo^{i+1}$. Then $\blue\state^i$ is positive.
		\end{itemize}
		
		\item Case $\tomachbtone$.
		\[\small\begin{array}{clc}
		\statetwo=\infer{\tjudg{}{\tm\tmtwo}{\lintytwo}} 
		{\tjudg{}{\tm}{\arr{\tty}{\lintytwo}}
			& \infer=[\tymany]{\tjudg{}{\tmtwo}{\leafctxp{\ltyctxp{\initty}_i} 
			(=\tty)} }{
				\ldots\ \tjudg{}{\blue\tmtwo}{\ltyctxp{\initty_\downpt}_i}\ 
				\ldots
		} }
		& \tomachbtone &
		\infer{\tjudg{}{\tm\tmtwo}{\lintytwo}} 
		{\tjudg{}{\red\tm}{\arr{\leafctxp{\ltyctxp{\initty_\uppt}_i}}{\lintytwo}}
			& \infer=[\tymany]{\tjudg{}{\tmtwo}{\tty} }{
				\ldots\ \tjudg{}{\tmtwo}{\ltyctxp{\initty}_i}\ \ldots
		} }=\state
		\end{array}\]
		\begin{itemize}
			\item \emph{Judgement tests.} All judgement tests of $\state$ are judgement 
			test of $\statetwo$, which are this way positive by \ih
			
			\item \emph{Type tests.} The first type test of $\state$ is 
			$\state^1 \defeq \tjudg{}{\red\tm}{\ltyctxp{\initty}_{i\uppt}, 
				\arr{\mset{\ldots \ctxhole\ldots}}{\linty}}$. 
			Please note that $\statetwo^\tmtwo \defeq \tjudg{}{\blue\tmtwo} {\ltyctxp{\initty}_{i\downpt},\ctxhole}$ is a 
judgement test of $\statetwo$ such that $\statetwo^\tmtwo \tomachbtone \state^1$. By \ih, 
$\statetwo^\tmtwo$ is positive. By determinism of the \TIAM, $\state^1$ is positive.

			For each odd type test 
			$\red\state^i$ of $\state$ (whose direction is $\uppt$), the corresponding even 
			type test $\blue\statetwo^{i-1}$ of $\statetwo$ has direction $\downpt$, is positive by \ih, and such that 
			$\blue\statetwo^{i-1}\tomachbtone\red\state^i$. Then $\red\state^i$ is positive by determinism 
of the \TIAM. For each even type 
			test 
			$\blue\state^i$ of $\statetwo$ (whose direction is $\downpt$), the corresponding odd 
			type test $\red\statetwo^{i-1}$ of $\statetwo$ has direction $\uppt$, is positive by \ih, and such that 
			$\blue\state^i\tomacharg\red\statetwo^{i-1}$. Then $\blue\state^i$ is positive.
			
		\end{itemize}
	\end{itemize}
\end{proof}

\subsubsection{Extracting \LIAM States from \TIAM T-Exhaustible States, and the 
\LIAM/\TIAM Strong Bisimulation} 
From T-exhaustible states one is able to \emph{extract} \LIAM states, as the 
following definition shows. Please note that the definition is well-founded, 
precisely because the objects are T-exhaustible states. Indeed, the induction 
principle used to define T-exhaustability allows recursive definition on 
T-exhaustible states to be well-behaved. 
\begin{definition}[Extraction of logged positions]
	Let $\state$ be an T-exhaustible \TIAM state in a derivation $\tyd$, $\tm$ 
	be the final term in $\tyd$, and $\state_\focus$ be a judgement or type 
	test of $\state$. Since 
	$\state$ is 
	T-exhaustible, there is an exhausting run 
	$\state_\focus\totiam^+ \statetwo
	\in\exstates_T$. Let $\var$ be the variable of $\statetwo$. Then the logged 
	position extracted from 
	$\state_\focus$ is 
	$\elpos{\state_\focus} \defeq 
	(\var,\la\var\ctxtwo_n,\elpos{\statetwo^1}\cdot\ldots\cdot\elpos{\statetwo^n})$,
	 where 
	$\ctxtwo_n$ is the context (of level $n$)
	retrieved traversing $\tyd$ from $\statetwo$ to 
	the binder of $\l\var$ of $\var$ in $\tm$ and $\statetwo^i$ is the 
	$i$-th judgement test of $\statetwo$.
\end{definition}

\begin{definition}[Extraction of logs, tapes, and states]
	Let $\state=(\tyd, \ruleoc, (\linty,\ltyctx), \pol)$ be an T-exhaustible 
	\TIAM 
	state where $\tm$ is the final term in 
	$\tyd$, and $\ruleoc$ is $\tjudg{\tye}{\tmtwo}{\ltyctxp\linty}$. The \LIAM 
	state extracted from $\state$  is 
	$\estate{\state} \defeq 
	\nopolstate{\tmtwo}{\ctx_\state}{\etape\state}{\elog{\state}}{\pol}$\footnote{We
	 leave the color of $\pol$ unchanged, in the sense that $\estate\state$ is 
	 red/blue if 
	 $\state$ is red/blue, \ie $\downp$ becomes $\uppt$ and $\upp$ becomes 
	 $\downpt$.} 
	where
	\begin{itemize}
		\item \emph{Context}: $\ctx_\state$ is the only term context such that 
		$\tm = \ctx_\state\ctxholep\tmtwo$;
		\item \emph{Log}:
		$\elog{\state}\defeq\lpos_1\cdots\lpos_i\cdots\lpos_n$ where $\lpos_i = 
		\elpos{\state^i_\focus}$ where 
		$\state^i_\focus$ is the $i$-th judgement test of $\state$.		
		\item \emph{Tape}: $\etape\state = \etapeauxs{\ltyctx,0}$ where 
		$\etapeauxs{\ltyctx,i}$ is the auxiliary function 
		defined by induction on $\ltyctx$ as follows.
		\[\begin{array}{lcl}
			\etapeauxs{\ctxhole,i} &\defeq &\stempty 
			\\
			\etapeauxs{\arr\tty{\ltyctx}, i} & \defeq & \resm\cdot 
			\etapeauxs{\ltyctx,i}
			\\
			\etapeauxs{\arr{\leafctxp\ltyctx}\lintytwo,i}
			& \defeq & \elpos{\state^i_\focus}\cdot\etapeauxs{\ltyctx, i+1}
		\end{array}
		\]
		where $\state^i_\focus$ is the $i$-th type test of $\state$.
	\end{itemize}
	We use $\bisimtypes$ for the extraction relation between T-exhaustible 
	\TIAM states and \LIAM states defined as $(\state, \estate{\state}) 
	\in\bisimtypes$.
\end{definition}

First of all, we show that the extracted state respects the \LIAM invariant 
about the length of the log.

\begin{lemma}
	\label{l:extraction-length}
	Let $\state$ be an T-exhaustible \TIAM state and $\estate{\state} = 
	\nopolstate{\tm}{\ctx_\state}{\etape\state}{\elog{\state}}{\pol}$ the \LIAM 
	state extracted from it. Then the level of 
	$\ctx_\state$ is exactly the length of $\elog{\state}$, that is, $(\tm, 
	\ctx_\state,\elog{\state})$ is a logged 
	position.
\end{lemma}

\begin{proof}
	The length of $\elog{\state}$ is the number of judgement tests of $\state$, 
	which is the number of $\tymany$ rule trees, and thus of
	$\tyapp$ rules, traversed descending from the focused judgement $\ruleoc$ 
	of 
	$\state$ to the final judgement of $\tyd$. 
	The level of $\ctx_\state$ is the number of arguments in which the hole of 
	$\ctx_\state$ is contained, which are 
	exactly the number of 
	$\tyapp$ rules traversed descending from $\ruleoc$ to the final judgement 
	of 
	$\tyd$.
\end{proof}

\begin{proposition}[\TIAM-\LIAM bisimulation]
	Let $\tm$ a closed and $\towh$-normalizable term, and 
	$\tyd\pof\tjudg{}{\tm}{\initty}$ a type derivation. Then 
	$\bisimtypes$ is a strong bisimulation between T-exhaustible \TIAM states 
	on $\tyd$ and \LIAM states on $\tm$. 
	Moreover, if $\state_\tyd \bisimtypes \state_\l$ then $\state_\tyd$ is 
	\TIAM reachable if and only if $\state_\l$ is 
	\LIAM reachable.
\end{proposition}
\begin{proof}
  Assuming the bisimulation part of the statement, the moreover part follows from a trivial induction on the length of 
the initial run, since initial state are bisimilar and the bisimulation is exactly the fact that $\bisimtypes$ is 
stable by transitions.

For the bisimulation part, we consider each possible transitions. We focus on 
the half of the proof showing that \TIAM 
transitions are simulated by the \LIAM, the other half is essentially identical. 
	
	\begin{itemize}
		\item Case $\tomachdotone$.
		\[\begin{array}{clc}
		\statetwo=\infer{\tjudg{}{\red{\tm\tmtwo}}{\ltyctxp{\initty_{\uppt}}(=\linty)}}
		{\tjudg{}{\tm}{\arr{\tty}{\linty}} & \vdash} &
		\tomachdotone &
		\infer{\tjudg{}{\tm\tmtwo}{\linty 
			}}{\tjudg{}{\red\tm}{\arr{\tty}{\ltyctxp{\initty_{\uppt}}}} & 
			\vdash}=\state
		\\[8pt]	
			\bisimtypes&&
		\\[8pt]
		\estate\statetwo=\dstate{ \tm\tmtwo }{ \ctx_\statetwo }{ 
		\etape{\statetwo} }{ \elog{\statetwo} } 
		&\iamdap& 
		\dstate{ \tm }{ \ctxp{\ctxhole\tmthree} }{ \resm\cdot \etape{\statetwo} }{ 
			\elog{\statetwo} } = \state_\l
		\end{array}\]
		
			Note that $\ctx_\state = \ctx_\statetwo\ctxholep{\ctxhole\tmthree}$, $\elog\state = 
\elog{\statetwo}$, and $\etape{\state} = \resm\cdot\etape{\statetwo}$. Then, 
$\state_\l = \estate{\state}$, that is, 
$\state \bisimtypes \state_\l$.

    \item Case $\tomachdottwo$. Identical to the previous one.
    
		\item Case $\tomachvar$.
		\[\begin{array}{clc}
		\statetwo=\infer*{\infer{\tjudg{}{\la\var\ctxp{\var}} 
				{\arr{\leafctxp{\linty_i}}\lintytwo}}{}}
		{\infer[i]{\tjudg{}{\red\var}{\ltyctxp{\initty_\uppt}_i (= 
		\linty_i)}}{}}   
		&\tomachvar&
		\infer*{\infer{\tjudg{}{\blue{\la\var\ctxp{\var}}} 
				{\arr{\leafctxp{\ltyctxp{\initty_\downpt}_i}}\lintytwo}}{}}
		{\infer[i]{\tjudg{}{\var}{\linty_i}}{}}=\state		
		\\[8pt]
		\bisimtypes&&
		\\[8pt]
		\estate\statetwo=\dstate{ \var }{ 
		\underbrace{\ctxp{\l\var.\ctxtwo_n}}_{=\ctx_\statetwo} }{ 
		\etape\statetwo }{ 
			\underbrace{\tlog_n\cdot\tlog}_{=\elog\statetwo} } 
		&\tomachvar &
		\ustate{ \l\var.\ctxtwo_n\ctxholep\var}{ \ctx }{ 
			(\var,\l\var.\ctxtwo_n,\tlog_n)\cdot\etape\statetwo }{ \tlog } = \state_\l		
		\end{array}\]
		First of all, $\ctx_\statetwo$ has shape $\ctxp{\l\var.\ctxtwo_n}$ for 
		some $n$, as the descending path from the 
focused judgment to the final judgment passes through the showed $\tylam$ rule. 
Then $\ctx_{\state} = \ctx$.

About the log, 
by \reflemma{extraction-length} there is a correspondence between the level of 
term contexts and the length of the 
extracted log, so that $\elog\statetwo$ is at least of length $n$, that is, 
$\elog\statetwo = \tlog_n\cdot\tlog$, and 
$\elog\state = \tlog$.
		
		About the tape, note that $\etape\state = \elpos{\state^1_\focus}\cons\etapeaux{\ltyctx,1}\state$ where 
$\state^1_\focus$ is the first type test of $\state$. To show that $\estate\state = 
(\var,\l\var.\ctxtwo_n,\tlog_n)\cdot\etape\statetwo$ we have to show two things:
\begin{enumerate}
 \item $\elpos{\state^1_\focus} = 
(\var,\l\var.\ctxtwo_n,\tlog_n)$.
		Note that $\state^1_\focus$ is $\tjudg{}{\red{\la\var\ctxp{\var}} }
			{\ltyctxp{\initty}_{i\uppt},\arr\leafctx\lintytwo}$.
			 Note that 
$\state^1_\focus\tomachbttwo\,\tjudg{} 
			{\blue\var}{\ltyctxp{\initty}_{i\downpt},\ctxhole} = \statethree$, where $\statethree$ focusses on the same 
judgement of $\statetwo$, and that $\statethree$ is the state that T-exhausts 
$\state^1_\focus$. By definition of 
extraction, $\elpos{\state^1_\focus} = (\var,\l\var.\ctxtwo_n,\tlog_n)$. 

  \item $\etapeaux{\ltyctx,1}\state = \etape\statetwo$, that is, $\etapeaux{\ltyctx,1}\state = 
\etapeaux{\ltyctx,0}\statetwo$. Note that $\etapeaux{\ltyctx,1}\state$ and $\etapeaux{\ltyctx,0}\statetwo$ may differ only 
in the content of logged positions (obtained by extracting from tape tests), which is the only thing that depends on 
the direction and the state, the rest being uniquely determined by the type context $\ltyctx$. Here one has to repeat 
the reasoning done in the $\tomachbttwo$ case of the proof of the T-exhaustible 
invariant 
(\reflemma{S-invariant-siam}), that shows that the tape test of index $i>1$  for $\state$ and the one of index $i-1$ of 
$\statetwo$ exhaust on the same state, and thus induce the same logged position. Then $\etapeaux{\ltyctx,1}\state = 
\etape\statetwo$.
\end{enumerate}
Then $\estate\state = (\var,\l\var.\ctxtwo_n,\tlog_n)\cdot\etape\statetwo$, and so $\state_\l = \estate{\state}$, that 
is, $\state \bisimtypes \state_\l$.

		\item Case $\tomachbttwo$.
		\[\small\begin{array}{clc}
		\statetwo=\infer*{\infer{\tjudg{}{\red{\la\var\ctxp{\var}}} 
				{\arr{\leafctxp{\ltyctxp{\initty_\uppt}_i}}\lintytwo}}{}}
		{\infer[i]{\tjudg{}{\var}{\linty_i (=\ltyctxp{\initty}_i)}}{}}
		& \tomachbttwo &
		\infer*{\infer{\tjudg{}{\la\var\ctxp{\var}} 
				{\arr{\leafctxp{\linty_i}}\lintytwo}}{}}
		{\infer[i]{\tjudg{}{\blue\var}{\ltyctxp{\initty_\downpt}_i}}{}} =\state
		\\[8pt]
		\bisimtypes&&
		\\[8pt]
		\estate\statetwo=\dstate{ \la\var\ctxtwo_n\ctxholep{\var} }{ \ctx_\statetwo }{
			\underbrace{(\var,\la\var\ctxtwo_n,\tlog_n)\cons\etapeaux{\ltyctx,1}\statetwo}_{=\etape\statetwo} }{ \elog\statetwo 
}
		&	\tomachbttwo & 
		\ustate{ \var}{ \ctx_{\statetwo}\ctxholep{\la\var\ctxtwo_n} }{ \etapeaux{\ltyctx,1}\statetwo }{	\tlog_n 
\cons\elog\statetwo } = \statetwo_\l
		\end{array}\]
		
		About the tape of $\estate\statetwo$, note that $\etape\statetwo = 
\elpos{\statetwo^1_\focus}\cons\etapeaux{\ltyctx,1}\statetwo$ where $\state^1_\focus$ is 
the first type test of $\statetwo$. We have to show that $\state^1_\focus$ exhausts on $\var$, so that 
$\elpos{\state^1_\focus} = (\var,\l\var.\ctxtwo_n,\tlog_n)$ for some $\tlog_n$.
		Note that $\state^1_\focus$ is $\tjudg{}{\red{\la\var\ctxp{\var}} }
			{\ltyctxp{\initty}_{i\uppt},\arr\leafctx\lintytwo}$.
			 Note that 
$\state^1_\focus\tomachbttwo\,\tjudg{} 
			{\blue\var}{\ltyctxp{\initty}_{i\downpt},\ctxhole} = \statethree$, 
			where $\statethree$ focusses on the same 
judgement of $\state$, and that $\statethree$ is the state that S-exhausts $\state^1_\focus$. By definition of 
extraction, $\elpos{\state^1_\focus} = (\var,\l\var.\ctxtwo_n,\tlog_n)$ where $\tlog_n$ is the extraction of the first 
$n$ judgement tests of $\state$.  Then $\ctx_\state = \ctx_{\statetwo}\ctxholep{\la\var\ctxtwo_n}$ and $\elog\state = 
\tlog_n \cons\elog\statetwo$.

About the tape, for $\state$ we have to prove that $\etapeaux{\ltyctx,1}\statetwo = \etape\state = 
\etapeaux{\ltyctx,0}\state$. This is done as for $\tomachvar$, mimicking the 
reasoning in the proof of the T-exhaustible 
invariant (\reflemma{S-invariant-siam}).

Then, $\state_\l = \estate{\state}$, that is, $\state \bisimtypes \state_\l$.

		\item Cases $\tomachdotthree$ and $\tomachdotfour$. They are identical to case 
		$\tomachdotone$.
		
		\item Case $\tomacharg$.
		\[\small\begin{array}{clc}
		\statetwo=\infer{\tjudg{}{\tm\tmtwo}{\linty}} 
		{\tjudg{}{\blue\tm}{\arr{\leafctxp{\ltyctxp{\initty_\downpt}_i}}{\linty}}
			& \infer=[\tymany]{\tjudg{}{\tmtwo}{\tty 
			(=\leafctxp{\ltyctxp{\initty}_i)}} }{
				\ldots\ \tjudg{}{\tmtwo}{\ltyctxp{\initty}_i}\ \ldots
		} }
		& \tomacharg &
		\infer{\tjudg{}{\tm\tmtwo}{\linty}} 
		{\tjudg{}{\tm}{\arr{\tty}{\linty}}
			& \infer=[\tymany]{\tjudg{}{\tmtwo}{\leafctxp{\ltyctxp{\initty}_i}} 
			}{
				\ldots\ \tjudg{}{\red\tmtwo}{\ltyctxp{\initty_\uppt}_i}\ \ldots
		} }=\state
		\\[8pt]
		\bisimtypes&&
		\\[8pt]
		\estate\statetwo=\ustate{ \tm }{ \underbrace{\ctxtwop{\ctxhole\tmtwo}}_{=\ctx_{\statetwo}} }{ 
\underbrace{\elpos{\statetwo^1_\focus}\cdot\etapeaux{\ltyctx,1}{\statetwo}}_{=\etape\statetwo} }{ \elog\statetwo } 
		& \tomacharg &
		\dstate{ \tmtwo }{ \ctxtwop{\tm\ctxhole} }{ \etapeaux{\ltyctx,1}{\statetwo} }{ 
\elpos{\statetwo^1_\focus}\cdot\elog\statetwo } = \statetwo_\l
		\end{array}\]
		where $\statetwo^1_\focus$ is the first type test of $\statetwo$. Obviously, $\ctx_{\state} = \ctxtwop{\tm\ctxhole} 
$. For the log we have to show that $\elog\state$ is equal to $\elpos{\statetwo^1_\focus}\cdot\elog\statetwo $, which 
amounts to show that the first judgement test $\state^1$ of $\state$ exhausts on the same state as the first tape test 
$\statetwo^1_\focus$ of $\statetwo$. This is exactly the reasoning done in the 
proof of the T-exhaustible invariant. 
Similarly, one obtains that $\etapeaux{\ltyctx,1}\statetwo = \etape\state = 
\etapeaux{\ltyctx,0}\state$.

		\item Case $\tomachbtone$.
		\[\begin{array}{clc}
		\statetwo=\infer{\tjudg{}{\tm\tmtwo}{\linty}} 
		{\tjudg{}{\tm}{\arr{\tty}{\linty}}
			& \infer=[\tymany]{\tjudg{}{\tmtwo}{\leafctxp{\ltyctxp{\initty}_i} 
			(=\tty)} }{
				\ldots\ \tjudg{}{\blue\tmtwo}{\ltyctxp{\initty_\downpt}_i}\ 
				\ldots
		} }
		& \tomachbtone &
		\infer{\tjudg{}{\tm\tmtwo}{\linty}} 
		{\tjudg{}{\red\tm}{\arr{\leafctxp{\ltyctxp{\initty_\uppt}_i}}{\linty}}
			& \infer=[\tymany]{\tjudg{}{\tmtwo}{\tty} }{
				\ldots\ \tjudg{}{\tmtwo}{\ltyctxp{\initty}_i}\ \ldots
		} }=\state
			\\[8pt]
			\bisimtypes&&
			\\[8pt]
		\estate\statetwo=\ustate{ \tmtwo }{ \underbrace{\ctxtwop{\tm\ctxhole}}_{=\ctx_{\statetwo} }}{ \etape\statetwo }{ 
\underbrace{\elpos{\state^1_\focus}\cdot\tlog}_{=\elog\statetwo} }
		& \tomachbtone &
		\dstate{ \tm }{ \ctxtwop{\ctxhole\tmtwo} }{ \elpos{\state^1_\focus}\cdot\etape\statetwo }{ \tlog } 
= \statetwo_\l
		\end{array}\]
  where $\statetwo^1_\focus$ is the first judgement test of $\statetwo$. Obviously, $\ctx_{\state} = 
\ctxtwop{\ctxhole\tmtwo} 
$. For the log, there is nothing to prove. For the tape, we have to show that $\etape\state$ is equal to 
$\elpos{\state^1_\focus}\cdot\etape\statetwo$, which 
amounts to show two things. First, that the first tape test $\state^1$ of $\state$ exhausts on the same state as 
the first judgement test 
$\statetwo^1_\focus$ of $\statetwo$. Second, that $\etapeaux{\ltyctx,1}\state = \etape\statetwo = 
\etapeaux{\ltyctx,0}\statetwo$. Both points follow exactly the reasoning done 
in the proof of the T-exhaustible 
invariant.   
  \end{itemize}
\end{proof}

\subsection{The \TIAM is acyclic}
\label{ssect:acyclic-app}
First of all, we prove the abstract lemma that says that every state is 
reachable in a bi-deterministic transition 
system with only one initial state.

\begin{lemma}\label{lemma:bidet}
	Let $\tsys{}$ be an acyclic bi-deterministic transition system on a finite 
	set of states $\mathcal{S}$ and with only 
	one 	initial state $\state_i$. Then all states in $\mathcal{S}$ are 
	reachable from $\state_i$, and reachable only once.
\end{lemma}
\begin{proof}
	Let us consider a generic state $\state\in \mathcal{S}$ and show that it is 
	reachable from $\state_i$. If 
	$\state=\state_i$ we are done. 
	Otherwise, since the system is bi-deterministic we can deterministically go 
	backwards from $\state$. Since the set of states is finite and there are no 
	cycles, then the backward sequence must end on an initial state, that is, 
	on $\state_i$. 
	Thus $\state$ is reachable from $\state_i$. If a state is reachable twice, 
	then clearly there is a cycle, absurd.
\end{proof}

In order to prove that the \TIAM is acyclic, we need to show that if 
$\tm\towh\tmtwo$, then cycles are preserved between the tree type 
derivation $\tyd$ for $\tm$ and the sequence type derivation $\tydtwo$ for 
$\tmtwo$. One way to show this fact is building a (non-)termination-preserving 
bisimulation between states 
of $\tyd$ and states of $\tydtwo$. This idea has been already exploited 
in~\cite{IamPPDPtoAppear}, where bisimulations called \emph{improvements} are 
used to prove the correctness of the \LIAM, from which we now recall a few 
definitions.

\paragraph{Improvements.} A deterministic transition system (DTS) 
is a pair $\mathcal{S} =(S,\mathcal{T})$, where $S$ is a set of \emph{states} 
and 
$\mathcal{T}:S\rightharpoonup S$ a partial function.
If $\mathcal{T}(s)=s'$, then we write $s\rightarrow s'$, and if $s$ rewrites in 
$s'$ in $n$ steps then we write $s\rightarrow^n s'$. We note with 
$\mathcal{F}_S$ the set of final states, \ie 
the subset of $\mathcal{S}$ containing all $s\in\mathcal{S}$ such that 
$\mathcal{T}(s)$ is undefined. A state $s$ is terminating if there exists 
$n\geq 0$ and 
$s'\in\mathcal{F}_S$ such that $s\rightarrow^n s'$. We call $S_\downarrow$ the 
set of terminating states of $S$ and $S_\uparrow$ stands for $S\setminus 
S_\downarrow$. The \emph{evaluation length map} 
$|\cdot|:S\rightarrow\mathbb{N}\cup\{\infty\}$ is defined as $\size 
s\defeq n$ if $s\rightarrow^n s'$ and $s'\in\mathcal{F}_\mathcal{S}$, 
and $\size s \defeq \infty$ if $s\in\mathcal{S}_\uparrow$.

\begin{definition}[Improvements]
	Given two DTS $\mathcal{S}$ and $\mathcal{Q}$, a relation 
	$\mathcal{R}\subseteq S\times Q$ is 
	an \emph{improvement} if 
	given $(s,q)\in\mathcal{R}$ the following conditions hold.
	\begin{enumerate}
		\item \emph{Final state right}: if $q\in\mathcal{F}_\mathcal{Q}$, then 
		$s\rightarrow^n s'$, for some 
		$s'\in\mathcal{F}_\mathcal{S}$ and $n\geq 0$.
		\item \emph{Transition left}: if $s\rightarrow s'$, then there exists 
		$s'',q',n,m$ such 
		that $s'\rightarrow^m s''$, $q\rightarrow^n q'$, 
		$s''\mathcal{R}q'$ and $n\leq m+1$.
		\item \emph{Transition right}: if $q\rightarrow q'$, then there exists 
		$s',q'',n,m$ such that 
		$s\rightarrow^m s'$, $q'\rightarrow^n q''$, $s'\mathcal{R}q''$ and 
		$m\geq n+1$.
	\end{enumerate}
\end{definition}
What improves along an improvement is the number of transitions required to 
reach a final state, if any.
\begin{proposition}[\cite{IamPPDPtoAppear}]
	\label{prop:imp} 
	Let $\mathcal{R}$ be an improvement on two DTS $\mathcal{S}$ and 
	$\mathcal{Q}$, and $s\mathcal{R}q$. 
	\begin{enumerate}
		\item \label{p:imp-termination}
		\emph{Termination equivalence}: $s\in\mathcal{S}_\downarrow$ if and 
		only if $q\in\mathcal{Q}_\downarrow$.
		
		\item \label{p:imp-improv}
		\emph{Improvement}: $|s|\geq|q|$.
	\end{enumerate}
\end{proposition}

\paragraph{Weak Head Contexts} Next, we need the notion of weak head context 
$\hctx$ defined as:
$$\hctx \defeq \ctxhole \mid \hctx \tm$$
Note that if $\tm \towh \tmtwo$ then $\tm = \hctxp{(\la\var\tmthree)\tmfour}$ 
and $\tmtwo = 
\hctxp{\tmthree\isub\var\tmfour}$.

\paragraph{Explaining the Bisimulation} Let us give an intuitive 
explanation of the improvement 
$\relf$ that we are going to build next.	Given two type derivations 
$\tyd\pof\tjudg{}{\hctxp{(\la\var\tmthree)\tmfour}}{\initty}$ and 
$\tydtwo\pof\tjudg{}{\hctxp{\tmthree\isub\var\tmfour}}{\initty}$, it is 
possible 
to define a relation $\relf$ between states of the former and of the latter. 
The key points are:
\begin{enumerate}
	\item each axiom for $\var$ in $\tyd$ is $\relf$-related with the 
	judgement for the argument $w$ that replaces it in $\tydtwo$.
	\item Both the judgement for $r$ and the one for $(\la\var r)w$ are 
	$\relf$-related to $r\isub\var w$.
	\item The judgement for $\la\var r$ is not $\relf$-related to any judgement 
	of $\tydtwo$.
	
\end{enumerate}

\paragraph{Defining $\relf$} In order to define $\relf$ formally, we enrich 
each type judgment (occurrence)
$\tjudg{}{\tm}{\ltyctxp\initty}$ with a context $\ctx$ such that $\ctxp\tm$ is 
the term in the final judgement of the derivation $\tyd$, obtaining 
$\tjudg{}{(\tm,\ctx)}{\ltyctxp\initty}$. 

\begin{definition}[Bisimulation $\relf$]
	The 
	definition of $\relf$ for $\tjudg{}{(\tm,\ctx)}{\ltyctxp\initty}$ has 4 
	clauses:

	\begin{itemize}
		\item 
		$\mathsf{rdx}$: the redex is in $\tm$, that is, $\tm = 
		\hctxp{(\la\var\tmtwo)\tmthree}$, and so $\ctx$ is a head 
		context $\hctxtwo$:
		$$\tjudg{}{(\hctxp{(\la\var\tmtwo)\tmthree},\hctxtwo)}{\ltyctxp\initty} 
		\,\relfrdx
		\tjudg{}{(\hctxp{\tmtwo\isub\var\tmthree},\hctxtwo)}{\ltyctxp\initty}$$
		
		\item $\mathsf{body}$: the term $\tm$ is part of the body of the 
		abstraction involved in the redex:
		$$\tjudg{}{(\tm,\hctxp{(\la\var\ctxtwo)\tmtwo})}{\ltyctxp\initty} 
		\,\relfbody 
		\tjudg{}{(\tm\isub\var\tmtwo,\hctxp{\ctxtwo\isub\var\tmtwo})}{\ltyctxp\initty}$$
		
		\item $\mathsf{arg}$: the term $\tm$ is part of the argument of the 
		redex:
		$$\tjudg{}{(\tm,\hctxp{(\la\var\ctxtwop\var)\ctxthree})}{\ltyctxp\initty}
		\,\relfarg 
		\tjudg{}{(\tm,\hctxp{\ctxtwo\isub\var{\ctxthreep\tm}\ctxholep\ctxthree})}{\ltyctxp\initty}$$
		
		\item $\mathsf{ext}$: The term $\tm$ is disjoint form the redex, that 
		then takes place only in $\ctx$:
		$$\tjudg{}{(\tm,\hctxtwop{\hctxp{(\la\var\tmthree)\tmtwo}\ctxtwo})}{\ltyctxp\initty}
		\relfext\,
		\tjudg{}{(\tm,\hctxtwop{\hctxp{\tmthree\isub\var\tmtwo}\ctxtwo})}{\ltyctxp\initty}$$
	\end{itemize}
\end{definition}
Please note that the only states of $\tyd$ which are not mapped to any state 
of 
$\tydtwo$ are those relative to the judgment 
$\tjudg{}{\la\var\tmthree}{\arr{\mset{\gtytwo_1...\gtytwo_n}}\linty}$.

\begin{proposition}
	$\relf$ is an improvement between \TIAM states.
\end{proposition}

\begin{proof}\footnote{Also this proof requires colors.}
We inspect the 4 cases of the definition of $\relf$.
\begin{itemize}
 \item Rule $\mathsf{rdx}: \tjudg{}{(\hctxp{(\la\var\tm)\tmtwo},\hctxtwo)}{\ltyctxp\initty} 
		\,\relfrdx
		\tjudg{}{(\hctxp{\tm\isub\var\tmtwo},\hctxtwo)}{\ltyctxp\initty}$. Cases for $\uppt$ (by cases of $\hctx$):
	\begin{itemize}		
		\item $\hctx=\ctxhole$. The diagram is closed by rule $\mathsf{body}$:
		\[\small
		\begin{array}{ccccc}
		\tjudg{}{(\red{(\la\var\tm)\tmtwo},\hctxtwo)}{\ltyctxp{\initty}} & \totiam 
		& 
		\tjudg{}{(\red{\la\var\tm},\hctxtwop{\ctxhole\tmtwo})}{\arr\tty\ltyctxp{\initty}}
		& \totiam & 
		\tjudg{}{(\red{\tm},\hctxtwop{(\la\var\ctxhole)\tmtwo})}{\ltyctxp{\initty}}\\[4pt]
		\relfrdx&&&&\relfbody\\[4pt]
		\tjudg{}{(\red{\tm\isub\var\tmtwo},\hctxtwo)}{\ltyctxp{\initty}} &&=&& 
\tjudg{}{(\red{\tm\isub\var\tmtwo},\hctxtwo)}{\ltyctxp{\initty}}
		\end{array}
		\]	
		
		\item $\hctx=\hctxthree\tmfive$. The diagram is closed by rule $\relfrdx$:
		\[
		\begin{array}{ccc}
		\tjudg{}{(\red{\hctxthreep\tmthree\tmfive},\hctxtwo)}{\ltyctxp{\initty}} 
		& \totiam & 
		\tjudg{}{(\red{\hctxthreep\tmthree},\hctxtwop{\ctxhole\tmfive})} 
		{\arr\tty{\ltyctxp{\initty}}}\\[4pt]
		\relfrdx&&\relfrdx\\[4pt]
		\tjudg{}{(\red{\hctxthreep\tmfour\tmfive},\hctxtwo)}{\ltyctxp{\initty}} 
		& 
		\totiam & 
		\tjudg{}{(\red{\hctxthreep\tmfour},\hctxtwop{\ctxhole\tmfive})} 
		{\arr\tty\ltyctxp{\initty}}
		\end{array}
		\]

		\end{itemize}
		
		Cases for $\downpt$ (by cases of $\hctxtwo$):
		\begin{itemize}
		\item $\hctxtwo=\ctxhole$. Both machines are stuck.
		\[
		\begin{array}{c}
		\tjudg{}{(\tmthree,\blue\ctxhole)}{\ltyctxp{\initty}}\\[4pt]
		\relfrdx\\[4pt]
		\tjudg{}{(\tmfour,\blue\ctxhole)}{\ltyctxp{\initty}}
		\end{array}
		\]
		\item $\hctxtwo=\hctxthreep{\ctxhole\tmfive}$. Two subcases depending 
		on the type context.
		If the focus is on the right of the arrow the diagram is closed by rule 
		$\mathsf{rdx}$.
		\[
		\begin{array}{ccc}
		\tjudg{}{(\tmthree,\blue{\hctxthreep{\ctxhole\tmfive}})} 
		{\arr\tty{\ltyctxp{\initty}}} & \totiam & 
		\tjudg{}{(\tmthree\tmfive,\blue{\hctxthree})} {{\ltyctxp{\initty}}}
		\\[4pt]
		\relfrdx&&\relfrdx\\[4pt]
		\tjudg{}{(\tmfour,\blue{\hctxthreep{\ctxhole\tmfive}})} 
		{\arr\tty{\ltyctxp{\initty}}} & \totiam &
		\tjudg{}{(\tmfour\tmfive,\blue{\hctxthree})} {{\ltyctxp{\initty}}}
		\end{array}
		\]
		If the focus is on the left of the arrow the diagram is closed by rule 
		$\mathsf{ext}$.
		\[
		\begin{array}{ccc}
		\tjudg{}{(\tmthree,\blue{\hctxthreep{\ctxhole\tmfive}})} 
		{\arr{\leafctxp{\ltyctxp{\initty}}}{\linty}} & \totiam & 
		\tjudg{}{(\red\tmfive,\hctxthreep{\tmthree\ctxhole})} {{\ltyctxp{\initty}}}
		\\[4pt]
		\relfrdx&&\relfext\\[4pt]
		\tjudg{}{(\tmfour,\blue{\hctxthreep{\ctxhole\tmfive}})} 
		{\arr{\leafctxp{\ltyctxp{\initty}}}{\linty}} & \totiam & 
		\tjudg{}{(\red\tmfive,\hctxthreep{\tmfour\ctxhole})} {{\ltyctxp{\initty}}}
		\end{array}
		\]
	\end{itemize}	
	
	\item Rule $\mathsf{body}$: $\tjudg{}{(\tm,\hctxp{(\la\var\ctxtwo)\tmtwo})}{\ltyctxp\initty} 
		\,\relfbody 
		\tjudg{}{(\tm\isub\var\tmtwo,\hctxp{\ctxtwo\isub\var\tmtwo})}{\ltyctxp\initty}$. Cases of $\uppt$ (by cases of 
$\tm$):
	\begin{itemize}
		\item $\tm =\tmthree\tmfour$. Trivially 
		closed by rule $\mathsf{body}$.
		\item  $\tm=\la\vartwo\tmthree$. If $\tm:\initty$ both machines are 
		stuck. If $\tm:\arr\tty\linty$, the diagram is trivially closed by rule 
		$\mathsf{body}$.
		\item $\tm=\var$. Diagram closed by rule $\mathsf{arg}$.
		\[
		\small\begin{array}{ccccl}
		\tjudg{}{(\red\var,\hctxp{(\la\var\ctxtwo)\tmtwo}}{\ltyctxp{\initty}} 
		&\totiam& 
		\tjudg{}{(\la\var\ctxtwop\var,\blue{\hctxp{\ctxhole\tmtwo}})} 
		{\arr{\leafctxp{\ltyctxp{\initty}}}\lintytwo}{} 
		& \totiam &
		\tjudg{}{(\red{\tmtwo},\hctxp{(\la\var\ctxtwop\var)\ctxhole})}
		{\ltyctxp{\initty}}\\[4pt]
		\relfbody&&&&\relfarg\\[4pt]
		
\tjudg{}{(\red{\tmtwo},\hctxp{\ctxtwo\isub\var\tmtwo})}{\ltyctxp{\initty}}&&=&&\tjudg{}{(\red{\tmtwo},\hctxp{
\ctxtwo\isub\var\tmtwo})}{\ltyctxp{\initty}}
		\end{array}
		\]
		\end{itemize}
		
		Cases of $\downpt$ (by cases of $\ctxtwo$):
	\begin{itemize}
		\item $\ctxtwo=\ctxhole$. The diagram is closed by rule $\mathsf{rdx}$
		\[
		\small\begin{array}{ccccl}
		\tjudg{}{(\tm,\blue{\hctxp{(\la\var\ctxhole)\tmtwo})}}{\ltyctxp\initty} 
		&\totiam& 
		\tjudg{}{(\la\var\tm,\blue{\hctxp{\ctxhole\tmtwo}})} 
		{\arr\tty{\ltyctxp\initty}}
		& \totiam &
		\tjudg{}{((\la\var\tm)\tmtwo,\blue\hctx)}
		{\ltyctxp{\initty}}\\[4pt]
		\relfbody&&&&\relfrdx\\[4pt]
		
		\tjudg{}{(\tm\isub\var\tmtwo,\blue{\hctx})}{\ltyctxp\initty}&&=&& 
		\tjudg{}{(\tm\isub\var\tmtwo,\blue{\hctx})}{\ltyctxp\initty}
		\end{array}
		\]
		
		\item $\ctxtwo=\ctxthreep{\la\vartwo\ctxhole}$, 
		$\ctxtwo=\ctxthreep{\ctxhole\tmthree}$ and 
		$\ctxtwo=\ctxthreep{\tmthree\ctxhole}$. The diagram is trivially 
		closed by rule $\mathsf{body}$.
	\end{itemize}
	
	\item Rule $\relfarg$: 
	$\tjudg{}{(\tm,\hctxp{(\la\var\ctxtwop\var)\ctxthree})}{\ltyctxp\initty}
	\,\relfarg 
	\tjudg{}{(\tm,\hctxp{\ctxtwo\isub\var{\ctxthreep\tm}\ctxholep\ctxthree})}{\ltyctxp\initty}$.
	Cases of $\uppt$ (by cases of $\tm$) are all trivial: they are closed by 
	rule $\relfarg$ itself. The only non trivial case for $\downpt$ (by cases 
	of 
	$\ctxthree$) is when $\ctxthree=\ctxhole$.
	\[
	\small\begin{array}{ccccl}
	\tjudg{}{(\tm,\blue{\hctxp{(\la\var\ctxtwop\var)\ctxhole}})}
	{\ltyctxp{\initty}}
	&\totiam& 
	\tjudg{}{(\red{\la\var\ctxtwop\var},\hctxp{\ctxhole\tm})} 
	{\arr{\leafctxp{\ltyctxp{\initty}}}\lintytwo}{} 
	& \totiam &
	\tjudg{}{(\var,\blue{\hctxp{(\la\var\ctxtwo)\tm}}}{\ltyctxp{\initty}} \\[4pt]
	\relfarg&&&&\relfbody\\[4pt]
	\tjudg{}{(\tm,\blue{\hctxp{
			\ctxtwo\isub\var\tm}})}{\ltyctxp{\initty}}&&=&&
	\tjudg{}{(\tm,\blue{\hctxp{\ctxtwo\isub\var\tm})}}{\ltyctxp{\initty}}
	\end{array}
	\]
	
	\item Rule $\relfext$: 
	$\tjudg{}{(\tm,\hctxtwop{\hctxp{(\la\var\tmthree)\tmtwo}\ctxtwo})}{\ltyctxp\initty}
	\relfext\,
	\tjudg{}{(\tm,\hctxtwop{\hctxp{\tmthree\isub\var\tmtwo}\ctxtwo})}{\ltyctxp\initty}$.
	Cases of $\uppt$ (by cases of $\tm$) are all trivial: they are closed by 
	rule $\relfext$ itself. The only non trivial case for $\downpt$ (by cases 
	of 
	$\ctxtwo$) is when $\ctxtwo=\ctxhole$. We put $\tmfive\defeq 
	\hctxp{(\la\var\tmthree)\tmtwo}$ and $\tmfour\defeq 
	\hctxp{\tmthree\isub\var\tmtwo}$.
	\[
	\begin{array}{ccc}
	 \tjudg{}{(\tm,\blue{\hctxtwop{\tmfive\ctxhole}})} {{\ltyctxp{\initty}}} & 
	 \totiam & \tjudg{}{(\red\tmfive,{\hctxtwop{\ctxhole\tm}})} 
	 {\arr{\leafctxp{\ltyctxp{\initty}}}{\linty}}
	\\[4pt]
	\relfext&&\relfrdx\\[4pt]
	\tjudg{}{(\tm,\blue{\hctxtwop{\tmfour\ctxhole}})} {{\ltyctxp{\initty}}} & 
	\totiam & 
	\tjudg{}{(\red\tmfour,{\hctxtwop{\ctxhole\tm}})} 
	{\arr{\leafctxp{\ltyctxp{\initty}}}{\linty}}
	\end{array}
	\]
	\end{itemize}
\end{proof}

\begin{corollary}
	If $\tyd\pof\tjudg{}{\hctxp{(\la\var\tmthree)\tmfour}}{\initty}$ contains a 
	cycle, the also 
	$\tydtwo\pof\tjudg{}{\hctxp{\tmthree\isub\var\tmfour}}{\initty}$ contains a 
	cycle.
\end{corollary}
\begin{proof}
	If the run of the \TIAM on 
	$\tyd\pof\tjudg{}{\hctxp{(\la\var\tmthree)\tmfour}}{\initty}$ loops then 
	there exists a 
	state $\state_\tyd$ such that a computation 
	starting from $\state_\tyd$ diverges. Every state but 
	$\tjudg{}{(\la\var\tmthree,\hctxp{\ctxhole\tmfour})}{\ltyctxp\initty}$, 
	which however is not final, is 
	related by $\relf$ to a state $\state_{\tydtwo}$ of $\tsys\tydtwo$. Since 
	improvements preserve non-termination (\refpropp{imp}{termination}), also  
	$\state_{\tydtwo}$ diverges. Since 
	$\state_{\tydtwo}$ has a finite number of states, there must be a cycle.
\end{proof}
\begin{corollary}
	For each type derivation $\tyd\pof\tjudg{}{\tm}{\initty}$, $\tsys\tyd$ has 
	no cycles.
\end{corollary}
\begin{proof}
	Since $\tm$ is typable, then it has normal form, call it $\tmtwo$. Clearly 
	the type derivation for $\tmtwo$ has no cycles. By the previous corollary, 
	also $\tyd$ cannot have any of them.
\end{proof}

\begin{proposition}
	Let $\tm$ a closed term and $\tyd\pof\tjudg{}{\tm}{\initty}$ a tree type 
	derivation. Then every state of $\tyd$ is reached exactly once.
\end{proposition}
\begin{proof}
	Immediately by Lemma~\ref{lemma:bidet}.
\end{proof}

\section{Proofs from Section~\ref{sect:space}}
\begin{proposition}[Space of Single Extracted States]
	Let $\state=(\tyd, \ruleoc, \ltyctx, \pol)$ be a reachable \TIAM state. 
	Then $\ttm{\ltyctx}=\lm{\etape\state}$ and 
	$\tlm{\ruleoc}=\lm{\elog\state}$, and thus $\bsize\state = 
	\lm{\extr\state}$. Moreover,
	\begin{enumerate}
		\item if $\elog\state=\lpos_1\mydots\lpos_n$, and let $h_i$ be the 
		number 
		of $\tymany$ rules of the $i^{th}$ 
		$\tymany$ rule tree found descending from $\ruleoc$ to the root of 
		$\tyd$, 
		then $\lm{\lpos_i}=h_i$;
		\item for each extracted tape position $\lpos$, \ie for each $\gtyctx$ 
		such 
		that $\ltyctx=\gtyctxp{\arr{\leafctxp{\ltyctxtwop{\initty}}}\linty}$, 
		then 
		$\lm\lpos=\size\leafctx\cdot\indet$.
	\end{enumerate}
\end{proposition}
\begin{proof}
	We proceed by induction on the length of the run 
	$\runtwo:\state_0\totiam^*\state$. If the length is $0$, then 
	$\state=\state_0=(\tyd, \ruleoc, \ctxhole, \uppt)$, 
	$\ruleoc=\tjudg{}{\tm}{\initty}$ and is the root of $\run$. Then 
	$\ttm{\ctxhole}=0=\lm{\etape\state}$ and 
	$\tlm{\ruleoc}=0=\lm{\elog\state}$. 
	Otherwise, $\runtwo:\state_0\totiam^n\statetwo\totiam\state$. We analyze 
	the different cases of the last transition.
	\begin{itemize}
		\item Case $\iamdap$.
		\[\statetwo=\infer{\tjudg{}{\red{\tm\tmtwo}}{\ltyctxp{\initty_{\uppt}}(=\linty)}}
		{\tjudg{}{\tm}{\arr{\tty}{\linty}} & \vdash} 
		\tomachdotone 
				\infer{\tjudg{}{\tm\tmtwo}{\linty 
			}}{\tjudg{}{\red\tm}{\arr{\tty}{\ltyctxp{\initty_{\uppt}}}} & 
			\vdash}=\state\]
		 The log is unchanged. 
		$\etape\state=\resm\cons\etape\statetwo$. Thus  
		$\ttm{\arr\tty\ltyctx}=\ttm{\ltyctx}+1=_{\ih}\lm{\etape\statetwo}+1=\lm{\etape\state}$.
		\item Case $\iamdlamone$. Equivalent to the previous one.
		\item Case $\iamdvar$.
		\[
		\statetwo=\infer*{\infer{\tjudg{}{\la\var\ctxp{\var}} 
				{\arr{\leafctxp{\linty_i}}\lintytwo}}{}}
		{\infer[i]{\tjudg{}{\red\var}{\ltyctxp{\initty_\uppt}_i (= 
					\linty_i)}}{}}   
		\tomachvar
		\infer*{\infer{\tjudg{}{\blue{\la\var\ctxp{\var}}} 
				{\arr{\leafctxp{\ltyctxp{\initty_\downpt}_i}}\lintytwo}}{}}
		{\infer[i]{\tjudg{}{\var}{\linty_i}}{}}=\state	
		\]
		Let us set $k\defeq\size\leafctx-1$. We observe that $k$ is exactly the 
		number of rules $\tymany$ which the judgment $i$ lies in until the 
		judgment $\ruleoc$ corresponding to the binder. We have 
		$\estate\statetwo=\dstate{\var}{\ctxtwop{\la\var\ctx}} 
		{\etape{\ltyctx_i}}{\elog i\cdot\elog\ruleoc}$ and 
		$\estate\state=\ustate{\la\var\ctxp\var}{\ctxtwo} 
		{(\var,\la\var\ctx,\elog i)\cdot\etape{\ltyctx_i}}{\elog\ruleoc}$. By 
		\ih we have $k\cdot \indet=\lm{\elog i}$. Then
		$\ttm{\arr{\leafctxp{\ltyctx_i}}\lintytwo}=\indet+k\cdot 
		\indet+\ttm{\ltyctx_i}=\indet+\lm{\elog 
		i}+\lm{\etape\statetwo}=\lm{(\var,\la\var\ctx,\elog
		 i)}+\lm{\etape\statetwo}=\lm{\etape\state}$. About the log, it 
		 suffices to note that $\tlm{\ruleoc}=_{\ih}\lm{\elog{\state}}$.
		\item Case $\iamdlamtwo$.
		\[
		\statetwo=\infer*{\infer{\tjudg{}{\red{\la\var\ctxp{\var}}} 
				{\arr{\leafctxp{\ltyctxp{\initty_\uppt}_i}}\lintytwo}}{}}
		{\infer[i]{\tjudg{}{\var}{\linty_i (=\ltyctxp{\initty}_i)}}{}}
		 \tomachbttwo 
		\infer*{\infer{\tjudg{}{\la\var\ctxp{\var}} 
				{\arr{\leafctxp{\linty_i}}\lintytwo}}{}}
		{\infer[i]{\tjudg{}{\blue\var}{\ltyctxp{\initty_\downpt}_i}}{}}= \state	
		\]
		Let us set $k\defeq\size\leafctx-1$. We observe that $k$ is exactly the 
		number of rules $\tymany$ which the judgment $i$ lies in until the 
		judgment $\ruleoc$ corresponding to the binder. We have 
		$\estate\statetwo=\dstate{\la\var\ctxp\var}{\ctxtwo} 
		{(\var,\la\var\ctx,\elog i)\cdot\etape{\ltyctx_i}}{\elog\ruleoc}$ and 
		$\estate\state=\ustate{\var}{\ctxtwop{\la\var\ctx}} 
		{\etape{\ltyctx_i}}{\elog i\cdot\elog\ruleoc}$. 
		$\ttm{\ltyctx_i}=\ttm{\arr{\leafctxp{\ltyctxp{\initty_\uppt}_i}}\lintytwo}
		 -\size{\leafctx}\cdot\indet=_{\ih}\lm{\etape{\statetwo}} 
		 -\lm{(\var,\la\var\ctx,\elog i)}=\lm{(\var,\la\var\ctx,\elog 
		 i)\cdot\etape{\ltyctx_i}} 
		 -\lm{(\var,\la\var\ctx,\elog i)}=\lm{\etape{\ltyctx_i}}$. About the 
		 log, since $\lm{\elog{i}}=k\cdot\indet$ by \ih and 
		 $\tlm{\ruleoc}=_{\ih}\lm{\elog{\ruleoc}}$, then 
		 $\tlm{\ruleoc\ctxholep{i}}=\tlm{\ruleoc}+k\cdot\indet=\lm{\elog{\ruleoc}}+\lm{\elog{i}}
		  =\lm{\elog{\state}}$.
		  \item Cases $\iamuapltwo$ and $\iamulam$. Equivalent to case 
		  $\iamdap$.
		  \item Case $\tomacharg$.
		  \[
		  \statetwo=\infer{\tjudg{}{\tm\tmtwo}{\linty}} 
		  {\tjudg{}{\blue\tm}{\arr{\leafctxp{\ltyctxp{\initty_\downpt}_i}}{\linty}}
		  	& \infer=[\tymany]{\tjudg{}{\tmtwo}{\tty 
		  	(=\leafctxp{\ltyctxp{\initty}_i)}} }{
		  		\ldots\ \tjudg{}{\tmtwo}{\ltyctxp{\initty}_i}\ \ldots
		  	} }
		  	 \tomacharg 
		  	\infer{\tjudg{}{\tm\tmtwo}{\linty}} 
		  	{\tjudg{}{\tm}{\arr{\tty}{\linty}}
		  		& 
		  		\infer=[\tymany]{\tjudg{}{\tmtwo}{\leafctxp{\ltyctxp{\initty}_i}}
		  		 }{
		  			\ldots\ \tjudg{}{\red\tmtwo}{\ltyctxp{\initty_\uppt}_i}\ 
		  			\ldots
		  		} }= \state
		  \]
		  $\estate{\statetwo}=\ustate{\tm}{\ctx}{\etape{\leafctxp{\ltyctx_i}}} 
		  {\elog{\statetwo}}= \ustate{\tm}{\ctx}{\lpos\cdot\etape{\ltyctx_i}} 
		  {\elog{\statetwo}}$ and 
		  $\estate{\state}=\dstate{\tmtwo}{\ctxtwo}{\etape{\ltyctx_i}} 
		  {\elog{\state}}=\dstate{\tmtwo}{\ctxtwo}{\etape{\ltyctx_i}} 
		  {\lpos\cdot\elog{\ruleoc'}}$. We have by \ih 
		  $\lm{\lpos}+\ttm{\ltyctx_i}=\size{\leafctx}\cdot\indet+\ttm{\ltyctx_i}=\ttm{\leafctxp{\ltyctx_i}}=\lm{\etape{\leafctxp{\ltyctx_i}}}=
		  \lm{\lpos\cdot\etape{\ltyctx_i}}=\lm{\lpos}+ 
		  \lm{\etape{\ltyctx_i}}$. About the log, we have 
		  $\tlm{\ruleoc}=\tlm{\ruleoc'}+\size{\leafctx}\cdot\indet=_{\ih} 
		  \lm{\elog{\ruleoc'}}+\lm{\lpos}=\lm{\elog{\ruleoc}}$.
		  \item Case $\iamuapr$. Equivalent to the previous one.
	\end{itemize}
\end{proof}

\begin{lemma}
	Let $\gty$ be a type. Then 
	\[\ssize\gty=\max_{\gtyctx|\gty=\gtyctxp\initty}\ttm\gtyctx\]
\end{lemma}
\begin{proof}
	We proceed by induction on the structure of $\gty$.
	\begin{itemize}
		\item Case $\gty=\initty$. Then there is only one 
		$\gtyctx$ such that $\gty=\gtyctxp\initty$, \ie $\gtyctx=\ctxhole$.
		\item Case $\gty=\arr\tty\linty$. By \ih 
		$\ssize\tty=\max_{\ttyctx|\tty=\ttyctxp\initty}\ttm{\ttyctx}$ and 
		$\ssize\linty=\max_{\ltyctx|\linty=\ltyctxp\initty}\ttm{\ltyctx}$. We 
		have 
		that 
		$\{\gtyctx|\gty=\gtyctxp\initty\}=\{\arr\ttyctx\linty|\tty=\ttyctxp\initty\}
		\cup\{\arr\tty\ltyctx|\linty=\ltyctxp\initty\}$. Then
		\begin{equation*}\begin{split}\ssize{\gty}&=\ssize{\arr\tty\linty}=\max\{\ssize{\tty},
		\ssize{\linty}+1 
		\}= \max\left\lbrace 
		\max_{\ttyctx|\tty=\ttyctxp\initty}\ttm{\ttyctx}, 
		\max_{\ltyctx|\linty=\ltyctxp\initty}\ttm{\ltyctx}+1 
		\right\rbrace\\
		&=\max\{\{\ttm{\ttyctx}|\tty=\ttyctxp\initty\},\{\ttm{\ltyctx}+1| 
		\linty=\ltyctxp\initty\}\}\\
		&=\max\{\{\ttm{\arr\ttyctx\linty}|\tty=\ttyctxp\initty\}, 
		\{\ttm{\arr\tty\ltyctx}| \linty=\ltyctxp\initty\}\}
		=\max_{\gtyctx|\gty=\gtyctxp\initty}\ttm{\gtyctx}
		\end{split}\end{equation*}
		\item Case $\gty=\mset{\gty_1,\mydots,\gty_n}$. By \ih\ for each $1\leq 
		i\leq n$,
		$\ssize{\gty_i}=\max_{\gtyctx|\gty_i=\gtyctxp\initty}\ttm{\gtyctx}$.
		Then
		\begin{equation*}\begin{split}
		\ssize\gty&=\ssize{\mset{\gty_1,\mydots\gty_n}}=\indet+\max_{i}\{\ssize{\gty_i}\}
		=\indet+ \max_{i}\left\lbrace 
		\max_{\gtyctx|\gty_i=\gtyctxp\initty}\ttm{\gtyctx}
		\right\rbrace
		= \max_{i}\left\lbrace \indet+
		\max_{\gtyctx|\gty_i=\gtyctxp\initty}\ttm{\gtyctx}
		\right\rbrace\\
		&= \max_{i}\left\lbrace 
		\{\ttm{\mset{\gty_1\cdots\gtyctx_i\cdots\gty_n}}|\gty_i=\gtyctx_i\ctxholep\initty\}\right\rbrace
				=\max_{\gtyctx|\gty=\gtyctxp\initty}\ttm{\gtyctx}
				\end{split}\end{equation*}
	\end{itemize}
\end{proof}

\begin{lemma}[Weights bound extracted tapes]
	Let $\tyd:\wtjudg{\tye}{w}{\tm}{\gty}$ be a weighted derivation and $\mathcal{J}$ be the set of all the judgments  
	occurring in $\tyd$. 
	Then 
	\[w\geq\max_{\tjudg{\tyetwo}{\tmtwo}{\gtytwo}\in\mathcal{J}}\ssize{\gtytwo}\]
\end{lemma}
\begin{proof}
	We proceed by induction on $\tyd$.
	\begin{itemize}
		\item Case \tyvar. This case is trivial.
		\[\infer[\tyvar]{\wtjudg{\var:\mset{\linty}}{\ssize{\linty}}{\var}{\linty}}{}\]
		\item Case $\tylamstar$. Also this case is trivial, since 
		$\ssize{\initty}=0$. 
		\[\infer[\tylamstar]{\wtjudg{\tye}{0}{\lambda\var.\tm}{\initty}}{}\]
		\item Case $\tylam$. The thesis follows by the \ih\ applied to $v$.
		\[\infer[\tylam]{\wtjudg{\tye}{\max\{v,\,\ssize{\arr{\tty}{\linty}}\}} 
			{\lambda\var.\tm}{\arr{\tty}{\linty}}}{\wtjudg{\tye,\var:\tty}{v}{\tm}{\linty}}\]
		\item Case \tyapp. The thesis follows by the \ih\ applied to $u$ and 
		$v$ 
		and the fact that $\ssize{\linty}\leq\ssize{\arr\tty\linty}$.
		\[\infer[\tyapp]{\wtjudg{\tye\uplus\tyetwo}{\max\{u,v\}}{\tm\tmtwo}{\linty
			}}{\wtjudg{\tye}{u}{\tm}{\arr{\tty}{\linty}} & 
			\wtjudg{\tyetwo}{v}{\tmtwo}{\tty}}\]
		\item Case $\tymany$. The $\tyd$ has the following shape.
		\[\infer[\tymany]{\wtjudg{\mset{\uplus_{i=1}^n\tye_i} 
		}{\indet+\max_{i}\set{v_i}} 
			{\tm}{\mset{\gty_1,\mydots,\gty_n}}}
		{\wtjudg{\tye_i}{v_i}{\tm}{\gty_i} & 1\leq i\leq n}\]
		By \ih, $v_i \geq \ssize{\gty_i}$, so $w$ is $\geq$ of the weight of the right-hand type of any internal judgement 
of $\tyd$. We only have to show that $w$ also bounds the weight of $\mset{\gty_1,\mydots,\gty_n}$. Note that	$w = 
\indet+\max_{i}\set{v_i}\geq_{\ih}	\indet+\max_{i}\set{\ssize{\gty_i}}=:\ssize{\mset{\gty_1,\mydots,\gty_n}}$.
    
    \item Case $\tynone$. Trivial since $\ssize{\emmset} = 0$.
    \[\infer[\tynone]{\wtjudg{}{0}{\tm}{\emmset}}{}\]
	\end{itemize}
\end{proof}

\begin{lemma}[Weights bound also extracted logs]
	Let $\tyd:\wtjudg{\tye}{w}{\tm}{\gty}$ be a weighted derivation. Then $w\geq v+\tlm{\ruleoc}$ for every weighted judgment $\ruleoc$ $\wtjudgone{v}$ in $\tyd$.
\end{lemma}
\begin{proof}
By induction on the length $n$ of path from $\ruleoc$ to the final judgement of $\tyd$. If $n=0$ then $\tlm{\ruleoc} 
= 0$ and $w = v$, giving $w = v = v + \tlm{\ruleoc}$. If $n>0$ then we look at the rule of which $\ruleoc$ is a 
premise. Let $\ruleoc'$ $\wtjudgone{v'}$ be the concluding judgement of such a rule. By \ih, $w\geq 	
v'+\tlm{\ruleoc'}$. Now, for all rules but $\tymany$ we have that $v'\geq v$ and $\tlm{\ruleoc'} = \tlm\ruleoc$, so that 
$w\geq 	v'+\tlm{\ruleoc'}\geq v+\tlm{\ruleoc}$. For $\tymany$, we have $v' \geq \indet + v$ and $\tlm{\ruleoc'} = 
\tlm\ruleoc - \indet$, so that
$$w\geq 	v'+\tlm{\ruleoc'}\geq \indet + v +\tlm{\ruleoc} - \indet = v +\tlm{\ruleoc}.$$
\end{proof}

Let $\states\tyd$ be the set of \TIAM states during the execution of $\tyd$.

\begin{theorem}[\LIAM space bounds]
\label{thm:space-bounds}
	Let $\tyd\pof\wtjudg{}{\weight}{\tm}{\initty}$ be a weighted tree types derivation. Then $\lm{\extr\state} \leq w$ 
for every $\state\in \states\tyd$.
\end{theorem}
\begin{proof}
We prove the bound using $\bsize\state$ instead of $\lm{\extr\state}$, and obtain the statement because $\bsize\state$ 
instead of $\lm{\extr\state}$ by \refprop{space-single-states}. Let $\state=(\tyd, \ruleoc, \ltyctx, \pol)\in 
\states\tyd$ be a \TIAM state and let 
$\weighttwo$ be its weight. By \reflemma{weight}, $\weight\geq \tlm{\ruleoc}+\weighttwo$. By \reflemma{weightfour}, 
$\weighttwo \geq \ssize{\ltyctxp\initty}$, and by \reflemma{weightthree} $\ssize{\ltyctxp\initty} \geq 
\ttm\ltyctx$. Then $\weighttwo \geq \ttm\ltyctx$. Therefore, $\weight\geq \tlm{\ruleoc}+\ttm{\ltyctx} = \bsize\state$. 
\end{proof}

\begin{proposition}[Weight witness]
Let $\tyd:\wtjudg{\tye}{w}{\tm}{\gty}$ be a weighted derivation and $\gty\neq\emmset$. Then there exists a 
	\TIAM state $\state$ over $\tyd$ such that 	$w=\bsize\state$.
\end{proposition}

\begin{proof}
	We proceed by induction on the structure of $\tyd$.
	\begin{itemize}
		\item Case \tyvar:
		\[\infer[\tyvar]{\wtjudg{\var:\mset{\linty}}{\ssize{\linty}}{\var}{\linty}}{}\]
		 There is no log and thus $s_L=0$, and by 
		\reflemma{weightthree}, 
		$\ssize\linty=\max_{\ltyctx|\linty=\ltyctxp\initty}\ttm{\ltyctx}$.
		
		\item Case $\tylamstar$: 
		\[\infer[\tylamstar]{\wtjudg{\tye}{0}{\lambda\var.\tm}{\initty}}{}\]
		There is no log and thus $s_L=0$, and $\ssize{\initty}=0=\ttm{\ctxhole}$. 
		
		\item Case $\tylam$:
		\[\infer[\tylam]{\wtjudg{\tye}{\max\{v,\,\ssize{\arr{\tty}{\linty}}\}} 
			{\lambda\var.\tm}{\arr{\tty}{\linty}}}{\wtjudg{\tye,\var:\tty}{v}{\tm}{\linty}}\]
		There are two sub-cases:
		\begin{enumerate}
		 \item $\weight =v\geq\ssize{\arr{\tty}{\linty}}$: then the statement follows by the \ih 
		 \item $\weight = \ssize{\arr{\tty}{\linty}} >v$, by \reflemma{weightthree}, there is a state $\state=(\tyd, 
\ruleoc, \ltyctx, \pol)$ over the concluding judgement $ \ruleoc =
\tjudg{\tye}{\lambda\var.\tm}{\arr{\tty}{\linty}}$ for which $\ssize{\arr{\tty}{\linty}} = \ttm{\ltyctx}$. Since for the 
concluding judgement $\tlm\ruleoc=0$, we obtain 
$$\weight = \ssize{\arr{\tty}{\linty}} = \ttm{\ltyctx} = \ttm{\ltyctx} + 
\tlm\ruleoc = \bsize\state.$$
		\end{enumerate}

		\item Case \tyapp: 
		\[\infer[\tyapp]{\wtjudg{\tye\uplus\tyetwo}{\max\{u,v\}}{\tm\tmtwo}{\linty
			}}{\wtjudg{\tye}{u}{\tm}{\arr{\tty}{\linty}} & 
			\wtjudg{\tyetwo}{v}{\tmtwo}{\tty}}\]
			The thesis follows by the \ih\ applied to $u$ if 
		$u\geq v$ and to $v$ otherwise.
		
		\item Case $\tymany$: 
		\[\infer[\tymany]{\wtjudg{\mset{\uplus_{i=1}^n\tye_i}
			}{\indet+\max_{i}\set{v_i}} 
			{\tm}{\mset{\gty_1,\mydots,\gty_n}}}
		{\wtjudg{\tye_i}{v_i}{\tm}{\gty_i} & 1\leq i\leq n}\]
		Let us set $m\defeq\max_{i}\set{v_i}$. 
		Then, we apply the \ih to the sub-derivation $\tyd'$ with weight $m$. 
		Let us call $\statetwo=(\tyd', \ruleoc, \gtyctx, \pol)$ the state 
		obtained through the \ih. Then $\ttm{\gtyctx}+\tlm{\ruleoc}=m$. Let us 
		now consider the same state in the new type derivation $\tyd$, which 
		includes the $\tymany$ rule, $\state=(\tyd, \ruleoc, \gtyctx, \pol)$. 
		Now $\ttm{\gtyctx}+\tlm{\ruleoc}=\indet+m=\indet+\max_{i}\set{v_i}$.
		
		\item Case $\tynone$: 
		\[\infer[\tynone]{\wtjudg{}{0}{\tm}{\emmset}}{}\]
		Impossible, because by hypothesis $\gty\neq\emmset$.
	\end{itemize}
\end{proof}

\begin{corollary}[\LIAM exact bound via tree types derivations]
	Let $\tyd\pof\wtjudg{}{\weight}{\tm}{\initty}$ be a tree types derivation and $\runtwo$
	the complete \LIAM run on $\tm$. Then $\spacem{\runtwo}=\weight$.
\end{corollary}

\begin{proof}
By \refthm{space-bounds}, $\spacem{\runtwo}\leq\weight$.
 By \refprop{witness}, there exists a state $\state$ of the \TIAM over $\tyd$ such that $\bsize{\state} = w$. By 
\refprop{space-single-states}, $\lm{\extr\state} = \bsize{\state}$. Therefore, 
$\spacem{\runtwo}=\weight$.
\end{proof}

\section{Proofs from Section~\ref{sect:inefficiency}}
We give the type schema for the Turing's fixed point combinator $\Theta$, 
needed for the encoding of Turing machines,
defined as follows:
\[
\Theta\defeq\theta\theta\qquad\qquad\text{where 
}\quad\theta\defeq\la\var{\la\vartwo{\vartwo(\var\var\vartwo)}}
\]
 In doing so, we can safely assume, when typing $\Theta$, that
the argument we plan to pass to it, will use its argument linearly, and let
us attribute $\Theta$ a type with this simplifying
assumption in mind.

Consider a list of types $\tyl\defeq\linty_k,\mydots,\linty_0$,
type type $\linty_i$ being the type one would like to attribute to
$\Theta\tm$ after $i$ unfolding steps inside the recursion.
We can first of all type $\Theta$ in the following way:
\[
\begin{array}{lll}\Theta:\tyF_0^{\tyl}\defeq{\arr{\tyt_0^{\tyl}}{\linty_0}} 
&\text{ 
	where }& 
\tyt_0^{\tyl}\defeq\mset{\tyy_0^{\tyl}} \\
&\text{ and }&
\tyy_0^{\tyl}\defeq\arr\emmset\linty_0
\end{array}
\]
Considering that 
$\Theta\tm\towh(\la\vartwo{\vartwo(\Theta\vartwo)})\tm\towh\tm(\Theta\tm)$, 
$\vartwo:\tyt_0^{\tyl}$ and $\tm:\tyt_0^{\tyl}$.
But this is not the end of the story. What if recursion is unfolded
more than once? These type schemes can be inductively defined to
accommodate the general case:
\[
\begin{array}{rcl}
\Theta:\tyF_{n+1}^{\tyl}&=&\arr{\tyt_{n+1}^{\tyl}}{{\linty_{n+1}}}\\[3pt]
\text{where}\quad\tyt_{n+1}^{\tyl}&=&\mset{\tyy_{n+1}^{\tyl},\mset{\tyt_{n}^{\tyl}}}\\[3pt]
\text{and}\quad\tyy_{n+1}^{\tyl}&=&\arr{\mset{{\linty_n}}}{{\linty_{n+1}}}
\end{array}
\] 
Again, we note that the $n+2$ leaves of $\tyt_{n+1}^{\tyl}$ correspond to the 
fact that $\tm:\tyt_{n+1}^{\tyl}$ is evaluated $n+2$ times. Moreover, notice 
that
$\tyt_{n}$, seen as a tree type, is sparse, having a topology identical to
the one in Figure~\ref{fig:secondfigure}. Moreover, tree types can be used to 
derive the
space consumption of the \LIAM when used to evaluate $\Theta$.

\begin{lemma}\label{lemma:vartree}
	For each $n\geq 0$, and for each list of types $\tyl$ such that 
	$\size\tyl\geq n+1$, $\wtjudg{\mset{\tyt_n^{\tyl}}}{\geq 
	(2n+1)\indet}{\vartwo}{\tyt_n^{\tyl}}$.
\end{lemma}
\begin{proof}
	We proceed by induction on $n$. If $n=0$, we can type $\vartwo$ as follows:
	\[
	\infer{\wtjudg{\vartwo:\mset{\mset{\arr\emmset{\linty_0}}}}{\geq \indet}
		{\vartwo}{\mset{\arr\emmset{\linty_0}}}}{
		\infer{\wtjudg{\vartwo:\mset{\arr\emmset{\linty_0}}}{\geq 0}
			{\vartwo}{\arr\emmset{\linty_0}}}{}}
	\]
	Case $n+1$:
	\[
	\infer{\wtjudg{\vartwo:\mset{\mset{\tyy_{n+1}^{\tyl},\mset{\tyt_{n}^{\tyl}}}}}{\geq(2(n+1)+1)\indet}{\vartwo}
		{\mset{\tyy_{n+1}^{\tyl},\mset{\tyt_{n}^{\tyl}}}}}{
		\infer{\wtjudg{\vartwo:\mset{\tyy_{n+1}^{\tyl}}}{\geq 
		0}{\vartwo}{\tyy_{n+1}^{\tyl}}}{}&
		\infer{\wtjudg{\vartwo:\mset{\mset{\tyt_{n}^{\tyl}}}}{\geq(2n+2)\indet}{\vartwo}{\mset{\tyt_{n}^{\tyl}}}}{
			\infer{\wtjudg{\vartwo:\mset{\tyt_{n}^{\tyl}}}{\geq(2n+1)\indet}{\vartwo}{\tyt_{n}^{\tyl}}}{\ih}}}
	\]
\end{proof}

\begin{proposition}
	For each $n\geq 0$, and for each list of types $\tyl$ such that 
	$\size\tyl\geq n+1$, $\wtjudg{}{\geq (2n+1)\indet}{\Theta}{\tyF_n^{\tyl}}$.
\end{proposition}
\begin{proof}
	Since 
	$\Theta\to\la\vartwo\vartwo(\Theta\vartwo)=:\Theta_2$
	and types are preserved by reduction and expansion, we type $\Theta_2$. 
	We proceed by induction on $n$. If $n=0$, we can  type $\Theta_{(2)}$ 
	as follows:
	\[
	\infer{\wtjudg{}{\geq\indet}{\Theta_2\defeq\la\vartwo{\vartwo(\Theta\vartwo)}}{
			\arr{\mset{\arr{\emmset}{\linty_0}}}{\linty_0}}}
	{\infer{\wtjudg{\vartwo:\mset{\arr{\emmset}{\linty_0}}}{\geq 0} 
			{\vartwo(\Theta\vartwo)}{\linty_0}}
		{\infer{\wtjudg{\vartwo:\mset{\arr{\emmset}{\linty_0}}}{\geq 
		0}{\vartwo}{\arr{\emmset}{\linty_0}}}{}
			& \infer{}{}}}
	\]
	Now, we prove that $\Theta$ can be typed by $\tyF_{n+1}^{\tyl}$, 
	knowing that by \ih\ it can be typed with $\tyF_{n}^{\tyl}$. 
	%
	%
	%
	%
	%
	
	\[
	\infer{\wtjudg{}{\geq(2(n+1)+1)\indet}{\Theta_2\defeq\la\vartwo\vartwo(\Theta\vartwo)}
	 {\arr{\mset{\tyy_{n+1}^{\tyl},\mset{\tyt_n^{\tyl}}}}
			{{\linty_{n+1}}}=:\tyF_{n+1}^{\tyl}}}
	{\infer{\wtjudg{\vartwo:\mset{\tyy_{n+1}^{\tyl},\mset{\tyt_n^{\tyl}}}}{\geq(2n+2)\indet}
	 {\vartwo(\Theta\vartwo)}{{\linty_{n+1}}}}{
			\infer{\wtjudg{\vartwo:\mset{\tyy_{n+1}^{\tyl}}}{\geq 
			\indet}{\vartwo}{\tyy_{n+1}^{\tyl}
					\defeq\arr{\mset{\linty_n}}{\linty_{n+1}}}}{}
			& 
			\infer{\wtjudg{\vartwo:\mset{\mset{\tyt_{n}^{\tyl}}}}{\geq(2n+2)\indet}
				{\Theta\vartwo}{\mset{\linty_n}}}{ 
				\infer{\wtjudg{\vartwo:\mset{\tyt_{n}^{\tyl}}}{\geq(2n+1)\indet}
				 {\Theta\vartwo}{\linty_n}}
				{
					\infer{\wtjudg{}{\geq(2n+1)\indet}{\Theta}{\arr{\tyt_{n}^{\tyl}}
							{\linty_n}=:\tyF_n^{\tyl}}}{\ih} & 
					\infer{\wtjudg{\vartwo:\mset{\tyt_{n}^{\tyl}}}{\geq(2n+1)\indet}{\vartwo}{\tyt_{n}^{\tyl}}}{\text{Lemma
						}\ref{lemma:vartree}}
				}}}}
				\]
			\end{proof}
\subsection{The Witness of Space Efficiency}

Let us consider the following families of terms:
\[
\tm_0\defeq\la\var\var\qquad\qquad\tm_{n+1}\defeq(\la\var\var\var)\tm_n
\]
One immediately realizes that $\tm_n$ needs an exponential number of $\beta$ 
steps $\#\beta(\tm_n)$ to reduce to normal form and that the size of $\tm_n$ is 
linear in $n$. More precisely, 
$\#\beta(\tm_n)=2^{n+1}-2$. In order to measure the space consumption on the 
\LIAM, we proceed by giving $\tm_n$ a suitable family of types.
\begin{proposition}
	For each linear type $\linty=\arr{\mset{\lintytwo}}{\lintytwo}$, 
	$\wtjudg{}{2n\indet}{\tm_n}{\linty}$.
\end{proposition}
%

Case $n=0$:
\[
\infer[\tylam]{\tjudgw{}{0}{\la\var\var}{\arr{\mset{\linty}}{\linty}}}{
	\infer[\tyvar]{\tjudg{\var:\mset{\linty}}{\var}{\linty}}{}}
\]
Please notice that in this case we are using the generalized machine that can 
follow also types different from $\initty$, namely 
$\arr{\mset{\linty}}{\linty}$ in this case. This is why the judgment is marked 
with weight $0$: it is a final state.

Case $n+1$:
\[{
	\infer[\tyapp]{\tjudgw{}{2(n+1)\indet}{\tm_{n+1}\defeq(\la\var\var\var)\tm_n}{\linty=:\arr{\mset{\lintytwo}}{\lintytwo}}}{
		\infer[\tylam]{\tjudgw{}{2\indet}{\la\var{\var\var}}{\arr{\mset{\arr{\mset{\linty}}{\linty},\mset{\linty}}}{\linty}}}{
			\infer[\tyapp]{\tjudgw{\var:\mset{\arr{\mset{\linty}}{\linty},\mset{\linty}}}{\indet}{\var\var}{\linty}}{
				\infer[\tyvar]{\tjudgw{\var:\mset{\arr{\mset{\linty}}{\linty}}}{\indet}{\var}{\arr{\mset{\linty}}{\linty}}}{}
				& 	
				\infer[\tymanys]{\tjudgw{\var:\mset{\mset{\linty}}}{\indet}{\var}{\mset{\linty}}}
				{\infer[\tyvar]{\tjudgw{\var:\mset{\linty}}{0}{\var}{\linty}}{}}}}
		&
		\infer[\tymanys]{\tjudgw{}{2(n+1)\indet}{\tm_n}{\mset{\arr{\mset{\linty}}
					{\linty},\mset{\linty}}}}{
			\infer{\tjudgw{}{2n\indet}{\tm_n}{\arr{\mset{\linty}}
					{\linty}}}{\ih} & 
					\infer[\tymanys]{\tjudgw{}{(2n+1)\indet}{\tm_n}{{\mset{\linty}}
					}}{\infer{\tjudgw{}{2n\indet}{\tm_n}{{\linty}}}{\ih}}}}
}\]
Since $\indet=\log(\size{\tm_n})=\log(n)$, we have that the space consumption 
of the \LIAM is $\Theta(n\log n)$, \ie quasi-logarithmic in $\#\beta$.

\end{document}